\documentclass[11pt]{article}

\usepackage{fullpage}

\usepackage{graphicx}

\usepackage{amsmath}
\usepackage{amsthm}
\usepackage{amssymb}
\usepackage{natbib}
\usepackage{mathtools}
\usepackage{hyperref}
\usepackage{bm}
\usepackage{enumitem}

\usepackage[noabbrev]{cleveref}

\usepackage[dvipsnames]{xcolor}

\usepackage{pgfplots}
\usetikzlibrary{patterns}
\usepgfplotslibrary{fillbetween}
\usepackage{xfrac}
\usepackage{framed}
\usepackage[stable]{footmisc}

\usepackage{amsmath,amsfonts,amssymb,amsthm,latexsym}
\usepackage{graphicx}
\usepackage{color}
\usepackage{bm}
\newtheorem{theorem}{Theorem}
\newtheorem{definition}{Definition}
\newtheorem{lemma}{Lemma}
\newtheorem{prop}{Proposition}
\newtheorem{fact}{Fact}
\newtheorem{corollary}{Corollary}

\newtheorem{program}{Linear Program}

\newenvironment{numberedtheorem}[1]{%
\renewcommand{\thetheorem}{#1}%
\begin{theorem}}{\end{theorem}\addtocounter{theorem}{-1}}

\newenvironment{numberedprop}[1]{%
\renewcommand{\theprop}{#1}%
\begin{prop}}{\end{prop}\addtocounter{prop}{-1}}

\newenvironment{numberedlemma}[1]{%
\renewcommand{\thelemma}{#1}%
\begin{lemma}}{\end{lemma}\addtocounter{lemma}{-1}}

\newenvironment{numbereddefinition}[1]{%
\renewcommand{\thedefinition}{#1}%
\begin{definition}}{\end{definition}\addtocounter{definition}{-1}}

\newenvironment{numberedcorollary}[1]{
\renewcommand{\thecorollary}{#1}
\begin{corollary}}{\end{corollary}\addtocounter{corollary}{-1}}

\newcommand{\vv}{{\phi}}

\newcommand{\agind}[1][i]{_{#1}}

\newcommand{\ironed}{\bar}
\newcommand{\constrained}{\hat}
\newcommand{\optconstrained}{\composed{\optimized}{\constrained}}
\newcommand{\optimized}[1]{#1\opt}
\newcommand{\differentiated}[1]{#1'}

\newcommand{\tagged}[2]{{#2}^{#1}}

\newcommand{\primedarg}[1]{#1\primed}
\newcommand{\noaccents}[1]{#1}
\newcommand{\composed}[3]{#1{#2{#3}}}

\newcommand{\newagentvar}[3][\noaccents]{%
\expandafter\newcommand\expandafter{\csname #2\endcsname}{#1{#3}}%
\expandafter\newcommand\expandafter{\csname #2s\endcsname}{#1{\boldsymbol{#3}}}%
\expandafter\newcommand\expandafter{\csname #2smi\endcsname}[1][i]{#1{\boldsymbol{#3}}_{-##1}}%
\expandafter\newcommand\expandafter{\csname #2i\endcsname}[1][i]{#1{#3}\agind[##1]}%
\expandafter\newcommand\expandafter{\csname #2ith\endcsname}[1][i]{#1{#3}_{(##1)}}%
}

\newcommand{\newitemvar}[3][\noaccents]{%
\expandafter\newcommand\expandafter{\csname #2\endcsname}{#1{#3}}%
\expandafter\newcommand\expandafter{\csname #2s\endcsname}{#1{\boldsymbol{#3}}}%
\expandafter\newcommand\expandafter{\csname #2smj\endcsname}[1][j]{#1{\boldsymbol{#3}}_{-##1}}%
\expandafter\newcommand\expandafter{\csname #2j\endcsname}[1][j]{#1{#3}_{##1}}%
\expandafter\newcommand\expandafter{\csname #2jth\endcsname}[1][j]{#1{#3}_{(##1)}}%
}

\newcommand{\maxval}{h}

\newcommand{\forrezs}[1]{{#1}^{\rezs}}
\newagentvar[\forrezs]{imbal}{\phi}
\newagentvar[\forrezs]{cumimbal}{\Phi}

\newagentvar{alloc}{x}
\newagentvar{Dalloc}{\Delta x}

\newagentvar{falloc}{z}

\newagentvar{quant}{q}
\newagentvar[\constrained]{exquant}{\quant}
\newagentvar[\constrained]{critquant}{\quant}  
\newcommand{\monoq}{\quant_{m}}
\newagentvar{Val}{\nu}
\newagentvar{toquant}{Q}

\newagentvar{qprice}{\price}
\newagentvar{qrev}{R}
\newagentvar[\ironed]{iqrev}{\qrev}

\newagentvar[\bar]{bstrat}{\strat}
\newagentvar[\bar]{bpay}{\pay}
\newagentvar[\bar]{bwal}{\wal}
\newagentvar[\bar]{bbid}{\bid}
\newagentvar{zpay}{\pay^{\rezs}}
\newagentvar{zwal}{\wal^{\rezs}}
\newagentvar{zdelt}{\Delta^{\rezs}}
\newagentvar{bomega}{\bm{\omega}}

\newagentvar{qalloc}{y}
\newagentvar{cumalloc}{Y}
\newagentvar[\constrained]{calloc}{\qalloc}
\newagentvar[\constrained]{cumcalloc}{\cumalloc}
\newagentvar[\ironed]{ialloc}{\qalloc}
\newagentvar[\ironed]{icumalloc}{\cumalloc}

\newcommand{\exposted}[1]{#1^{\text{\it EP}}}
\newagentvar[\composed{\exposted}{\constrained}]{excalloc}{\qalloc}
\newagentvar[\exposted]{exalloc}{\qalloc}
\newagentvar[\exposted]{extalloc}{\talloc}
\newagentvar[\exposted]{exfalloc}{\falloc}

\newagentvar{rev}{R}
\newagentvar[\differentiated]{marg}{\rev}
\newagentvar{rawrev}{P}
\newagentvar[\differentiated]{rawmarg}{\rawrev}

\newagentvar[\primedarg]{inducedrev}{\rev}

\newagentvar[\tilde]{pseudorev}{\rev}
\newagentvar[\tilde]{pseudorawrev}{\rawrev}

\newagentvar{typespace}{{\cal T}}
\newagentvar{typesubspace}{S}

\newagentvar{type}{t}
\newagentvar{othertype}{s}
\newagentvar{val}{v}
\newagentvar{outcome}{w}
\newagentvar{outcomespace}{{\cal W}}

\newcommand{\served}[1]{#1^1}
\newcommand{\nonserved}[1]{#1^0}
\newcommand{\alloced}[1]{#1^{\alloc}}
\newcommand{\allocedi}[1]{#1^{\alloci}}
\newagentvar[\alloced]{xoutcome}{\outcome}
\newagentvar[\allocedi]{xioutcome}{\outcome}
\newagentvar[\served]{soutcome}{\outcome}

\newagentvar[\nonserved]{nsoutcome}{\outcome}

\newagentvar{price}{p}
\newagentvar{talloc}{\alloc}
\newagentvar[\tilde]{eval}{\val}
\newagentvar{balloc}{y}
\newagentvar[\tilde]{dalloc}{y}
\newagentvar[\tilde]{eballoc}{y}
\newagentvar{ealloc}{y}
\newagentvar[\tilde]{bprice}{\price}

\newagentvar{act}{a}
\newagentvar{bidspace}{A}
\newagentvar{actspace}{A}

\newagentvar[\constrained]{critval}{\val}
\newagentvar[\constrained]{crittype}{\type}
\newagentvar[\constrained]{critvirt}{\virt}
\newagentvar[\constrained]{reserve}{\val} 
\newagentvar[\constrained]{bidreserve}{\bid} 
\newagentvar[\optconstrained]{monop}{\val}
\newagentvar[\constrained]{monot}{\type}
\newagentvar[\optimized]{monorev}{\rev}

\newagentvar{rcalloc}{y}
\newagentvar[\optimized]{optrcalloc}{\rcalloc}
\newagentvar{biddist}{G}

\newagentvar[\constrained]{critbid}{\bid}
\newagentvar[\constrained]{cbid}{B}
\newagentvar[\primedarg]{wbid}{\bid}
\newagentvar{gfunc}{\vartheta}
\newagentvar{block}{C}

\newagentvar{util}{u}

\newagentvar{strat}{\beta}
\newagentvar{bid}{b}

\newagentvar{pay}{\pi}
\newagentvar{rez}{\rho}
\newagentvar[\tilde]{trez}{\rez}

\newagentvar{virt}{\phi}
\newagentvar{cumvirt}{\Phi}
\newagentvar{qvirt}{\phi}
\newagentvar[\ironed]{ivirt}{\virt}
\newagentvar[\ironed]{icumvirt}{\cumvirt}

\newagentvar[\tagged{\text{SD}}]{sdvirt}{\virt}
\newagentvar[\composed{\tagged{\text{SD}}}{\ironed}]{sdivirt}{\virt}
\newagentvar[\tagged{\text{MD}}]{mdvirt}{\virt}
\newagentvar[\composed{\tagged{\text{MD}}}{\ironed}]{mdivirt}{\virt}

\newagentvar{dist}{F}
\newagentvar{dens}{f}
\newagentvar{hazard}{h}
\newagentvar{cumhazard}{H}

\newagentvar[\ironed]{iprice}{\price}  
\newagentvar[\ironed]{ival}{\val}  

\newagentvar{ints}{{\cal I}}

\newagentvar{wal}{w}
\newagentvar{Util}{U}

\newitemvar{pos}{j}
\newitemvar{weight}{w}
\newitemvar[\differentiated]{mweight}{\weight}
\newitemvar{udtype}{\type}
\newitemvar[\constrained]{udcrittype}{\type}
\newitemvar{udalloc}{\alloc}
\newitemvar{udprice}{\price}
\newitemvar[\constrained]{udcalloc}{\qalloc}
\newitemvar[\constrained]{udcumcalloc}{\cumalloc}

\newagentvar{mech}{{\cal M}}
\newagentvar[\skew{5}{\hat}]{cmech}{{\cal M}}
\newagentvar{alg}{{\cal A}}

\DeclareMathOperator{\OPT}{OPT}

\newcommand{\algo}{A}
\newcommand{\algop}{A}

\newcommand{\mecha}{M}
\newcommand{\mechap}{M}

\newcommand{\reals}{{\mathbb R}}

\newagentvar{trans}{\sigma}

\newagentvar{demandset}{S}

\newcommand{\opt}{^{\star}}
\newcommand{\primed}{^\dagger}

\newagentvar{distout}{w}
\newagentvar{idistout}{\bar{w}}

\DeclareMathOperator{\I}{i-\hspace{-0.1cm}}

\newagentvar{gap}{\delta}

\newcommand{\bott}{\overrightarrow}
\newcommand{\topt}{\overleftarrow}
\newcommand{\botht}{\overleftrightarrow}

\newcommand{\iid}{i.i.d.~}

\newcommand{\valf}{V}
\newcommand{\ratio}{r}

\DeclareMathOperator{\tri}{Trvd}
\DeclareMathOperator{\Cvd}{Cvd}

\newcommand{\valspace}{\mathcal{V}}

\newcommand{\expect}{\mathbf{E}}
\newcommand{\expecta}{\mathbf{E}}
\newcommand{\distp}{f}

\newcommand{\mechaspace}{\mathcal{M}}
\newcommand{\algspace}{\mathcal{A}}
\newcommand{\mes}{s}
\newcommand{\messpace}{\mathcal{S}}
\newcommand{\info}{\mathcal{I}}

\newcommand{\qo}{\bar{\quant}}
\newcommand{\pquant}{p}

\newcommand{\quantf}{Q}

\newcommand{\asimp}{0.806}
\newcommand{\rsimp}{2.447}
\newcommand{\qsimp}{0.093}
\newcommand{\qsimpmhr}{0.426}
\newcommand{\apxsimp}{1.907}
\newcommand{\apxsimpmhr}{1.398}

\newcommand{\optratio}{\ratio^*}

\newcommand{\piratio}{\alpha}

\newcommand{\scF}{\mathcal{F}}

\newcommand{\hazr}{\beta}
\newcommand{\hazf}{\lambda}

\DeclareMathOperator{\Exd}{Exd}
\DeclareMathOperator{\exd}{exd}
\DeclareMathOperator{\Sed}{Sed}
\DeclareMathOperator{\sed}{sed}
\DeclareMathOperator{\Qud}{Qud}
\DeclareMathOperator{\qud}{qud}
\DeclareMathOperator{\Sqd}{Sqd}
\DeclareMathOperator{\sqd}{sqd}
\DeclareMathOperator{\Pmd}{Pmd}
\DeclareMathOperator{\pmd}{pmd}
\DeclareMathOperator{\Eqrd}{Eqrd}

\DeclareMathOperator{\Eqrsd}{Eqrsd}

\DeclareMathOperator{\Ud}{Ud}
\DeclareMathOperator{\ud}{ud}
\DeclareMathOperator{\Upd}{F}
\DeclareMathOperator{\upd}{f}
\DeclareMathOperator{\Dod}{\I\Upd}
\DeclareMathOperator{\dod}{\I\upd}

\newcommand{\cbs}[1]{{(#1)}}

\newcommand{\klookers}{$\mathcal{LA}$}

\newcommand{\RSMINH}{8.56}

\newcommand{\EXBREV}{1.2777}
\newcommand{\EXBRS}{1.00623}
\newcommand{\EXBRSWN}{18}

\newagentvar{blend}{\delta}
\newagentvar{abstrd}{\gamma}

\newcommand{\lowb}{L}

\DeclareMathOperator{\argmax}{argmax}

\newcommand{\given}{\,\mid\,}

\newcommand{\prob}[2][]{\text{\bf Pr}\ifthenelse{\not\equal{}{#1}}{_{#1}}{}\!\left[{\def\givenn{\middle|}#2}\right]}

\newcommand{\tparen}{\big}
\newcommand{\tprob}[2][]{\text{\bf Pr}\ifthenelse{\not\equal{}{#1}}{_{#1}}{}\tparen[{\def\given{\tparen|}#2}\tparen]}
\newcommand{\texpect}[2][]{\text{\bf E}\ifthenelse{\not\equal{}{#1}}{_{#1}}{}\tparen[{\def\given{\tparen|}#2}\tparen]}

\newcommand{\sprob}[2][]{\text{\bf Pr}\ifthenelse{\not\equal{}{#1}}{_{#1}}{}[#2]}
\newcommand{\sexpect}[2][]{\text{\bf E}\ifthenelse{\not\equal{}{#1}}{_{#1}}{}[#2]}

\newcommand{\dd}{{\mathrm d}}

\newif\iffocs
\focstrue

\newif\iffinalprint
\finalprintfalse

\newif\ifbmapp
\bmappfalse
\newif\iftensapp
\tensappfalse
\newif\ifextsapp
\extsappfalse
\newif\ifarxiv
\arxivtrue

\begin{document}



\begin{titlepage}

\title{Lower Bounds for Prior Independent Algorithms}

\newcommand{\email}[1]{\href{mailto:#1}{#1}}

\author{Jason Hartline\thanks{Northwestern U., Evanston IL.  Work done in part while supported by NSF CCF 1618502.\newline Email: \email{hartline@northwestern.edu}} \and Aleck Johnsen\thanks{Northwestern U., Evanston IL.  Work done in part while supported by NSF CCF 1618502.\newline Email: \email{aleckjohnsen@u.northwestern.edu}}}


\maketitle
The prior independent framework for algorithm design considers how
well an algorithm that does not know the distribution of its inputs
approximates the expected performance of the optimal algorithm for
this distribution. This paper gives a method that is agnostic to
problem setting for proving lower bounds on the prior independent
approximation factor of any algorithm. The method constructs a
correlated distribution over inputs that can be generated both as a
distribution over i.i.d.\ good-for-algorithms distributions and as a
distribution over i.i.d.\ bad-for-algorithms distributions. Prior
independent algorithms are upper-bounded by the optimal algorithm for
the latter distribution even when the true distribution is the
former. Thus, the ratio of the expected performances of the Bayesian
optimal algorithms for these two decompositions is a lower bound on
the prior independent approximation ratio. The techniques of the paper
connect prior independent algorithm design, Yao’s Minimax Principle,
and information design. We apply this framework
to give new lower bounds on several canonical prior independent
mechanism design problems.

\end{titlepage}

\section{Introduction}
\label{s:intro}


This paper develops a novel method for establishing lower bounds on
prior independent approximation algorithms.

Stochastic models are enabling theoretical understanding of algorithms
beyond those provided by classical worst-case treatments
\citep[see][]{rou-19}.  These models are especially interesting for
algorithm design problems with information theoretic constraints such
as online algorithms, mechanism design, streaming algorithms, etc.
The Bayesian algorithm design problem can be viewed as a two stage
process.  In the first stage the input is the prior distribution and
an algorithm is constructed for the distribution. In the second stage
the constructed algorithm is run on the realized input.  The Bayesian
optimal algorithm is the one with the highest expected performance.

The prior independent framework evaluates algorithms, which are not
privy to the (first stage) prior distribution of inputs, against a
benchmark defined as the performance of the Bayesian optimal algorithm
that is constructed for this prior.  With no constraints on the prior
distribution, this problem is equivalent to classical worst-case
algorithm design.  Alternatively, prior independent analyses in
mechanism design \citep[e.g.,][]{DRY-15} and online learning
\citep[e.g.,][]{ACF-02} restrict the distributions to be
independent and identically distributed (i.i.d.), respectively
over values of agents in a mechanism and rounds of online inputs.

This paper develops a method for establishing lower bounds on the
performance of prior independent algorithms (for classes of
i.i.d.\ distributions).  The method is based on Yao's Minimax
Principle \citep{yao-77}.  The prior independent framework asks for the designer to
pick one algorithm that is good on an adversary's chosen worst-case
distribution.  Yao's minimax principle allows the order of moves of
the designer and adversary to be swapped.  Thus, the prior independent
optimal approximation ratio can be equivalently identified by an adversary
choosing a distribution over prior distributions and then the designer
choosing a best algorithm.  Note that the class of
i.i.d.\ distributions is not closed under convex combination, thus,
the adversary's distribution over distributions generally gives a
symmetric, correlated distribution over inputs.

The main object of study of this paper is {\em dual blends}, which are pairs of
distinct distributions over i.i.d.\ distributions of inputs that induce the
same correlated distribution.  To establish a prior independent lower
bound, we will be considering dual blends where one side of the dual
blend mixes over good-for-algorithms distributions and the other side
mixes over bad-for-algorithms distributions.  The adversary can choose
the mix over good-for-algorithms distributions in which case the
expectation over Bayesian optimal performances for this mix defines the
benchmark of the prior independent framework.  On the other hand, the
algorithm cannot tell the two blends apart and thus its expected
performance is upper bounded by the expectation over performances of the
Bayesian optimal algorithms for the bad-for-algorithms mix.

As a simple example, consider the mechanism design problem of posting
a price to a single agent with value on $[1,h]$.  (Here the
restriction to i.i.d.\ distributions is trivial as there is only one
agent.)  A class of good-for-algorithms distributions is given by
point masses.  Note that the Bayesian optimal pricing mechanism for a
point mass is to post identically the same price as the value (at which
the agent always buys).  A class of bad-for-algorithms distributions
is given by the equal revenue distribution with cumulative distribution
$F(x) = 1-\sfrac{1}{x}$ and a point mass of $\sfrac{1}{\maxval}$ at $\maxval$.  The equal revenue
distribution has the property that the expected revenue from any
posted price is 1 (the agent buys if her value is at least the price).
Now consider the dual blend where on the good-for-algorithms side we
have the equal revenue distribution over point masses and on the
bad-for-algorithms side we have a point mass on the equal revenue
distribution.  The expected revenue over Bayesian optimal algorithms (in response to point mass distributions)
from the good-for-algorithms side is the expected value of the equal
revenue distribution on $[1,\maxval]$, i.e., $1+ \ln \maxval$.  The expected
revenue from the bad-for-algorithms side is 1.  Thus, we have
established a lower bound of $1+ \ln \maxval$ on the approximation factor of
single-agent posted pricing.  
(In fact, this example analysis is tight due to a matching upper bound from \citealp{HR-14}.)

There are two challenges in establishing lower bounds for
prior independent algorithms via the blends method.  The first
challenge is in sufficiently understanding the Bayesian optimal
algorithm for the class of distributions under consideration.  In
several of the central studied areas of Bayesian algorithms, this
first challenge is solved in closed form.  Bayesian optimal mechanisms
are identified broadly by \citet{mye-81}.  For online learning with
payoffs that are i.i.d.\ across rounds, the Bayesian optimal algorithm
is trivial, it selects the action with the highest expected payoff
(which is the same in each round).  Of course, when closed forms are
not available, bounds on the Bayesian optimal performance can be
employed instead.  An important observation of the method of dual blends is
that not only are Bayesian optimal algorithms used to define the benchmark,
but they can also be used to get non-trivial bounds on any algorithm's prior independent approximation ratio.

The second challenge of the blends method is in identifying
dual blends where the expected Bayesian-optimal performances for
good-for-algorithms and bad-for-algorithms distributions are
significantly separated.  In pursuit of this challenge we give two
general approaches for constructing dual blends for inputs of size
two.  (Many of the challenge problems in prior independent mechanism
design are for inputs of size two, e.g., \citealp{HJL-20}.)  The first
approach is based on the observation that when the density function of
a correlated distribution on inputs of size two can be written as a
separable product of independent functions {\em per order statistic} of the inputs, then it can be
decomposed into two distinct distributions over i.i.d.\ distributions.
The second approach considers one side of the dual blend constructed from
any scaled class of distributions with the other side given by the inverse-distributions of these (for which, as a class, the roles of values and scales are reversed in comparison to the original class).


We apply the blends method to two canonical problems in mechanism
design.  Both are two-agent single-item environments.  One considers
the objective of revenue maximization under a standard regularity
assumption on the distribution.  The other considers the objective of
residual surplus maximization (i.e., maximizing the value of the
winner minus any payments made).  Under the restriction to scale
invariant mechanisms, \citet{HJL-20} identified the prior independent
optimal mechanism for revenue (and its approximation factor of about
$1.907$).  It is unknown whether the restriction to scale-invariant
mechanisms is with loss.  We use the blends method to establish an
unconditional lower bound of $\sfrac{23}{18} \approx \EXBREV$.  For the residual
surplus objective, an upper bound of $\sfrac{4}{3}$ exists as a corollary of \citet{HR-14}.  We establish a lower bound of $\EXBRS$ (no previous
lower bound was known).

There are a number of significant open questions pertaining to lower
bounds for prior independent algorithm design from the method of dual
blends.  First, determine whether there are non-trivial settings where
the method from dual blends is tight.  Second, develop methods for
optimizing the lower bound over classes of dual blends.  Third,
generalize the method beyond two-input models.  On this last point,
while there are important problems in mechanism design with inputs of
size two, other settings would benefit from generalization to larger
inputs, such
as online algorithms.

\paragraph{Related Work}
\label{s:related}


The prior independent model was introduced in mechanism design by
\citet{HR-08} and further refined by \citet{DRY-15}.  At the time it
was conjectured that the second-price auction was the
prior independent revenue-optimal mechanism for selling a single item to one
of two agents with i.i.d.\ values from a regular distribution (\citet{DRY-15} had shown that it guaranteed an upper bound of 2-approximation).
\citet{FILS-15} disproved this conjecture by identifying a mechanism
with an improved upper bound.  \citet{AB-18} -- with an additional restriction to
scale-invariant mechanisms -- proved a weaker version of the conjecture (restriction to monotone hazard rate distributions); and for regular distributions: improved the upper bound and gave the first non-trivial
lower bound for prior independent approximation (by establishing a gap for specific distributions).  \citet{HJL-20} proved the tight result for regular distributions
under the scale-invariance restriction.

\citet{HJL-20} connected the prior independent model from mechanism design with the standard
model for online learning.  Most relevantly in relation to our work on
prior independent lower bounds, they showed that the simple
follow-the-leader algorithm is optimal for expert learning
in prior independent settings (by direct analysis rather than by showing a matching lower bound).


\paragraph{Main Paper Outline}
\label{s:outline}
\Cref{s:setup} gives formal preliminaries of the prior independent
setting and proves lower bounds of the Blends Technique.
\Cref{s:exampleblackwell} gives an explicit example of dual blends and
applies it to two distinct settings within mechanism design to show
novel prior independent lower bounds.  \Cref{s:blendsfrominvdist}
identifies two large classes of blends solutions, each
distinctively motivated as a generalization of the example of
\Cref{s:exampleblackwell}.  \Cref{s:blendsisidd} connects blends to
information design and considers the structure of blends' information
as it relates to Blackwell ordering.  A secondary outline for the appendix sections is included at the beginning of \Cref{a:setup}.

\section{Prior Independent Setup and Lower Bound Technique}
\label{s:setup}
\label{s:mechprobs}

Let $\scF$ be a class of probability distributions with known fixed support $\valspace$ (e.g., $[0,\infty)$).  
In the {\em prior independent algorithm design setting} (PI), there is a distribution $F$ which is known to come from the class $\scF$ and $n$ inputs are drawn \iid from $\dist$ (thus input space is $\valspace^n$).  Critically, the algorithm designer does not know the specific $\dist\in\scF$.  The notation $\dist$ is overloaded to be the cumulative distribution function (CDF), and its probability density function (PDF) is $\distp$.

Fix an algorithm design problem that takes $n$ i.i.d.\ inputs.  Denote a class of feasible algorithms by $\algspace$ and an algorithm in this class by $\algo$ with expected performance $\algo(\vals)$ for inputs $\vals$.  When evaluating the performance in expectation over inputs drawn from a distribution $\dist$, we adopt the notation $\algo(\dist) = \expect_{\vals \sim  \dist}[{\algo(\vals)}]$.  
An algorithm's performance for an unknown distribution $\dist$ is measured against the performance of the optimal algorithm which knows $\dist$.  
With these abstractions, we formally define the Bayesian and prior independent (PI) optimization problems.

\begin{definition}
\label{def:bayesianmechdp}
  The {\em Bayesian optimal algorithm design problem} is given by a distribution $\dist$ and class of algorithms $\algspace$; and solves for the algorithm $\OPT_{\dist}$ with the maximum expected performance:
  \begin{align}
    \tag{$\OPT_{\dist}$}
    \OPT_{\dist} &= \argmax_{\algo \in \algspace} \algo(\dist).
  \end{align}
\end{definition}

\noindent Note that $\OPT_{\dist}$ is an algorithm.  Given a distribution $\dist$, the expected performance of the optimal algorithm is $\OPT_{\dist}(\dist)$ and is the {\em benchmark} that we use for prior independent algorithms:

\begin{definition}
\label{def:pidesign}
  The {\em prior independent algorithm design problem} is given by a class of algorithms $\algspace$ and a class of distributions $\scF$; and searches for the algorithm that minimizes its worst-case approximation:
\begin{align}
  \label{eq:pi}
  \tag{$\piratio^{\scF}$}
\piratio^{\scF} &= \min_{\algo \in \algspace} \left[\max_{\dist \in \scF}
\frac{\OPT_{\dist}(\dist)}{\algo(\dist)}\right]
\end{align}
where the value of the program $\piratio^{\dist}$ is the {\em optimal prior independent approximation factor} for class $\scF$ and class $\algspace$ (which we leave implicit).
\end{definition}

\subsection{Theoretical Lower Bounds from Minimax}
\label{s:yaoandid}
\label{s:pitechnique}
\label{s:adversary}

Yao's Minimax Principle (\Cref{thm:yaominmax}) illustrates the role of the adversary through a direct connection to a 2-player zero-sum game.   First we define additional terms for use in \Cref{thm:yaominmax} and throughout the paper.  Given a space $\Omega$, denote the set of all possible distributions by $\Delta(\Omega)$ -- i.e., the {\em probability simplex}.  Denote a distribution over elements $\omega\in\Omega$ by $\abstrd\in\Delta(\Omega)$.  Given a function $f:\Omega_1 \times \Omega_2 \rightarrow \reals$ where $\Omega_1$ and $\Omega_2$ have arbitrary dimensions, we denote the expectation of $f$ over arguments $\bomega_i\in\Omega_i$ according to $\abstrdi\in\Delta(\Omega_i)$ as $f(\abstrdi,\bomega_{j\neq i})=\mathbf{E}_{\bomega_i\sim\abstrdi}\left[ f(\bomega_i,\bomega_{j})\right]$, e.g., in \Cref{thm:yaominmax}.


\begin{theorem}[\citealp{yao-77}]
\label{thm:yaominmax}
{\em [Yao's Minimax Principle]}  Given a $2$-player zero-sum game $\mathcal{G}$ in which sequentially player $1$ chooses mixed action $\abstrdi[1]\in\Delta(\Omega_1)$, then player $2$ chooses action $\omega_2\in\Omega_2$.  The players are cost minimizers and the cost functions on pure actions are (any real-valued function) $C_1(\omega_1,\omega_2)\geq 0$ and $C_2=-C_1$.  Then the {\em value} of game $\mathcal{G}$ (the left-hand side) satisfies:
\begin{equation}
\label{eqn:yaochain}
    \inf_{\abstrdi[1]\in\Delta(\Omega_1)}~\sup_{\omega_2\in\Omega_2}C_1(\abstrdi[1],\omega_2)\geq \sup_{\abstrdi[2]\in\Delta(\Omega_2)}~\inf_{\omega_1\in\Omega_1}C_1(\omega_1,\abstrdi[2])
\end{equation}
\end{theorem}

\subsection{A Technique for Prior Independent Lower Bounds: Blends}
\label{s:info_bounds}

There is a detailed explanation of the high-level technique of lower bounds from Yao's Minimax Principle in the textbook by \citet{BE-98}.  This section gives a minimax approach that is specific to prior independent design.  To outline, we: (a) fix a randomization over adversary strategies in advance; (b) prove an upper bound on the performance of the best-response algorithm {\em from an alternative description of the adversary's induced correlated distribution over inputs}; and (c) measure the gap between the adversary's expected optimal performance and the upper bound on the expected performance of any algorithm.  The key idea is the correlation in (b):
\begin{definition}
\label{def:blend}
A {\em blend} is a \underline{distribution-over-distributions} $\blend\in\Delta(\scF)$.  (Thus, $\blend(\dist)$ is the density at $\dist$.)  
A {\em blended distribution} $\blend^n\in\Delta(\valspace^n)$ is the induced \underline{density function of the correlated distribution} resulting from $n$ i.i.d.\ draws from a common distribution $\hat{\dist}$, with $\hat{\dist}$ drawn from $\blend$.

Two blends $\blendi[1],~\blendi[2]$ are called {\em dual blends} if there exists correlated density function $g$ such that:
\begin{align*}
\blendi[1]^n(\vals) &= g(\vals) = \blendi[2]^n(\vals)\quad\forall ~\vals
\end{align*}
 \noindent Each of $\blendi[1],~\blendi[2]$ are a {\em side} of the dual blend.  Finally, define $\text{{\em opt}}_{n,i}= \mathbf{E}_{\dist\sim\blendi}\left[\text{{\em OPT}}_{\dist}(\dist) \right]$ to be the expected performance of an optimal algorithm which knows $\dist$ over a blend $\blendi$.
\end{definition}
\noindent The point is: an arbitrary blend $\blend$ can be ``flattened" to describe a specific (symmetric) correlated distribution $\blend^n=g$ over input space $\valspace^n$.  Now suppose in fact two distinct blends $\blendi[1]$ and $\blendi[2]$ as choices of the PI adversary induce the same correlated distribution, i.e., they satisfy \Cref{def:blend}.  Because both induce the same description of input profiles, every algorithm is limited by the structure of either description.  The lower bound of the technique has the following intuition: the adversary chooses $\blendi[2]$ which fixes the benchmark of the current scenario to $\text{opt}_{n,2}=\expecta_{\dist\sim\delta_2}\left[\OPT_{\dist}(\dist)\right]$;\footnote{\label{foot:eorequalsroe} \Cref{lem:mixedbenchmark} in \Cref{a:blendsbound} shows that we can set the prior independent benchmark in this way.}  $\blendi[2]$ induces the correlated distribution $g$ and the algorithm best responds to $g$; however the fact that $\blendi[1]$ also induces $g$ means that every algorithm is upper bounded by $\text{opt}_{n,1}$; if this upper bound is strictly smaller than the benchmark, then a strict gap necessarily ensues.  The proof of \Cref{thm:blendsbound} appears in \Cref{a:blendsbound}.
\begin{theorem}
\label{thm:blendsbound}
Consider a prior independent setting with input space $\valspace^n$, class of algorithms $\algspace$, and class of distributions $\scF$.  Let $\scF^{\text{{\em all}}}$ be all distributions.  Assume there exist two distinct dual blends $\blendi[1]\in\Delta(\scF^{\text{\em all}})$ and $\blendi[2]\in\Delta(\scF)$ and correlated density function $g$ (of \Cref{def:blend}) such that:
\begin{align*}
\blendi[1]^n(\vals) &= g(\vals) = \blendi[2]^n(\vals)\quad\forall ~\vals
\end{align*}

\noindent Then the optimal prior independent approximation factor $\piratio^{\scF}$ is at least the ratio $\sfrac{\text{{\em opt}}_{n,2}}{\text{{\em opt}}_{n,1}}$:
\begin{equation}
    \label{eqn:blendsthmratios}
    \piratio^{\scF} = \min_{\algo\in\algspace} \max_{\dist\in \scF} \frac{\text{{\em OPT}}_{\dist}(\dist)}{\algop(\dist)} \geq \frac{\text{{\em opt}}_{n,2}}{\text{{\em opt}}_{n,1}}
\end{equation}
\end{theorem}
\begin{definition}
\label{def:thetheblendstechnique}
{\em The Blends Technique} is the proof technique for  approximation lower bounds which applies \Cref{thm:blendsbound} to a specified prior independent design problem.
\end{definition}

\noindent A detailed outline of the necessary computations to confirm that descriptions of $\blendi[1]$ and $\blendi[2]$ are dual blends is given in \Cref{s:example}, which also includes a first {\em non-trivial} $n=2$ example of a dual blend.  Construction of dual blends does not depend on problem domain -- e.g., mechanism design or online algorithms -- but which dual blend induces the largest lower bound does depend on domain.  Subsequently in this paper we will (a) give examples of dual blends and use them to prove lower bounds per 
\Cref{def:thetheblendstechnique}, and (b) give general methods for identifying dual blends.

\section{Results in Blends Analysis}
\label{ch:horizons}
\label{s:exampleblackwell}

The first goal of this section is to exhibit a concrete example of dual blends.  The example proceeds in two steps: (1) we describe a relaxed solution that allows {\em infinite weight} which is not directly usable for lower bounds but has simpler algebraic form; and (2), we show that this relaxed solution can be modified to become proper dual blends.

The second part of the section uses the dual blends example to state novel lower bounds for two distinct problems from mechanism design.  Our lengthy introduction to mechanism design and the proofs of these results are deferred to \Cref{a:mechanismdesignsetting}.  Interestingly, the distinct objectives of these two problems results in {\em the two sides of the dual blends playing opposite roles} (as choice of the adversary versus gap-inducing upper bound).   Later in \Cref{s:info_design}, we discuss the implications of this observation in terms of precluding {\em Blackwell ordering} between the two sides of the dual blend.

\subsection{A Concrete Dual Blends Example}
\label{s:tensorexampleoutline}
This section provides an explicit example of dual blends -- with motivation for the chosen distributions from themes in mechanism design.  First, we will describe a blends-type solution that has unbounded input support and infinite total weight (so it is not a probability distribution and it is not possible to re-normalize its weights to become one).\footnote{\label{foot:blendsweights} The elements of a blend $\blend$ are technically densities but we generally refer to them as {\em weights}, i.e., the weight corresponding to a distribution within the mixture over $\scF$ according to $\blend$.  We do this to accommodate a relaxed definition for {\em blend} which allows arbitrary total weight (including infinite).}  Second, we modify the infinite-weight solution to have finite weight in a finite input space (which can be normalized to 1 for any fixed weight).  We provide a solution outline with some confirming calculations deferred to \Cref{a:finiteblendsconfirm}.

For this running dual blends example, the $\blendi[1]$ side will be parameterized by a base class of upward-closed Quadratics (called ``equal revenue" in the mechanism design literature), with PDF given by $\qud_{z}(x) = \sfrac{z}{x^2}$ and CDF given by $\Qud_z(x) = 1-\sfrac{z}{x}$ on $[z,\infty)$.  The $\blendi[2]$ side will be a base class of downward-closed Uniforms, with PDF given by $\ud_{0,z}(x) = \sfrac{1}{z}$ and CDF given by $\Ud_{0,z}(x) = \sfrac{x}{z}$ on $[0,z]$.  (Generally, let $\Ud_{a,b}$ be the Uniform distribution on $[a,b]$.)
%
\paragraph{Infinite-weight Blends}  We start by describing the weights $o_{\dist}$ corresponding to $\blendi[1]$ and weights $\omega_{\dist}$ corresponding to $\blendi[2]$.  Because we first allow the total weight to be infinite, we only require the function $g$ (relaxed to be a ``correlated function" rather than a correlated distribution) to match up its output {\em mass} at every input (cf., density of a correlated distribution).

The weights of the upward-closed Quadratics blend ($\blendi[1]$) are as follows:
\begin{itemize}
\item weights $o_{Qz} = \frac{2}{z}dz$ on all upward-closed distributions $\Qud_{z}$ for $z\in\left(0,\infty\right)$.
\end{itemize}
The weights of the downward-closed Uniforms blend ($\blendi[2]$) are as follows:
\begin{itemize}
\item weights $\omega_{Uz} = \frac{2}{z}dz$ on all downward-closed distributions $\Ud_{0,z}$ for $z\in\left(0,\infty\right)$.
\end{itemize}
Using symmetry, we analyze mass in the cone $\vali[1]\geq\vali[2]\geq0$.  The calculations of total mass at any point $\vals\in(0,\infty)^2$ \underline{are confirmed to be equal} from either dual blends description of the common correlated function $g$.
\begin{align}
\label{eqn:qudversusudinfiniteo}
 \text{result of $\Qud_z$ blend}&=\int_0^{\vali[2]}o_{Qz}\cdot \qud_{z}(\vali[1])\cdot\qud_{z}(\vali[2])
    = \int_0^{\vali[2]} \frac{2}{z} \cdot \frac{z}{\vali[1]^2}\cdot\frac{z}{\vali[2]^2}~dz &&\hspace{-0.2cm}= \frac{1}{\vali[1]^2} = g(\vals)\\
\label{eqn:qudversusudinfiniteomega}
\text{result of $\Ud_{0,z}$ blend}&= \int_{\vali[1]}^{\infty}\omega_{Uz}\cdot \ud_{0,z}(\vali[1])\cdot\ud_{0,z}(\vali[2])
= \int_{\vali[1]}^{\infty} \frac{2}{z} \cdot \frac{1}{z}\cdot\frac{1}{z}~dz &&\hspace{-0.2cm}= \frac{1}{\vali[1]^2} = g(\vals)
\end{align}

\noindent The setup of these calculations is expanded in detail in \Cref{a:example}.  As desired, each side of the dual blends describes exactly the same function $g$ over $\valspace^2$.  The remaining issue to be addressed is that the total weight of all included distributions is divergent: $\int_0^{\infty} \frac{2}{z}dz = \infty$.

\label{page:finiteeqrversusuniform}  
\paragraph{Modification to Finite-weight Blends}  Next we show how to modify the infinite-weight solution above to a proper dual blends solution with approximately the same elements.  Consider input support $\valspace = \left[1,\maxval\right]$ for $1<\maxval<\infty$.  First we define the weights $o_{\dist}$ and $\omega_{\dist}$, largely informed by the infinite-weight solution.  We let the total weight in the system be any constant and can assume that it gets normalized to 1 later.  In fact the total weight will be: $1+\int_1^{\maxval} \frac{2}{z} dz = 1 + 2\ln \maxval$.

The Quadratics have the same general description as the infinite-weight case but are now top-truncated at $\maxval$, with truncated density moved to a point mass at $\maxval$.\footnote{\label{foot:quicktruncationnotation} 
We briefly explain notation of $\topt{\dist}^{\maxval'}$.  Let a left-over-arrow modify the domain-upper-bound of $\dist$ to be $\maxval$.  The accent in $\topt{\dist}^{\maxval'}$ indicates that density above $\maxval$ is {\em truncated} to $\maxval$ as a point mass, i.e., the original CDF jumps to 1 at $\maxval$.}  
Formally, Quadratics have PDF $\topt{\qud}_z^{\maxval'}(x) = \sfrac{z}{x^2}$ on $[1,\maxval)$ and point mass $\topt{\qud}_z^{\maxval'}(\maxval) = \sfrac{1}{\maxval}$, correspondingly CDF $\topt{\Qud}_z^{\maxval'}(x) = 1-\sfrac{z}{x}$ on $[1,\maxval)$ and $\topt{\Qud}_z^{\maxval'}(\maxval) = 1$.

The Uniforms have the same general description as the infinite-weight case but now have domain lower bound at 1 and {\em allow} top-truncation at $\maxval$.  Formally, Uniforms without truncation have PDF $\ud_{1,z}(x) = \sfrac{1}{z-1}$ and CDF $\Ud_{1,z}(x) = \sfrac{x-1}{z-1}$  on $[1,z]$.  Uniforms with truncation have PDF $\topt{\ud}_{1,b}^{\maxval'}(x) = \sfrac{1}{b-1}$ on $[1,\maxval)$ and point mass $\topt{\ud}_{1,b}^{\maxval'}(\maxval) = \sfrac{b-\maxval}{b-1}$, correspondingly $\topt{\Ud}_{1,b}^{\maxval'}(x) = \sfrac{x-1}{b-1}$ on $[1,\maxval)$ and $\topt{\Ud}^{\maxval}(\maxval) = 1$.

The weights of the upward-closed Quadratics blend ($\blendi[1]$) are as follows:\label{page:finiteweightquadsvsunifs}
\begin{itemize}
\item point mass of weight $o_{\text{pm}} = 1$ on (truncated) distribution $\topt{\Qud}^{\maxval'}_1$;
\item weights $o_{Qz} = \frac{2}{z}dz$ on all upward-closed (truncated) distributions $\topt{\Qud}^{\maxval'}_{z}$ for $z\in\left[1,\maxval\right]$.
\end{itemize}

\noindent The weights of the downward-closed Uniforms blend ($\blendi[2]$) are as follows:
\begin{itemize}
\item point mass of weight $\omega_{\text{pm}} = \frac{(2\maxval-1)^2}{
\maxval^2}$ on (truncated) distribution $\topt{\Ud}_{1,2\maxval}^{\maxval'}$;
\item weights $\omega_{Uz} = \frac{2(z-1)^2}{z^3}dz$ on all downward-closed distributions $\Ud_{1,z}$ for $z\in\left[1,\maxval\right]$.
\end{itemize}
\noindent (In fact, we use only one uniform distribution with truncation: $\topt{\Ud}_{1,2\maxval}^{\maxval'}$.)  Calculations to show that these blends result in the same correlated distribution $g$ over $[1,\maxval]^2$ are given in \Cref{a:finiteblendsconfirm}.

\subsection{First Illustrative Results in Mechanism Design}
\label{s:examplemdresults}

We show two prior independent lower bounds in mechanism design from the exact same dual blends solution (using Quadratics-versus-Uniforms with finite weight of \Cref{s:tensorexampleoutline} 
and the Blends Technique of \Cref{def:thetheblendstechnique}).  {\em Revenue} and {\em residual surplus} are two objectives within mechanism design (see \Cref{s:mechintro}).  
\Cref{thm:finitequadsversusunifsrevpiauction} (below, for a revenue objective) uses an adversarial distribution over the Uniforms side of the dual blend.  By contrast, \Cref{thm:finitequadsversusunifsressurppiauction2} (for a residual surplus objective) uses an adversarial distribution over the Quadratics side.  This dichotomy of the respective adversaries' choices highlights how even a single example of dual blends can be distinctly applied to two algorithm settings in order to identify a PI approximation lower bound within each setting.

A fixed prior independent lower bound is stronger if it holds for a smaller class of distributions.  Let $\lowb^{\scF}$ be a lower bound on the optimal approximation factor $\piratio^{\scF}$ for a class $\scF$.  \Cref{fact:disthierarchypropfromblends} makes clear that $\lowb^{\scF}$ holds additionally for a superclass $\mathcal{E}$:

 \begin{fact}
\label{fact:disthierarchypropfromblends}
Given two classes of distributions $\mathcal{E}$ and $\scF$ such that $\mathcal{E}\supset \scF$.  Then $\piratio^{\mathcal{E}}\geq\piratio^{\scF}\geq\lowb^{\scF}$.
\end{fact}
 
\noindent  Thus, we give our results for the smallest classes of distributions in order to state the strongest bounds from our analysis.  Define two sub-classes: Uniforms $\scF^{\text{unif}}[1,\maxval]=\{\topt{\Ud}_{1,b}^{\maxval'}~:~ 1 \leq b \}\equiv \text{uniforms on}~[1,b]~\text{truncated at}~\maxval$; and Quadratics $\scF^{\text{quad}}[1,\maxval]=\{\topt{\Qud}_{a}^{\maxval'}~:~1\leq a \leq \maxval \}\equiv \text{quadratics on}~[a,\maxval]~\text{truncated at}~\maxval$.  We explain the approach for both theorems but full proofs are deferred to \Cref{a:mechdesign}.

\begin{theorem}
\label{thm:finitequadsversusunifsrevpiauction}
Given a single-item, 2-agent, truthful auction setting with a revenue objective and with agent values restricted to the space $[1,\maxval]$ for $\maxval>2$.  For the class of uniform distributions $\scF^{\text{{\em unif}}}$, the optimal prior independent approximation factor of any (truthful) mechanism is lower bounded as:
\begin{equation}
    \label{eqn:finitequadsversusunifsrevpiauction}
    \piratio^{\scF^{\text{\em unif}}}_{\maxval} \geq 
    \frac{\text{{\em opt}}_{2,2}}{\text{{\em opt}}_{2,1}}=
    \frac{\frac{23\maxval}{6}-\frac{7}{2}-\ln(\sfrac{\maxval}{2})}{3\maxval-2}=\lowb_{\maxval}^{\scF^{\text{\em unif}}}
\end{equation}
\noindent The lower bound $\lowb^{\scF^{\text{{\em unif}}}}_{\maxval}\rightarrow 
\sfrac{23}{18}\approx\EXBREV$ as $\maxval\rightarrow\infty$ and this is the supremum of $\lowb^{\scF^{\text{{\em unif}}}}_{\maxval}$ over $\maxval\geq 1$.
\end{theorem}

\noindent The canonical PI revenue maximization problem measures worst-case approximation with respect to the class of {\em regular} distributions $\scF^{\text{reg}}$ (\Cref{def:reg}).  All of our Uniforms are regular:  $\scF^{\text{reg}}\supset \scF^{\text{unif}}$.  As a corollary, we get a lower bound for regular distributions: $\piratio^{\scF^{\text{reg}}}_{\maxval} \geq \lowb_{\maxval}^{\scF^{\text{unif}}}$.


As already stated, the proof of \Cref{thm:finitequadsversusunifsrevpiauction} follows the script of the Blends Technique (\Cref{def:thetheblendstechnique}).  We set $\blendi[2]\in\Delta(\scF^{\text{unif}})$ to be the Uniforms blend with finite weights (page~\pageref{page:finiteeqrversusuniform}) and we set $\blendi[1]\in\Delta(\scF^{\text{all}})$ to be the corresponding Quadratics dual blend.  The Second Price Auction (SPA; \Cref{def:spaauction}) is optimal for all Quadratics in $\scF^{\text{quad}}$; 
the lower bound $\maxval>2$ is necessary so that the SPA is not also optimal for all Uniform distributions with positive weight in $\blendi[2]$ (otherwise there is no gap: $\sfrac{\text{ opt}_{2,2}}{\text{opt}_{2,1}}=1$).  Given these, the right-hand side of equation~\eqref{eqn:finitequadsversusunifsrevpiauction} is simply the result of evaluating $\sfrac{\text{opt}_{2,2}}{\text{opt}_{2,1}}$ (and recalling from \Cref{def:blend} that $\text{opt}_{n,i}= \mathbf{E}_{\dist\sim\blendi}\left[\OPT_{\dist}(\dist) \right]$).

\begin{theorem}
\label{thm:finitequadsversusunifsressurppiauction}
Given a single-item, 2-agent, truthful auction setting with a residual surplus objective and with agent values restricted to the space $[1,\maxval]$ for $\maxval\geq \RSMINH$.  For the class of quadratic distributions $\scF^{\text{{\em quad}}}$, the optimal prior independent approximation factor of any (truthful) mechanism is lower bounded as:
\begin{equation}
    \label{eqn:finitequadsversusunifsressurppiauction}
    \piratio^{\scF^{\text{\em quad}}}_{\maxval} \geq
    \frac{\text{{\em opt}}_{2,2}}{\text{{\em opt}}_{2,1}}> \frac{4\maxval^2-2\maxval-\maxval\ln\maxval-e\ln\maxval-e}{4\maxval^2-3\maxval-\maxval\ln\maxval} = \lowb_{\maxval}^{\scF^{\text{\em quad}}}
\end{equation}
\noindent The lower bound $\lowb^{\scF^{\text{{\em quad}}}}_{\maxval}\rightarrow1$ as $\maxval\rightarrow\infty$.  As an example bound: for $\maxval \in \mathbb{N}$, the maximum of $\lowb^{\scF^{\text{{\em quad}}}}_{\maxval}$ is achieved at $\maxval=\EXBRSWN$ with $\lowb^{\scF^{\text{\em quad}}}_{\EXBRSWN} \approx \EXBRS$.
\end{theorem}

\noindent  The canonical PI residual surplus maximization problem measures worst-case approximation with respect to the class of all distributions $\scF^{\text{all}}$.\footnote{\label{foot:justificationfromHR14} We note the contrast: $\scF^{\text{all}}$ is standard for prior independent design with a residual surplus objective, whereas $\scF^{\text{reg}}$ is standard with a revenue objective.  As partial explanation: for the class $\scF^{\text{all}}$, \citet{HR-14} show that constant-approximation is possible for residual surplus, and also show a super-constant lower bound for revenue.  Revenue maximization restricts to regular distributions which satisfy a natural concavity property, and for which constant-approximation is possible (the first upper bound was from \citet{DRY-15}).}  As a corollary, we get a lower bound for all distributions: $\piratio^{\scF^{\text{all}}}_{\maxval} \geq \lowb_{\maxval}^{\scF^{\text{quad}}}$.

Once again, the proof of \Cref{thm:finitequadsversusunifsressurppiauction} uses the Blends Technique.  This time we set $\blendi[2]\in\Delta(\scF^{\text{quad}})$ to be the Quadratics blend with finite weights and set $\blendi[1]\in\Delta(\scF^{\text{all}})$ to be the corresponding Uniforms.  The Lottery (\Cref{def:klotterymech}) is optimal for all Uniforms in $\scF^{\text{unif}}$; the lower bound $\maxval\geq\RSMINH$ is necessary so that the Lottery is not also optimal for all Quadratics with positive weight in $\blendi[2]$ (otherwise there is no gap).  Note, the right-hand side of equation~\eqref{eqn:finitequadsversusunifsressurppiauction} is a simplified lower bound on the ratio $\sfrac{\text{opt}_{2,2}}{\text{opt}_{2,1}}$ as shown in the statement.

Previously for 2-agent auctions for revenue and {\em unbounded value space}, with the additional restriction to {\em scale-invariant} mechanisms, \citet{AB-18} proved for monotone hazard rate distributions ($\scF^{\text{mhr}}$; \Cref{def:haz}) that the SPA is optimal and gave the optimal approximation $\piratio^{\scF^{\text{\em mhr}}} \approx \apxsimpmhr$ (\Cref{thm:pioptn2revmhr}); and also proved for regular distributions ($\scF^{\text{reg}}$) the first-ever PI lower bound.  \citet{HJL-20} gave the optimal mechanism and approximation $\piratio^{\scF^{\text{\em reg}}} \approx \apxsimp$ (\Cref{thm:pioptn2revenue}).  For residual surplus, there is no previous lower bound.  Our mechanism design results have not been optimized in order to identify best lower bounds from the Blends Technique.

\section{General Dual Blends Solutions: Order-statistic Separability and Inverse-distributions}
\label{s:blendsfrominvdist}

This section describes two broad approaches for infinite-weight dual blends solutions that may be useful for identifying good lower bounds for problems of interest, i.e., within a search over dual blends for the one that yields the best lower bound.

The  first blends structure exists when the common function $g$ can be written as {\em multiplicatively-separable functions per order-statistic of the inputs} (for $n=2$).  The second blends structure generates one side of the dual blend by parameterizing over scales of a fixed, base function $\dist$, and the other side is then automatically generated by parameterizing over scales of the {\em inverse-distribution} of $\dist$.  The example of \Cref{s:tensorexampleoutline} is a special case of both approaches.

For simplicity, we describe these constructions allowing for infinite-weight blends.  Similar methods as used in the example of \Cref{s:tensorexampleoutline} can convert them to proper probability distributions.

\subsection{Blends from Order-statistic Separability}
\label{s:mainbodyblendsfromseparate}

This section introduces {\em order-statistic-separable} functions and subsequently describes a class of dual blends based on these functions.  Fix $n=2$ and our inputs in the cone $\vali[1]\geq \vali[2]\geq0$ in which $\vali[1]$ represents the first (largest) order statistic and $\vali[2]$ the second (smaller) order statistic.
\begin{definition}
\label{def:orderstatsep}
Given $n=2$.  An {\em order-statistic-separable function} (with domain $\valspace^2$) is symmetric across the line $\vali[1]=\vali[2]$ and for inputs subject to $\vali[1]\geq \vali[2]\geq0$, has the form:
\label{def:orderstatseparable}
\begin{equation*}
g(\vals)= 
    g_1(\vali[1])\cdot g_2(\vali[2])
\end{equation*}
for which both $g_1$ and $g_2$ adopt the domain $\valspace$.
\end{definition}
\noindent 
To be clear, the separate functions $g_1$ and $g_2$ are \underline{not} independent factors of $g$ because of the condition $\vali[1]\geq \vali[2]$.  The function $g$ is correlated and is not a product itself.  
Let $G_1(z)=\int_z^{\infty} g_1(y)dy$ and $G_2(z)=\int_0^z g_2(y)dy$ be respectively upward-cumulative and downward-cumulative functions.  
(Intuitively, if $G_1(z)$ is finite, then a ``normalized" function $\sfrac{g_1(x)}{G_1(z)}$ gives the PDF of a conditional probability distribution parameterized by $z$, on domain $[z,\infty)$; and the same is true for finite $G_2(z)$ on domain $(0,z]$.)


Before stating a formal result in \Cref{thm:blendsgeneratorfromseparate} to construct dual blends, we show that the Quadratics-versus-Uniforms example of \Cref{s:tensorexampleoutline} exhibits order-statistic separability.  The blends' correlated density at every point $\vals\in\reals^2_+$ for $\vali[1]\geq\vali[2]$ was calculated in equations~\eqref{eqn:qudversusudinfiniteo} and~\eqref{eqn:qudversusudinfiniteomega} to be $g(\vals) =\sfrac{1}{\vali[1]^2}$.  It is easy to verify that 
$g_1(\vali[1])=\sfrac{1}{\vali[1]^2}$ and $g_2(\vali[2])=1$ satisfy \Cref{def:orderstatsep}.  %
%
%
%
The proof and discussion of \Cref{thm:blendsgeneratorfromseparate} are given in \Cref{a:generalblendresults}.

\begin{theorem}
\label{thm:blendsgeneratorfromseparate}
Consider non-negative functions $g_1(\cdot)$ and $g_2(\cdot)$ each with domain $(0,\infty)$.  For every $z>0$, let $g_{1,z}$ be $g_1$ restricted to the domain $[z,\infty)$ and $g_{2,z}$ be $g_2$ restricted to the domain $(0,z]$.

Each $\blendi$ blend is a distribution over the set $\{g_{i,z}~:~z>0\}$.  Let $o_{g_1}(z)$ and $\omega_{g_2}(z)$ be functions ({\em as free parameters which we may design)} to describe weights corresponding respectively to each $g_{1,z}$ and to each $g_{2,z}$.

First, assume $g_1(\cdot)$ and $g_2(\cdot)$ satisfy the following conditions:
\begin{enumerate}
\item The function $\chi(z) = \frac{g_1(z)}{g_2(z)}$ evaluated in the limit at $\infty$ is $0$, i.e., $\lim_{z\rightarrow \infty}\chi(z) = 0$;
\item the function $\psi(z) = \frac{g_2(z)}{g_1(z)}$ evaluated in the limit at $0$ is $0$, i.e., $\lim_{z\rightarrow 0} \psi(z) = 0$;
\item $\chi(z)$ must be weakly decreasing, equivalently, $\psi(z)$ must be weakly increasing;
\end{enumerate}
Then the weights functions $o_{g_1}(z)=d\psi(z)$ and $\omega_{g_2}(z)=-d\chi(z)$ 
give a dual blends solution with: $$g(\vals) = g_1(v_1)\cdot g_2(v_2)~\text{{\em for}}~\vals= (\vali[1],\vali[2]\leq\vali[1])$$

\noindent If the following condition additionally holds:
\begin{enumerate}
\setcounter{enumi}{3}
\item the integrals $G_1(z)=\int_z^{\infty} g_1(y)~dy$ and $G_2(z)=\int_0^z g_2(y)~dy$ are positive and finite for all $x\in(0,\infty)$;
\end{enumerate}
then for the same function $g$, there exists a dual blends solution (by modification from the original solution) for which all of the $g_{1,z}$ and $g_{2,z}$ functions are distributions.  
\end{theorem}
\noindent The modification for the last part of \Cref{thm:blendsgeneratorfromseparate} is defined by: the distributions making up the blends classes are $\tilde{g}_{1,z}(x) = \sfrac{g_{1,z}(x)}{G_1(z)}$ and $\tilde{g}_{2,z}(x)=\sfrac{g_{2,z}(x)}{G_2(z)}$ and the weights are $\tilde{o}_{g_1}(z)=d\psi(z)\cdot\left(G_1(z) \right)^2$ and $\tilde{\omega}_{g_2}(z)=-d\chi(z)\cdot\left(G_2(z) \right)^2$.

\subsection{Blends from Inverse-distributions}
\label{s:inversedistributions}

It is a remarkable feature of the infinite-weight Quadratics-versus-Uniforms dual blends {\em that both sides use the exact same weights parameters per $z$}, namely $o_{Qz}=\omega_{Uz}=\sfrac{2}{z}\cdot dz$.  
This structure is not an anomaly -- it is indicative of a class of infinite-weight dual blends solutions which we formalize in \Cref{thm:general2overzblends} (and give the key definitions and proof below).

The critical structure is the multiplicative inverse `$\sfrac{1}{z}$.'  
Its importance is highlighted from two perspectives: inverse-distributions and arbitrary distribution rescaling.  Notably, Quadratics and Uniforms are inverse-distributions to each other, which we see directly from $\Qud_1(x) = 1 - \sfrac{1}{x}$ on $[1,\infty)$ for which the inverse-distribution CDF is $1-\Qud_1(\sfrac{1}{x})=1-(1-\sfrac{1}{1/x})=x=\Ud_{0,1}(x)$ on $[0,1]$.  Additionally, the Quadratics blend assigns weights to all rescalings of $\Qud_1$ and the Uniform blend assigns weights to all rescalings of $\Ud_{0,1}$.  Fundamentally, \Cref{thm:general2overzblends} shows that there is a duality between distribution values and distribution scales, as can be observed in equation~\eqref{eqn:f1inputduality}.

\begin{theorem}
\label{thm:general2overzblends}
Given distribution $\Upd$, define members $\Upd_y$ of its parameterized class of all possible rescalings  $y>0$, and its inverse-distribution $\Dod$ by
\begin{equation}
    \label{eqn:f1inputduality}
    \Upd_z(x) = \Upd(\sfrac{x}{z}) = 1-\Dod(\sfrac{z}{x}) = 1 - \Dod_x(z)
\end{equation}
For $n=2$, $\Upd_z$ and $\Dod_z$ give classes that are dual blends using weights $o_z=\omega_z=\sfrac{1}{z}$, i.e., they describe a common function $g$ at every $\vals = (\vali[1],\vali[2]\leq\vali[1])$:
\begin{equation}
\label{eqn:generalblendversusinverse}
\int_0^{\infty} \frac{1}{z}\cdot \upd_z(\vali[1])\cdot\upd_z(\vali[2])~dz =g(\vals) = \int_{0}^{\infty}\frac{1}{z}\cdot \dod_z(\vali[1])\cdot\dod_z(\vali[2])~dz
\end{equation}
\end{theorem}

\begin{definition}
\label{def:inversedistribution}
Given a distribution $\dist$ with domain $[a,b]$ (or domain $[a,\infty)$), i.e., $\dist(a) = 0$ and $\dist(b)=1$.  The {\em inverse-distribution} of $\dist$ is defined by the CDF function $\I \dist(x) = 1-\dist(\sfrac{1}{x})$ on domain $[\sfrac{1}{b},\sfrac{1}{a}]$ (respectively domain $(0,\sfrac{1}{a}]$).  The PDF of the inverse-distribution is denoted $\I\distp$.  (Fact: as an operation, {\em distribution inversion} is its own inverse, i.e., it respects the identity $\I(\I\dist) = \dist$.)
\end{definition}
%
%
%
\begin{fact}
\label{def:distrescaling}
Given a distribution $\dist_{z=1}$ with default scaling parameter $z=1$ and with domain $[a,b]$ (or domain $[a,\infty)$).  The distribution $\dist_1$ can be arbitrarily re-scaled for $z\in(0,\infty)$ to $\dist_z(x) = \dist_1 (\sfrac{x}{z}$) with domain $[z\cdot a,z\cdot b]$ (respectively domain $[z\cdot a,\infty)$).
\end{fact}
\noindent These concepts come together in \Cref{thm:general2overzblends}.  Note that technically, \Cref{thm:general2overzblends} is a special case of \Cref{thm:blendsgeneratorfromseparate}.  However, it proves that an infinite-weight blends solution always exists effectively from fixing symmetric weights $o_z=\omega_z=\sfrac{1}{z}\cdot dz$ and then choosing the $g_1$ and $g_2$ as inverse-distributions of each other.  In comparison, $g_1$ and $g_2$ were (relatively) free parameters in \Cref{thm:blendsgeneratorfromseparate} to be chosen first, for which weights could then be identified to complete a dual blends solution.  We give a concise proof of \Cref{thm:general2overzblends} from the key ideas of this section (inverse-distributions and rescaling):

\begin{proof}
Given distribution $\Upd$ and its inverse-distribution $\Dod$, the rescaled CDFs and PDFs are:
\begin{align}
\nonumber
   \Upd_z(x) &= \Upd(\sfrac{x}{z}) &\quad && \Dod_z(x) & = \Dod(\sfrac{x}{z}) = 1-\Upd(\sfrac{z}{x})\\
   \nonumber
   \upd_z(x) &=  \frac{1}{z}\cdot \upd(\sfrac{x}{z}) &\quad && \dod_z(x) & = \frac{z}{x^2}\cdot \upd(\sfrac{z}{x})
\end{align}
Starting from the right-hand side of equation~\eqref{eqn:generalblendversusinverse}, 
the  following sequence completes the proof:
\begin{align*}
    &\int_{0}^{\infty}\frac{1}{z}\cdot \dod_z(\vali[1])\cdot\dod(\vali[2])~dz= \int_{0}^{\infty}\left[\frac{1}{z}\cdot dz\right]\cdot \left(\frac{z}{\vali[1]^2}\cdot \upd(\sfrac{z}{\vali[1]}) \right)\cdot \left( \frac{z}{\vali[2]^2}\cdot \upd(\sfrac{z}{\vali[2]})\right)
\intertext{(here we perform calculus-change-of-variables using $z = \zeta(y) = \frac{\vali[1]\cdot\vali[2]}{y}$; 
recall that part of the substitution is $dz = \zeta'(y)\cdot dy$, and integral endpoints get mapped by $\zeta^{-1}(\cdot)$)}
= & \int_{\infty}^0\left[\frac{1}{\frac{\vali[1]\cdot\vali[2]}{y}}\cdot \left(\frac{-\vali[1]\cdot\vali[2]}{y^2}\cdot dy  \right) \right]\cdot\left( \frac{\frac{\vali[1]\cdot\vali[2]}{y}}{\vali[1]^2}\cdot \upd(\sfrac{\vali[2]}{y})\right)\cdot\left( \frac{\frac{\vali[1]\cdot\vali[2]}{y}}{\vali[2]^2}\cdot \upd(\sfrac{\vali[1]}{y})\right)\\
= & \int_0^{\infty}\left[\frac{1}{y}\cdot dy\right]\cdot \left(\frac{1}{y}\cdot\upd(\sfrac{\vali[2]}{y}) \right)\cdot \left(\frac{1}{y}\cdot\upd(\sfrac{\vali[1]}{y}) \right) = \int_0^{\infty} \frac{1}{y}\cdot \upd_y(\vali[2])\cdot\upd_y(\vali[1])~dy \qedhere
\end{align*}
\end{proof}
\noindent An interesting property of (infinite-weight) dual blends from \Cref{thm:general2overzblends} that emerges from the proof is: {\em we don't have to solve for a closed-form expression} for the function $g$ in order to prove equality of its dual descriptions.  As a consequence, the process of obtaining lower bounds from dual blends may reduce to computation of expectations over optimal performances $\OPT_{\dist}(\dist)$.

%
%


\section{Blends Design is Information-Design-Design}
\label{s:info_design}
\label{s:blendsisidd}

This section connects {\em theoretical optimization} of the Blends Technique to the economics topic of information design, specifically as a procedure of {\em information-design-design} (IDD).  For a given prior independent problem (parameterized by class of distributions $\scF$), the main idea is to separate into modular problems the search for the optimal dual blend (which yields the largest lower bound of any dual blend).  (1) An ``outer" problem identifies an optimal correlated distribution $g^*\in\mathcal{G} = \{\blend^n~|~\blend\in\Delta(\scF)\}$.  The outer problem searches over: (2) for any exogenous $g\in\mathcal{G}$, an ``inner" problem identifies two blends that induce $g$ -- respectively from $\scF$ and $\scF^{\text{all}}$ -- to maximally separate the ratio of optimal performances given each blend (cf., the Blends Technique).

Effectively, the distributions that compose each blend acts as {\em signals} to which each corresponding optimal algorithm $\OPT_{\dist}$ may respond.  If signals can be designed as outputs of a mapping from underlying inputs as {\em fixed states}, then such signal-response games are called {\em information design}.  (We can design signals in this way for our problems, see \Cref{lem:iolemma} based on Bayes Law.)  We exhibit the separation of problems first and defer the presentation of information design.

Describing the sequence of inequalities below, the first line starts with a prior independent problem and its right-hand side optimizes over lower bounds from the Blends Technique.  This step removes the algorithm design problem of the $\min$-player and gives a new problem (which is constrained with respect to the original, possibly with loss).

Next where an adversary optimizes both steps of a $\sup-\sup$ program, we rearrange these two successive choices to: (a) optimize the correlated distribution $g$ which represents both (flattened) sides of the dual blends simultaneously; and then (b) optimize over sets of blends which induce $g$ to maximize the numerator (using $\scF$) and minimize the denominator (using $\scF^{\text{all}}$).\footnote{\label{foot:gtomultipleblendssolutions} This optimization may be non-trivial -- for a single exogenous $g$, there are generally multiple candidate blends which induce $g$.  
Intuitively, this is true because the set $\{\blend~|~\blend^n=g\}$ is closed under convex combination.  As illustration, first consider two distinct dual blends examples $g^a=\blendi[1]^n=\blendi[2]^n$ and $g^b=\blendi[3]^n=\blendi[4]^n$ as may be generated per the large class of \Cref{thm:general2overzblends}.  Then $g^{ab} = \sfrac{g^a}{2} + \sfrac{g^b}{2}$ has four blends solutions: $\sfrac{\blendi^n}{2}+\sfrac{\blendi[j]^n}{2}$ for all $i\in\{1,2\},~j\in\{3,4\}$.  (We count here the four combinations of ``corner" descriptions of $g^{ab}$.  We ignore that, e.g., the $\sfrac{\blendi^n}{2}$ term may mix over $\sfrac{\blendi[1]^n}{2}$ and $\sfrac{\blendi[2]^n}{2}$ -- an optimization never needs this mix by linearity of expectation.)  To generalize, the convex set $\{\blend~|~\blend^n=g\}$ is generally a Hilbert space, e.g., if $g$ is a continuous mixture over a continuum of dual blends.}   The final line is a reorganization using independence of numerator and denominator which now each comprise a {\em sub-problem of the Blends Technique}.
\begin{align}
\nonumber
    \piratio^{\scF} = \min_{\algo\in\algspace} \max_{\dist\in \scF} \frac{\OPT_{\dist}(\dist)}{\algop(\dist)}
    &\geq \sup_{\blendi[2]\in\Delta(\scF)}\left[\sup_{\blendi[1]\in\{\blend~|~\blend\in\Delta(\scF^{\text{all}})~\text{and}~\blend^n=g=\blendi[2]^n\}}\left[\frac
    { \expecta_{\dist\sim\blendi[2]}\left[\OPT_{\dist}(\dist) \right] }
    { \expecta_{\dist\sim\blendi[1]}\left[\OPT_{\dist}(\dist) \right] }
    \right]\right]\\
    \nonumber
    &= \sup_{g\in\mathcal{G}}~\left[\sup_{\substack{\blendi[2]\in\{\blend~|~\blend\in\Delta(\scF)~\text{and}~\blend^n=g\}\\\blendi[1]\in\{\blend~|~\blend\in\Delta(\scF^{\text{all}})~\text{and}~\blend^n=g\}}}\left[\frac
    { \expecta_{\dist\sim\blendi[2]}\left[\OPT_{\dist}(\dist) \right] }
    { \expecta_{\dist\sim\blendi[1]}\left[\OPT_{\dist}(\dist) \right] }
    \right]\right]\\
\label{eqn:blendsisidd}
    &= \sup_{g\in\mathcal{G}}~\left[\frac
    { \sup_{\blendi[2]\in\{\blend~|~\blend\in\Delta(\scF)~\text{and}~\blend^n=g\}}\left( \expecta_{\dist\sim\blendi[2]}\left[\OPT_{\dist}(\dist) \right]\right) }
    { \inf_{\blendi[1]\in\{\blend~|~\blend\in\Delta(\scF^{\text{all}})~\text{and}~\blend^n=g\}}\left(\expecta_{\dist\sim\blendi[1]}\left[\OPT_{\dist}(\dist) \right] \right)}
    \right]
\end{align}

\begin{definition}
\label{def:blendsisidd}
The optimization problem of equation~\eqref{eqn:blendsisidd} is {\em Information-Design-Design}.  Within the brackets, we refer to the optimizations respectively as the {\em Numerator} and {\em Denominator Games}.
\end{definition}

\noindent Thus, when $g$ is fixed exogenously by an outer maximization, there is a reduction to diametrically-opposite questions of constrained information design (\Cref{prop:blendsisidd} next).  Constraining the design is the key step -- informally information design is a {\em signalling game} and we require that signals be distributions $\dist\in\scF$ (which each induce a product distribution $\dist^n$).  Thus, (a) the marginal distribution over signals is a blend, and (b) an optimal algorithm can be run in response to a given signal $\hat{\dist}$ (cf., the use of distributions-as-signals in $\text{opt}_{n,i}= \mathbf{E}_{\dist\sim\blendi}\left[\OPT_{\dist}(\dist) \right]$).

\begin{prop}
\label{prop:blendsisidd}
Consider the prior independent design problem (\Cref{def:pidesign}) given a class of distributions $\scF$, a class of algorithms $\algspace$, and $n$ inputs.  Optimization of the Blends Technique approach to prior independent lower bounds is described by:
\begin{equation*}
    \piratio^{\scF} \geq \sup_{g\in\mathcal{G}}~\left[\frac
    { \sup_{\blendi[2]\in\{\blend~|~\blend\in\Delta(\scF)~\text{and}~\blend^n=g\}}\left( \expecta_{\dist\sim\blendi[2]}\left[\OPT_{\dist}(\dist) \right]\right) }
    { \inf_{\blendi[1]\in\{\blend~|~\blend\in\Delta(\scF^{\text{all}})~\text{and}~\blend^n=g\}}\left(\expecta_{\dist\sim\blendi[1]}\left[\OPT_{\dist}(\dist) \right] \right)}
    \right]
\end{equation*}
\noindent Further, its Numerator Game and its Denominator Game can be independently instantiated as problems of constrained information design.
\end{prop}

\noindent Most of this \Cref{s:blendsisidd} is deferred to the appendix. \Cref{s:intro_info_design} gives a formal introduction to information design.  \Cref{s:iddreduction} describes the respective reductions of the Numerator and Denominator Games to information design (thereby providing the proof for \Cref{prop:blendsisidd}).

\Cref{s:blackwell} evaluates dual blends from the perspective of {\em Blackwell (partial) ordering}, which compares two designs of signalling strategies equivalently in terms of both a strong measure of their information content, and a strong measure of their usefulness for arbitrary optimization objectives.  In our case, signalling strategies correspond to blends, and the IDD Numerator Game searches for the {\em best} signals using $\scF$ while its Denominator Game searches for the {\em worst} signals using $\scF^{\text{all}}$.  We include here an observation regarding our example of Quadratics-versus-Uniforms dual blends:

\begin{corollary}
\label{thm:blendsnotblackwell}
Finite-weight Quadratics-versus-Uniforms dual blends 
are an example for which there is no relationship according to Blackwell ordering.
\end{corollary}

\noindent If two distinct optimizations prefer expectation over optimal performances from distinct sides of a dual blend, then Blackwell ordering is precluded.  \Cref{thm:blendsnotblackwell} is a consequence of our results in \Cref{s:examplemdresults} whereby \Cref{thm:finitequadsversusunifsrevpiauction} (for revenue) used an adversarial distribution over the Uniforms side of the dual blend, versus, \Cref{thm:finitequadsversusunifsressurppiauction2} (for residual surplus) used an adversarial distribution over the Quadratics side.

\bibliographystyle{apalike}
\bibliography{bib}

\begin{appendix}

\section{Supporting Material for \Cref{s:setup} and \Cref{s:tensorexampleoutline}}
\label{a:mainpostmd}
\label{a:setup}

We start this \Cref{a:setup} with an outline of all appendix sections.  

\begin{enumerate}
    \item[\textbf{A.}] Supporting material for the Blends Technique in \Cref{s:setup} and our main blends example of Quadratics-versus-Uniforms in \Cref{s:tensorexampleoutline}; broadly, the naming scheme and notation for distributions is explained in \Cref{a:distnames}.
    \item[\textbf{B.}] (page~\pageref{page:mechsetting}) Introduction and preliminaries for Mechanism Design, and proofs of our mechanism design results in \Cref{s:examplemdresults}; applications to mechanism design play a more prominent role through the appendix than the main body of the paper.
    \item[\textbf{C.}] (page~\pageref{page:genblendsapp}) Proof, discussion, and corollaries of \Cref{thm:blendsgeneratorfromseparate} in \Cref{s:mainbodyblendsfromseparate} which introduced infinite-weight blends from order-statistic separability.
    \item[\textbf{D.}]  (page~\pageref{page:blendsisiddapp}) Deferred presentation of information-design-design and assessment of Blackwell ordering from \Cref{s:blendsisidd}; includes introductions to information design 
    and Blackwell ordering.
\end{enumerate}

\subsection{Proof of \Cref{thm:blendsbound} in \Cref{s:setup}}
\label{a:blendsbound}

For use in this section, recall our notation $\algo(\dist) = \expect_{\vals \sim  \dist}[{\algo(\vals)}]$ for the expected performance of algorithm $\algo$ on $n$ i.i.d.\ draws from a distribution $\dist$.

First we state and prove \Cref{lem:mixedbenchmark} which shows that for any fixed blend $\bar{\blend}$ (as implicit choice of the adversary), we can obtain a lower bound on prior independent approximation.  (This lower bound is used as an interim step within the proof of \Cref{thm:blendsbound}.)

\Cref{lem:mixedbenchmark} states that we can replace the adversary's maximization problem within prior independent design (for reference see equation~\eqref{eqn:mixedbenchmarklemma}).  In its place, the adversary effectively sets a benchmark as the expectation of optimal performance over distributions drawn from $\bar{\blend}$ (thus, the benchmark is $\mathbf{E}_{\dist\sim\bar{\blend}}\left[\OPT_{\dist}(\dist) \right]$).  Symmetrically, the algorithm's performance is its expected performance over distributions drawn from $\bar{\blend}$ (thus, its performance is $\mathbf{E}_{\dist\sim\bar{\blend}}\left[\algop(\vals)\right]$).

An algorithm's approximation of the benchmark is measured as the ratio of this benchmark to its performance, i.e., as {\em ratio-of-expectations} (ROE).  The lower bound results from the minimum ratio achieved by any algorithm $\algo\in\algspace$.  Practically, this lower bound is only an abstraction because we don't say anything about how to optimize the algorithm $\algo$.


\begin{lemma}[The Ratio-of-Expectations Benchmark Lemma]
\label{lem:mixedbenchmark}
Consider a prior independent setting with input space $\valspace^n$, class of algorithms $\algspace$, and class of distributions $\scF$.  Let $\bar{\blend}\in\Delta(\scF)$ be any fixed blend, i.e., a fixed distribution over the distributions of $\scF$.  Then

\begin{equation}
\label{eqn:mixedbenchmarklemma}
\piratio^{\scF}=\min_{\algo\in\algspace} \max_{\dist\in \scF} \frac{\OPT_{\dist}(\dist)}{\algop(\dist)} \geq \min_{\algo\in\algspace}\left[ \frac{\mathbf{E}_{\dist\sim\bar{\blend}}\left[\OPT_{\dist}(\dist) \right]}{\mathbf{E}_{\dist\sim\bar{\blend}}\left[\algop(\dist)\right]} \right],~~\text{{\em for fixed}}~\bar{\blend}
\end{equation}
\end{lemma}

\begin{proof}
\noindent We start with the prior independent design problem.  Explanations for each step of this sequence are given following.

\begin{align*}
\min_{\algo\in\algspace} \max_{\dist\in \scF} \frac{\OPT_{\dist}(\dist)}{\algop(\dist)} 
&=\min_{\algo\in\algspace} \max_{\blend\in\Delta(\scF)} \frac{\mathbf{E}_{\dist\sim\blend}\left[\OPT_{\dist}(\dist) \right]}{\mathbf{E}_{\dist\sim\blend}\left[\algop(\dist)\right]}\\
&\geq \max_{\abstrd\in\Delta(\Delta(\scF))}\min_{\algo\in\algspace}\mathbf{E}_{\blend\sim \abstrd}\left[ \frac{\mathbf{E}_{\dist\sim\blend}\left[\OPT_{\dist}(\dist) \right]}{\mathbf{E}_{\dist\sim\blend}\left[\algop(\dist)\right]} \right]\\
&= \max_{\blend\in\Delta(\scF)}\min_{\algo\in\algspace}\left[ \frac{\mathbf{E}_{\dist\sim\blend}\left[\OPT_{\dist}(\dist) \right]}{\mathbf{E}_{\dist\sim\blend}\left[\algop(\dist)\right]} \right]\\
&= \min_{\algo\in\algspace}\left[ \frac{\mathbf{E}_{\dist\sim\bar{\blend}}\left[\OPT_{\dist}(\dist) \right]}{\mathbf{E}_{\dist\sim\bar{\blend}}\left[\algop(\dist)\right]} \right],~~\text{for fixed}~\bar{\blend}
\end{align*}

\begin{itemize}
    \item The first line above both relaxes the adersary's action space to allow a mixture of distributions -- i.e., a blend $\blend\in\scF$ -- and changes the benchmark (numerator) to be set by the expected optimal performance over distributions from the blend.
    
    It holds with equality because by \Cref{lemma:roedominance} below, the value of the inner maximization program before-and-after this step is the same for every $\algo$ -- the adversary gains no extra advantage because the ratio on the right-hand side must always be dominated anyway by the ratio achieved by some distribution $\dist_+$ in the support of any chosen $\blend$.  (To explain in further detail, the adversary could choose $\dist_+$ in the left-hand program and can still choose a point mass on $\dist_+$ in the right-hand program.)
    \item The second line applies Yao's Minimax Principle (\Cref{thm:yaominmax}).  Note, the adversary's choice of actions $\abstrd\in\Delta(\Delta(\scF))$ represents the exact transformation using Minimax: the adversary now acts first and plays a distribution over actions in its support from the initial $\min-\max$ side.  Then:
    \item The third line holds because the set of all blends over $\scF$ -- namely, $\Delta(\scF)$ -- is closed under convex combination.
    \item The last line holds because fixing an argument of the outer program can only impair its objective (in this case by fixing $\blend=\bar{\blend}$ for any $\bar{\delta}\in\Delta(\scF)$ per the lemma statement).\qedhere
\end{itemize}
\end{proof}

\noindent \Cref{lemma:roedominance} supports the previous proof.  It states that for a ROE objective like we use above, a point mass on an element of the mixture must achieve at least the value of the overall ratio.  This statement is similar to a standard statement from the probabilistic method -- that there exists a point in the support of a distribution that is at least the expectation.  \Cref{lemma:roedominance} is proved using this standard statement.

\begin{lemma}
\label{lemma:roedominance}
Consider a domain $\Omega$ and two positive functions $a:\Omega\rightarrow\reals_+$ and $b:\Omega\rightarrow\reals_+$.  For every distribution $\abstrd$ over the elements of $\Omega$, there exists $\omega_+$ in the support of $\abstrd$ for which
\begin{equation}
    \label{eqn:roedominance}
    \frac{a(\omega_+)}{b(\omega_+)} \geq \frac{\expecta_{\omega\sim\abstrd}\left[ a(\omega)\right]}{\expecta_{\omega\sim\abstrd}\left[ b(\omega)\right]}
\end{equation}
\end{lemma}
\begin{proof}
Set $\alpha = \expecta_{\omega\sim\abstrd}\left[ a(\omega)\right]$ and $\beta = \expecta_{\omega\sim\abstrd}\left[ b(\omega)\right]$.  The first line uses these definitions and the second line is a simple re-organization:
\begin{align*}
    \frac{\alpha}{\beta} &= \frac{\expecta_{\omega\sim\abstrd}\left[ a(\omega)\right]}{\expecta_{\omega\sim\abstrd}\left[ b(\omega)\right]} \\
    0 &= \expecta_{\omega\sim\abstrd}\left[\beta\cdot a(\omega)-\alpha\cdot b(\omega) \right]
\end{align*}
\noindent Applying the probabilistic method (explained immediately before this lemma) to the last line, there must exist $\omega_+$ for which $\beta\cdot a(\omega_+)-\alpha\cdot b(\omega_+)\geq 0$ which is equivalent to $\sfrac{a(\omega_+)}{b(\omega_+)}\geq \sfrac{\alpha}{\beta}$.
\end{proof}

\noindent With the full proof of \Cref{lem:mixedbenchmark} in place, we are prepared to restate and prove \Cref{thm:blendsbound}.

\begin{numberedtheorem}{\ref{thm:blendsbound}}
\label{thm:blendsbound2}
Consider a prior independent setting with input space $\valspace^n$, class of algorithms $\algspace$, and class of distributions $\scF$.  Let $\scF^{\text{{\em all}}}$ be all distributions.  Assume there exist two distinct dual blends $\blendi[1]\in\Delta(\scF^{\text{\em all}})$ and $\blendi[2]\in\Delta(\scF)$ and correlated density function $g$ (of \Cref{def:blend}) such that:
\begin{align*}
\blendi[1]^n(\vals) &= g(\vals) = \blendi[2]^n(\vals)\quad\forall ~\vals
\end{align*}

\noindent Then the optimal prior independent approximation factor $\piratio^{\scF}$ is at least the ratio $\sfrac{\text{{\em opt}}_{n,2}}{\text{{\em opt}}_{n,1}}$:
\begin{equation*}
    \piratio^{\scF} = \min_{\algo\in\algspace} \max_{\dist\in \scF} \frac{\text{{\em OPT}}_{\dist}(\dist)}{\algop(\dist)} \geq \frac{\text{{\em opt}}_{n,2}}{\text{{\em opt}}_{n,1}}
\end{equation*}
\end{numberedtheorem}
\begin{proof}

\noindent We start with the prior independent design problem and apply \Cref{lem:mixedbenchmark} (given above; by assigning $\bar{\blend} = \blendi[2]$).  Justifications for the next steps are given afterwards.

\begin{align}
\nonumber
\min_{\algo\in\algspace} \max_{\dist\in\scF} \frac{\OPT_\dist(\dist)}{\algop(\dist)} &\geq \min_{\algo\in\algspace}\left[ \frac{\mathbf{E}_{F\sim\blendi[2]}\left[\OPT_\dist(\dist) \right]}{\mathbf{E}_{\dist\sim\blendi[2]}\left[\algop(\dist)\right]} \right]\\
\nonumber
&= \min_{\algo\in\algspace}\left[ \frac{\text{opt}_{n,2}}{\mathbf{E}_{\vals\sim g}\left[\algop(\vals)\right]} \right]\\
\nonumber
&= \min_{\algo\in\algspace}\left[ \frac{\text{opt}_{n,2}}{\mathbf{E}_{\dist\sim\blendi[1]}\left[\algop(\dist)\right]} \right]\\
\label{eqn:theoblendsconclusion}
&\geq \min_{\algo\in\algspace}\left[ \frac{\text{opt}_{n,2}}{\mathbf{E}_{\dist\sim\blendi[1]}\left[\OPT_{\dist}(\dist)\right]} \right] =  \frac{\text{opt}_{n,2}}{\text{opt}_{n,1}} 
\end{align}
\begin{itemize}
    \item The second and third lines substitute using the definition of $\text{opt}_{n,i}$ and the assumption in the theorem statement that $\blendi[1]^n(\vals) = g(\vals) = \blendi[2]^n(\vals)$.
    
    Note, the adversary's choice of $\blendi[2]$ is restricted to the set $\Delta(\scF)$ up front in the prior indepdent problem (i.e., the parameter $\scF$ is fixed exogenously), and $\blendi[2]$ induces $g=\blendi[2]^n$.  However given $g$, there may exist any alternative description $\blendi[1]$ with $g=\blendi[1]^n$, including a $\blendi[1]\in\Delta(\scF^{\text{all}})$ that uses distributions outside the original class $\scF$.  This freedom to design $\blendi[1]$ is an inherent {\em consequence of nature}.
    \item The fourth line inequality recognizes that expectation over locally optimal performances -- each knowing the true $\dist$ when realized -- must weakly dominate the performance of a single algorithm run against all realizations of $\dist$ (formally: \Cref{fact:localoptisworseglobal} after this proof).
    \item The final equality substitutes and realizes that the algorithm no longer appears in the function to be minimized, i.e., the objective is constant.\qedhere
\end{itemize}
\end{proof}

\noindent The following holds because each $\OPT_{\dist}$ algorithm is optimal pointwise per $\dist$, whereas running $\algo$ against each $\dist$ is itself immediately upper bounded by $\OPT_{\dist}$:

\begin{fact}
\label{fact:localoptisworseglobal}
Given an arbitrary prior independent algorithm design setting with class of distributions $\scF$ and class of algorithms $\algspace$, and given $\blend\in\Delta(\scF)$.  For any fixed algorithm $\algo\in\algspace$:
\begin{equation*}
\mathbf{E}_{\dist\sim\blend}\left[\OPT_{\dist}(\dist)\right]\geq \mathbf{E}_{\dist\sim\blend}\left[\algo(\dist)\right]
\end{equation*}
\end{fact}


\subsection{An Alternative Proof of \Cref{thm:blendsbound} from Linear Programming}
\label{a:duals}
\label{s:optpiprogram}
\label{s:lpdesignintro}

We give a second proof of \Cref{thm:blendsbound} for algorithms settings in which it is possible to explicitly model the prior independent problem (\Cref{def:pidesign}) as a linear program, in particular in which the algorithm's performance is a linear combination over variables.  We use a specific example of truthful auctions within mechanism design (see \Cref{s:mechintro} for introduction) but it will be clear where algorithm-specific considerations ``disappear" and we are left with an alternative proof for the Blends Technique.  The techniques and principles of linear programming that we apply here follow from \citet{voh-11}.

To summarize, this section re-proves the Blends Technique using an example problem (a simple auction) in a restricted analytical setting (linear programming).  We identify two prominent structures:
\begin{enumerate}
    \item The Blends Technique describes lower bounds by measuring the prior independent approximation of an ``algorithm" that -- rather than choosing assignments of problem-specific variables -- can directly choose its {\em pseudo-performance} outcome on every input $\vals$ independently of problem-specific constraints, as long as for every distribution $\dist\in\scF^{\text{all}}$, its expected pseudo-performance on inputs drawn from $\dist$ does not exceed the optimal algorithm's performance $\OPT_{\dist}(\dist)$.  This structure is observed in \Cref{prog:optpinonsuperoptlaxsimple} below, which is a relaxation of the initial problem's LP.
    \item \Cref{prog:optpiappendeddual} is the dual program of the primal in the previous point.  Critically, our dual blends (of \Cref{def:blend}) give feasible solutions for this dual program.  The Blends Technique for obtaining lower bounds on prior independent approximation then follows from the inequality between the optimal value of the primal program and the value of the objective of the dual for feasible solutions.
\end{enumerate}

\noindent Regarding specifics of mechanism design: we write a program to describe the prior independent truthful mechanism design problem, for which it is sufficient to use virtual value maximization and characterization of truthful mechanisms (\Cref{thm:myerson}, \citet{mye-81}).  Note that we can write the program once and it applies for each objective using the corresponding virtual value function.  Further, the linear programming approach -- in conjunction with Myerson's characterization -- uses the fact that optimization over truthful mechanisms $\mecha=(\allocs,\prices)$ reduces to optimization over implementable allocations $\allocs$ (cf.\ \Cref{thm:myedsicchar}).  Thus, the arguments of the initial linear program are (monotone) allocations $\allocs$.  Let $\mecha(\dist)$ be the expected performance of mechanism $\mecha$ on $n$ i.i.d.\ draws from $\dist$.

In order to write the problem as a linear program, we define $\check{\piratio}^{\scF} = \sfrac{1}{\piratio^{\scF}}$ to be the multiplicative inverse of our standard approximation factor.  Thus, we may think of $\check{\piratio}^{\scF}\in[0,1]$ as the largest ($\max-\min$) fraction of $\OPT_{\dist}$ that optimal $\mecha^*$ can guarantee in worst-case (i.e., $\mecha^*(\dist) \geq \check{\piratio}^{\scF}\cdot \OPT_{\dist}(\dist)~\forall~\dist\in\scF$).

We need to write a linear program with a single objective.  The technique to ``unravel" the $\max-\min$ formulation (of prior independent design) in order to remove the embedded adversarial-$\min$-objective relies on moving it into a constraint (see the ``approximation" line below) and optimizing an approximation-ratio variable $\check{\piratio}$ as the value of the program.  The optimal factor $\check{\piratio}^{\scF}$ is necessarily at most 1 and we copy this fact into the objective function line.

\begin{program}[The Prior Independent Truthful Mechanism Design Program]
\label{prog:optpi}
Given a class of distributions $\scF$ and any auction objective -- along with its corresponding definition of the virtual value function -- the optimal single-item, $n$-agent truthful mechanism (described by $\allocs^*$) and its optimal approximation factor $\piratio^{\scF}=\sfrac{1}{\check{\piratio}^{\scF}}$ are given by the $\argmax$ of the following program:
\begin{gather}
    \hspace{-4cm}\check{\piratio}^{\scF}=\max_{\allocs,~\check{\piratio}} \check{\piratio}~\leq 1\\
    \hspace{-4cm}\nonumber \text{{\em s.t.}}
\end{gather}
\vspace{-1cm}
\newlength{\optpiprogunit}
\setlength{\optpiprogunit}{-1cm}
\begin{align*}
        \nonumber \int_{\valspace^n} \left(\sum_i \vv_i^F(\vals)\cdot \alloci(\vals) \right) d(\distp^n(\vals))&\geq \check{\piratio}\cdot\OPT_{\dist}(\dist)&&\forall~\dist\in\scF&\hspace{\optpiprogunit}\text{{\em (approximation)}}\\
    \nonumber
    \sum\nolimits_i \alloci(\vals) &\leq 1&&\forall~ \vals\in\valspace^n&\hspace{\optpiprogunit}\text{{\em (single-item feasibility)}}\\
    \nonumber
    \alloci(\vali,\valsmi)&\leq \alloci(\vali',\valsmi) &&\forall~i,~\vali,~\vali'>\vali,~\valsmi&\hspace{\optpiprogunit}\text{{\em (monotonicity)}}\\
    \nonumber
    \alloci(\vals)&\geq 0&&\forall~i,~\vals&\hspace{\optpiprogunit}\text{{\em (non-negativity)}}
\end{align*}
\end{program}

\noindent (From now on, we assume non-negativity without writing it.)  Starting from \Cref{prog:optpi}, we provide a sequence of modifications in order to reprove \Cref{thm:blendsbound} for linear prior independent algorithm design problems.  The goal from here is to obtain a linear program for which we can assign weights of a dual blend to its variables as a feasible solution, and then analysis of an identifiable bound on the objective function implies the desired inequality: $\piratio^{\scF}\geq \sfrac{\text{opt}_{n,2}}{\text{opt}_{n,1}}$.

The key observation for the first modification step is that {\em without loss} we can add to the program a constraint of {\em non-super-optimality}, and not only with respect to $\scF$ but with respect to all distributions (represented by the class $\scF^{\text{all}}$):

\begin{program}[The Appended Program]
\label{prog:optpinonsuperopt}
This program adds a {\em non-super-optimality} constraint to \Cref{prog:optpi} without loss.  We give only the new constraint:
\begin{align*}
        \nonumber \int_{\valspace^n} \left(\sum_i \vv_i^F(\vals)\cdot \alloci(\vals) \right) d(\distp^n(\vals))&\leq \OPT_{\dist}(\dist)&&\forall~\dist\in\scF^{\text{{\em all}}}&\hspace{\optpiprogunit}\text{{\em (non-super-optimality)}}
\end{align*}
\end{program}

\noindent The new constraint is without loss because no prior independent algorithm can do strictly better given $\dist$ than the optimal algorithm $\OPT_{\dist}$ which knows $\dist$ (\Cref{fact:localoptisworseglobal}), and further, this is true regardless of any restrictions imposed on the distribution by the class $\scF$.  The next step is to in fact {\em drop all of the setting-specific constraints} within the linear program, giving us a program whose optimal value $\check{\piratio}^{\scF}_{\text{lax}}$ upper bounds the previous program (i.e., the maximum may now be larger):

\begin{program}[The Appended-Relaxed Program]
\label{prog:optpinonsuperoptlax}
This program relaxes \Cref{prog:optpinonsuperopt} by dropping its mechanism-design-setting-specific constraints.  We are left with:
\begin{gather}
    \hspace{-4cm}\check{\piratio}^{\scF}\leq \check{\piratio}^{\scF}_{\text{{\em lax}}}=\max_{\allocs,~\check{\piratio}} \check{\piratio}~\leq 1\\
    \hspace{-4cm}\nonumber \text{{\em s.t.}}
\end{gather}
\vspace{-1cm}
\begin{align*}
        \nonumber \int_{\valspace^n} \left(\sum_i \vv_i^F(\vals)\cdot \alloci(\vals) \right) d(\distp^n(\vals))&\geq \check{\piratio}\cdot\OPT_{\dist}(\dist)&&\forall~\dist\in\scF&\hspace{\optpiprogunit}\text{{\em (approximation)}}\\
        \nonumber \int_{\valspace^n} \left(\sum_i \vv_i^F(\vals)\cdot \alloci(\vals) \right) d(\distp^n(\vals))&\leq \OPT_{\dist}(\dist)&&\forall~\dist\in\scF^{\text{{\em all}}}&\hspace{\optpiprogunit}\text{{\em (non-super-optimality)}}
\end{align*}
\end{program}

\noindent Of course, the bound $\check{\piratio}^{\scF}\leq \check{\piratio}^{\scF}_{\text{{\em lax}}}$ holds if and only $\piratio^{\scF}=\sfrac{1}{\check{\piratio}^{\scF}}\geq \sfrac{1}{\check{\piratio}^{\scF}_{\text{{\em lax}}}}$, therefore $\sfrac{1}{\check{\piratio}^{\scF}_{\text{{\em lax}}}}$ is a lower bound on the prior independent approximation factor of the original problem.  The next step is to notice that without coordinate-specific constraints on the variables $\allocs$, each parenthetical term may in fact be replaced by a pair of variables $\tilde{A}(\vals)$ and $\tilde{B}(\vals)$, which together represent a measure of algorithm pseudo-performance on input $\vals$ that is locally unconstrained.  (We use $\tilde{A}(\vals)-\tilde{B}(\vals)$ everywhere, effectively as one variable that may be positive or negative.)

The only remaining constraint on the assignment of the variables $\tilde{\mathbf{A}}=\{\tilde{A}(\vals)~:~\vals\in\valspace^n\}$ and $\tilde{\mathbf{B}}=\{\tilde{B}(\vals)~:~\vals\in\valspace^n\}$ is: the expectation of pseudo-performance on any distribution $\dist$ must not exceed the optimal algorithm given $\dist$ (which retains all constraints), i.e., per the non-super-optimality constraint which remains.

We make one more modification to the linear program in this step: we multiply its objective by a positive constant $\kappa$.  For now, we leave $\kappa$ to-be-defined but we will use it later to help short-cut the analysis.  This modification is obviously benign in terms of the $\argmax$.  (Note that if we want to ignore $\kappa$, we set $\kappa=1$ and the objective line here satisfies $\check{\piratio}^{\scF}\leq \check{\piratio}^{\scF}_{\text{{\em lax}}}=\max_{\tilde{\mathbf{A}},~\tilde{\mathbf{B}},~\check{\piratio}} \check{\piratio}~\leq 1$.)

\begin{program}[The Appended-Relaxed-Simplified Program (ARS)]
\label{prog:optpinonsuperoptlaxsimple}
This program simplifies the variable-space of \Cref{prog:optpinonsuperoptlax} without loss by replacing the original allocation variables $\allocs$ with algorithm pseudo-performance variables $\tilde{\mathbf{A}}$ and $\tilde{\mathbf{B}}$, i.e., by substituting $\tilde{A}(\vals)-\tilde{B}(\vals) = \left(\sum_i \vv_i^F(\vals)\cdot \alloci(\vals) \right)$:
\begin{gather}
    \hspace{-4cm}\kappa\cdot\check{\piratio}^{\scF}\leq \kappa\cdot\check{\piratio}^{\scF}_{\text{{\em lax}}}=\max_{\tilde{\mathbf{A}},~\tilde{\mathbf{B}},~\check{\piratio}} \kappa\cdot \check{\piratio}~\leq \kappa\\
    \hspace{-4cm}\nonumber \text{{\em s.t.}}
\end{gather}
\vspace{-1cm}
\begin{align*}
        \nonumber \int_{\valspace^n} \left(\tilde{A}(\vals)-\tilde{B}(\vals) \right) d(\distp^n(\vals))&\geq \check{\piratio}\cdot\OPT_{\dist}(\dist)&&\forall~\dist\in\scF&\hspace{\optpiprogunit}\text{{\em (approximation)}}\\
        \nonumber \int_{\valspace^n} \left(\tilde{A}(\vals)-\tilde{B}(\vals)\right) d(\distp^n(\vals))&\leq \OPT_{\dist}(\dist)&&\forall~\dist\in\scF^{\text{{\em all}}}&\hspace{\optpiprogunit}\text{{\em (non-super-optimality)}}
\end{align*}
\end{program}

\noindent At this point, {\em no structure of the original mechanism design setting remains in \Cref{prog:optpinonsuperoptlaxsimple}} -- thus, any algorithm setting may continue from this point if its prior independent program can drop setting-specific constraints and write pseudo-performance as a single variable (because also: any algorithm setting may add non-super-optimality).

We now convert \Cref{prog:optpinonsuperoptlaxsimple} to its dual program.  (The value of the dual program is at least the value of the primal program and we write this into the objective line.)  Each constraint-line of the dual is assigned an intuitive label to describe its behavior within the program; and the dual has the following variables (one per primal constraint):
\begin{description}
    \item[approximation:] $\omega_{\dist}~\forall~\dist\in\scF$
    \item[non-super-optimality:] $o_{\dist}~\forall~\dist\in\scF^{\text{all}}$
\end{description}

\begin{program}[The Dual of the ARS Program]
\label{prog:optpiappendeddual}
\begin{gather}
    \hspace{-4cm}\kappa\cdot\check{\piratio}^{\scF}_{\text{{\em lax}}}\leq\min_{\boldsymbol{\omega},~\boldsymbol{o}} 
    \int_{\scF^{\text{{\em all}}}}o_{\dist}\cdot \OPT_{\dist}(\dist)~d\dist\\
    \hspace{-4cm}\nonumber \text{{\em s.t.}}
\end{gather}
\vspace{-1cm}
\begin{align*}
        \nonumber
        \int_{\scF} \omega_{\dist}\cdot\OPT_{\dist}(\dist)~d\dist& \geq \kappa && (\text{{\em for}}~\check{\piratio}) & \text{{\em (scale-setting)}}
        \\
        \nonumber 
        \int_{\scF^{\text{\em all}}}\left(o_{\dist}-\omega_{\dist} \right)\cdot \distp^n(\vals)~d\dist &\geq 0&& \forall~\vals\in\valspace^n~(\text{{\em for}}~\tilde{A}(\vals))&\text{{\em (density-matching-A)}}
        \\
        \nonumber 
        \int_{\scF^{\text{\em all}}}\left(-o_{\dist}+\omega_{\dist} \right)\cdot \distp^n(\vals)~d\dist &\geq 0&& \forall~\vals\in\valspace^n~(\text{{\em for}}~\tilde{B}(\vals))&\text{{\em (density-matching-B)}}
\end{align*}
\end{program}

\noindent The final point is to choose (a) dual arguments $\bomega=\{\omega_{\dist}~:~\dist\in\scF\}$ and $\boldsymbol{o}=\{o_{\dist}~:~\dist\in\scF^{\text{all}}\}$ such that {\em these variables describe a finite-weight dual blend} with $\bomega$ the weights for distributions in a blend $\blendi[2]\in\scF$ and $\boldsymbol{o}$ the weights for distributions in a blend $\blendi[1]\in\scF^{\text{all}}$; and (b) choose $\kappa = \int_{\dist}\omega_{\dist}\cdot\OPT_{\dist}(\dist)d\dist$.  Making all of these substitutions into \Cref{prog:optpiappendeddual}, we see that this assignment of dual arguments gives a feasible solution to the dual constraints:
\begin{itemize}
    \item this assignment meets density-matching with equality by definition of a dual blend which is in fact a necessary structure to satisfy both constraints (and further, note that equality is necessary per complimentary slackness wherever we need to allow strictly positive assignment to the corresponding primal variables $\tilde{A}(\vals)$ and $\tilde{B}(\vals)$);
    \item and, it meets scale-setting with equality by choice of $\kappa$ (which makes it is easy to verify).
\end{itemize}
\noindent An assignment to variables that satisfies all constraints gives an upper bound on the optimal value of a minimization LP.  Thus, substituting, re-arranging the objective of \Cref{prog:optpiappendeddual}, and incorporating relationships stated previously gives
\begin{align*}
    \check{\piratio}^{\scF}_{\text{ lax}}&\leq\frac{\int_{\scF^{\text{all}}}o_{\dist}\cdot\OPT_{\dist}(\dist)~d\dist}{\int_{\scF}\omega_{\dist}\cdot\OPT_{\dist}(\dist)~d\dist}=\frac{\text{opt}_{n,1}}{\text{opt}_{n,2}}\\
    \Downarrow&\\
     \frac{\text{opt}_{n,2}}{\text{opt}_{n,1}} & \leq \frac{1}{\check{\piratio}^{\scF}_{\text{lax}}}\leq \frac{1}{\check{\piratio}^{\scF}} = \piratio^{\scF}
\end{align*}
\noindent which finishes the re-proof of \Cref{thm:blendsbound} for linear algorithm settings.

\subsection{Distribution Naming Conventions, Including Exogenous Restrictions}
\label{a:distnames}
\label{s:exorestricteddists}

All distributions will be ``named" functions written in un-italicized lettering, using the following scheme.  Distribution names will:

\begin{itemize}
\item use letters corresponding to the beginning letters of their standard names in the math community (or natural attempts to copy such); and end with the last letter `d' for ``distribution;"
\item use the first letter capitalized to reference the distribution itself as an object and to represent its CDF; and use all letters lowercase to reference the PDF;
\item use a lower-case-i prefix to identify an inverse-distribution (per \Cref{def:inversedistribution});
\item e.g., these should clearly distinguish the exponential function $\exp(x) = e^x = \sum\nolimits_{k=0}^{\infty} \frac{x^k}{x!}$; versus an exponential distribution `$\Exd_1$' with PDF $\exd_1(x) = e^{-x}$.
\item an exception to this naming scheme is {\em local definition and usage of a distribution} 
$\boldsymbol{\xi}$.
\end{itemize}

\noindent  We give further notation to represent operations to modify a given distribution $\dist$ to a related form.  For absolute clarity, we first explicitly explain these standard operations.  {\em Truncation} cuts off a distribution (either at the top or at the bottom) and re-allocates the deleted probability measure of the discarded support to a point mass at the truncation point.  {\em Conditioning} cuts off a distribution and re-normalizes the densities in the remaining domain by dividing by its remaining total probability measure.

Given a distribution $\dist$, we introduce the following formal notation.  Everything that follows applies to a distribution name $\dist$, its CDF $\dist$, and its PDF $\distp$.  Denote a bottom-conditioning and re-normalization of $\dist$ at input $a$ by $\bott{\dist}^a$, top-conditioning and re-normalization at input $b$ by $\topt{\dist}^b$, and both operations at $a$ and $b>a$ respectively by $\botht{\dist}^{a,b}$.\footnote{\label{foot:truncationarrows} The arrows, where present, indicate the deleted density's direction of movement on the real line.  This includes the use of `left-right-arrow' to indicate a both-top-and-bottom domain restriction which ``smushes" the density towards the middle of the original domain.}  
If the distribution $\dist$ instead becomes truncated on one side with a point mass (rather than being conditioned and re-normalized), we accent the end point to indicate the point mass, e.g., $\botht{\dist}^{a,b'}$ represents conditioning above $a$ and top-truncation at $b$.  If the original distribution is described by one or more parameters, e.g. $\dist_z$, these naturally persist as subscripts, for example $\botht{\dist}^{a,b'}_z$.

\subsection{Details of the Blends Technique and a First Example Implemented Directly with Finite Weights: Shifted-Exponentials versus Uniforms}
\label{a:blendsoutline}
\label{s:example}
\label{a:example}

The goal of this section is to illustrate (a) the process of proving a dual blends structure from description of its elements in order to fit into \Cref{thm:blendsbound}, and then (b) the process of obtaining an algorithm-specific lower bound on prior independent approximation (which requires a specific algorithm setting).  In addition to working through the process for (a) in detail in this section, it is fully outlined within \Cref{note:blendsanalysis} as an offset page.  This section includes some steps of the general process that do not apply in the case of our example here.  As previously mentioned -- both this process and the construction of concrete examples of dual blends exist independently of algorithm setting.  For both (a) and (b) we use a dual blends example of Shifted-Exponentials-versus-Uniforms.

Looking ahead, the setting for (b) will be an application of mechanism design (which is introduced in \Cref{s:mechintro}).  Specifically, the setting for (b) will use a 2-agent truthful auction with a revenue objective, which is sufficient description to analyze a revenue gap; only at the very end will we identify relevant classes of distributions for which the revenue gap is meaningful and then formalize the gap with \Cref{prop:blendsEversusUexample}.

\paragraph{Dual Blends.}  We now describe the elements of our example and prove that they describe dual blends.  We assume $n=2$ and start with $\blendi[1]$ and $\blendi[2]$ as follows.  The $\blendi[1]$ blend is a mixture restricted to upward-closed Shifted-Exponentials.  The shifted exponential distribution $\Sed_{z,1}$ has PDF $\sed_{z,1}(x) = e^{-(x-z)}$ and CDF $\Sed_{z,1}(x) = 1-e^{-(x-z)}$ on $[z,\infty)$.

The $\blendi[2]$ blend is a mixture restricted to downward-closed Uniforms.  The uniform distribution $\Ud_{0,z}$ has PDF $\ud_{0,z}(x) = 1/z$ and CDF $\Ud_{0,z}(x) = x/z$ on $[0,z]$.  Note explicitly, each input has support $\valspace = \left[0,\infty\right)$ and input space is $\valspace^2$.

\begin{figure}
\begin{framed}\small
  \begin{minipage}{\textwidth}
\begin{center}
    \textbf{Overview observation:}
\end{center}
{\em Blends} in this paper always result from distributions over: i.i.d.\ draws from a distribution $\dist$.  As such, measurements of density at any fixed input $(\vali[1], \vali[2])$ are always dimensionally-aligned to the obvious axes.  Measurements of density must respect the difference between continuous density $d\vali$ and point masses.  Regarding density at a fixed input $(\vali[1], \vali[2])$, it could consist of continuous density in both dimensions (i.e., $d\vali[1]d\vali[2]$), point mass density in both dimensions, or one dimension of each, e.g., point mass density of $\vali[1]$ multiplied by a continuous density per $d\vali[2]$.
\begin{center}
    \textbf{Steps to confirm $\blendi[1]$ and $\blendi[2]$ as dual blends:}
\end{center}
\begin{description}[leftmargin=0in]\setlength{\parskip}{0in}
    \item [1. Description] Explicitly enumerate the composition of the distributions $\blendi[1]$ and $\blendi[2]$.
    \item [2. Pointwise equality] For all inputs $\vals = (\vali[1],\vali[2])\in\valspace^2$, compute the total density resulting from each blend for each type of measurement of density at the given $\vals$.
    \item [3. Finiteness] (if applicable) Compute the total weight over all inputs for each blend to confirm they are finite; this computation doubles as a sanity check to help confirm that they are equal.
\end{description}
\begin{center}
    \textbf{Identification of Sufficient Integral End Points:}
\end{center}
When a blend contains an integral over distributions parameterized by a bound $z$ on the distributions' respective domains, distributions that contribute 0 at a point $\vals$ may -- without loss -- be excluded by the integral computation of density at $\vals$.  This can be implemented by assignment of the integral end points because by observation, the ignored distributions are described by either an upward or downward-closed set over parameters $z$.  E.g. for Uniforms at input $(4,2)$, ignore $z<4$ because only $\Ud_{0,z}$ with $z\geq 4$ contribute positive density at $(4,2)$; cf. for $z=3$, we have $\ud_{0,3}(4)\cdot\ud_{0,3}(1) = 0\cdot \sfrac{1}{3}~dz=0$.  Symmetric consequences apply for $z$ as a parameter for a distribution's lower bound.  See \Cref{fig:blendsends} for illustration.
\begin{center}
    \textbf{Steps to analyze resulting performance gap:}
\end{center}
\begin{description}[leftmargin=0in]\setlength{\parskip}{0in}
    \item [1. Optimal performances] Compute $\OPT_{\dist}(\dist)$ for every $\dist$ with positive weight in either $\blendi[1]$ or $\blendi[2]$.
    \item [2. Blend performance] Compute $\text{opt}_{2,1}$ and $\text{opt}_{2,2}$ as the measures of blend-weighted expected optimal performance, accounting for both continuous density and point mass blends weights.
    \item[3. Identify lower bound] The ratio of blend performances (arranged to be at least 1) proves a necessary gap between an adversary's choice and the performance of any algorithm, and therefore lower bounds optimal approximation $\piratio^{\scF}$.
    \item [4. Worst-case] (optional/ if applicable) If the analysis is parameterized, analyze worst-case assignment of the parameter -- e.g., for value space $[1,\maxval]$, consider $\maxval\rightarrow\infty$.
\end{description}
  \end{minipage}
  \end{framed}
 \caption{\label{note:blendsanalysis}
 Given class of distributions $\scF$ and $n=2$, this offset provides outlines of required steps to (a) prove that two distributions over distributions $\blendi[1]$ and $\blendi[2]$ are dual blends (of \Cref{def:blend}); and (b) prove a performance gap from \Cref{thm:blendsbound} applied to $\blendi[1]$ and $\blendi[2]$.}
\end{figure}

  The weights $o_F$ correspond to distributions in $\blendi[1]$ and the weights $\omega_F$ correspond to distributions in $\blendi[2]$.  We require here that the total weight in the system is $W=1$, though this could be relaxed for general blend examples:

\begin{fact}
\label{fact:finiteweightok}
The total weight $W$ of a dual blends analysis may be any positive constant as long as the total weight is finite -- any finite weight will divide-out regardless in the right-hand side of line~\eqref{eqn:theoblendsconclusion}).  The total weight on each side of the dual blend must be equal.
\end{fact}

\noindent Further, blends must match up exactly to the technical degree which recognizes the difference between continuous density $d\val$ and point masses.  The blend itself is included as a dimension if puts weight on a continuous mix over a parameter z.

\begin{definition}
    \label{def:dimensionsofdensity}
    Define the {\em count of dimensional density} by the number of (axis-aligned) dimensions $i\in\left\{1,\dots,n\right\}$ in which density is continuous: $d\vali$ or $dz$.
\end{definition}

    \begin{fact}
    \label{fact:blendsdimensionmatch}
    For distributions $\blendi[1]$ and $\blendi[2]$ to be dual blends, it is necessary at every input $\vals$ that they match up density exactly {\em for every type of measurement of density}, in order to account for the difference between continuous density and point masses.
    \end{fact}

\noindent The weights on the upward-closed Shifted-Exponentials blend ($\blendi[1]$) are as follows:\footnote{\label{foot:distnamesassubscripts} Subscripts on blends weights correspond to natural indicators of the distributions represented without respecting their exact naming schemes.  We trust these will be clear from context.  In this case, Shifted-Exponentials $\Sed$ are reduced to $E$ and Uniforms $\Ud$ to $U$.  The $z$ parameter is also present in the subscript.}
\begin{itemize}
\item point mass of weight $o_\text{pm} = \frac{1}{2}$ on the distribution $\Sed_{0,1}$.
\item weights $o_{Ez} = \frac{1}{2}e^{-z}dz$ on all upward-closed distributions $\Sed_{z,1}$ for $z\in\left[0,\infty\right)$.
\end{itemize}

\noindent The weights on the downward-closed Uniforms blend ($\blendi[2]$) are as follows:
\begin{itemize}
\item (explicitly) we don't need a point mass;
\item weights $\omega_{Uz} = \frac{1}{2}z^2e^{-z}dz$ on all downward-closed distributions $\Ud_{0,z}$ for $z\in\left[0,\infty\right)$.
\end{itemize}

\begin{figure}[t]
\begin{flushleft}
\begin{minipage}[t]{0.48\textwidth}
\centering
\begin{tikzpicture}[scale = 0.7,pile/.style={->}]

\fill[gray!20]    (6,6) -- ++(3.4,0) -- ++(0,3.4) -- ++(-3.4,0) -- ++(0,-3.4);

\draw [dashed] (6,6) -- (9,6);
\draw [dashed] (6,6) -- (6,9);

\draw [dashed] (2,2) -- (9,2);
\draw [dashed] (2,2) -- (2,9);

\draw [dashed] (1,1) -- (9,1);
\draw [dashed] (1,1) -- (1,9);

\draw [pile] (0,0) -- (9.5,9.5);

\draw [dotted] (4,4) -- (4,-0.2);
\draw [dotted] (4,4) -- (-0.2,4);

\fill[black] (4,2) circle (0.2cm);

\draw [pile] (-0.2,0) -- (9.5, 0);
\draw [pile] (0, -0.2) -- (0, 9.5);

\draw (0, -0.5) node {$0$};
\draw (9, -0.5) node {$\vali[1]$};
\draw (8, -0.5) node {$\ldots$};
\draw (4, -0.5) node {$4$};
\draw (1, -0.5) node {$1$};
\draw (2, -0.5) node {$2$};
\draw (6, -0.5) node {$6$};

\draw (-0.4, 0) node {$0$};
\draw (-0.4, 1) node {$1$};
\draw (-0.4, 2) node {$2$};
\draw (-0.4, 4) node {$4$};
\draw (-0.4, 6) node {$6$};
\draw (-0.4, 8) node {$\vdots$};
\draw (-0.4, 9) node {$\vali[2]$};

\draw (5, 2.5) node {$(4,2)$};
\draw (8.5, 6.4) node {$z=6$};
\draw (8.5, 2.4) node {$z=2$};
\draw (8.5, 1.4) node {$z=1$};

\draw (5,-1.4) node {Upward-closed density regions $[z,\infty)$};

\end{tikzpicture}
\end{minipage}
\begin{minipage}[t]{0.48\textwidth}
\centering
\begin{tikzpicture}[scale = 0.7,pile/.style={->}]

\fill[gray!20]    (0,0) -- ++(1,0) -- ++(0,1) -- ++(-1,0) -- ++(0,-1);

\draw [dashed] (6,6) -- (-0.2,6);
\draw [dashed] (6,6) -- (6,-0.2);

\draw [dotted] (2,2) -- (9,2);
\draw [dotted] (2,2) -- (2,9);

\draw [dashed] (1,1) -- (-0.2,1);
\draw [dashed] (1,1) -- (1,-0.2);

\draw [pile] (0,0) -- (9.5,9.5);

\draw [dashed] (4,4) -- (4,-0.2);
\draw [dashed] (4,4) -- (-0.2,4);

\fill[black] (4,2) circle (0.2cm);

\draw [pile] (-0.2,0) -- (9.5, 0);
\draw [pile] (0, -0.2) -- (0, 9.5);

\draw (0, -0.5) node {$0$};
\draw (9, -0.5) node {$\vali[1]$};
\draw (8, -0.5) node {$\ldots$};
\draw (4, -0.5) node {$4$};
\draw (1, -0.5) node {$1$};
\draw (2, -0.5) node {$2$};
\draw (6, -0.5) node {$6$};

\draw (-0.4, 0) node {$0$};
\draw (-0.4, 1) node {$1$};
\draw (-0.4, 2) node {$2$};
\draw (-0.4, 4) node {$4$};
\draw (-0.4, 6) node {$6$};
\draw (-0.4, 8) node {$\vdots$};
\draw (-0.4, 9) node {$\vali[2]$};

\draw (5, 2.5) node {$(4,2)$};
\draw (1, 6.4) node {$z=6$};
\draw (1, 4.4) node {$z=4$};
\draw (1, 1.4) node {$z=1$};

\draw (5,-1.4) node {Downward-closed density regions $[0,z]$};

\end{tikzpicture}
\end{minipage}
\end{flushleft}
\caption{
\label{fig:blendsends}
We illustrate identification of integral endpoints for blends calculations when one side (left) is a blend over distributions with upward-closed domain $[z,\infty)$ for all $z$ and the other side (right) is a blend over distributions with downward-closed domain $[0,z]$ for all $z$.  Fix input values, e.g., $(\vali[1],\vali[2])=(4,2)$.  
On each side, the shaded region illustrates a parameter that contributes 0 density (or mass) at the point $(4,2)$.  
Observably, the effective ranges of integration for a blends calculation are respectively $[0,2]$ and $[4,\infty)$.  Generally, they are $[0,\vali[2]]$ and $[\vali[1],\infty)$.
}
\end{figure}
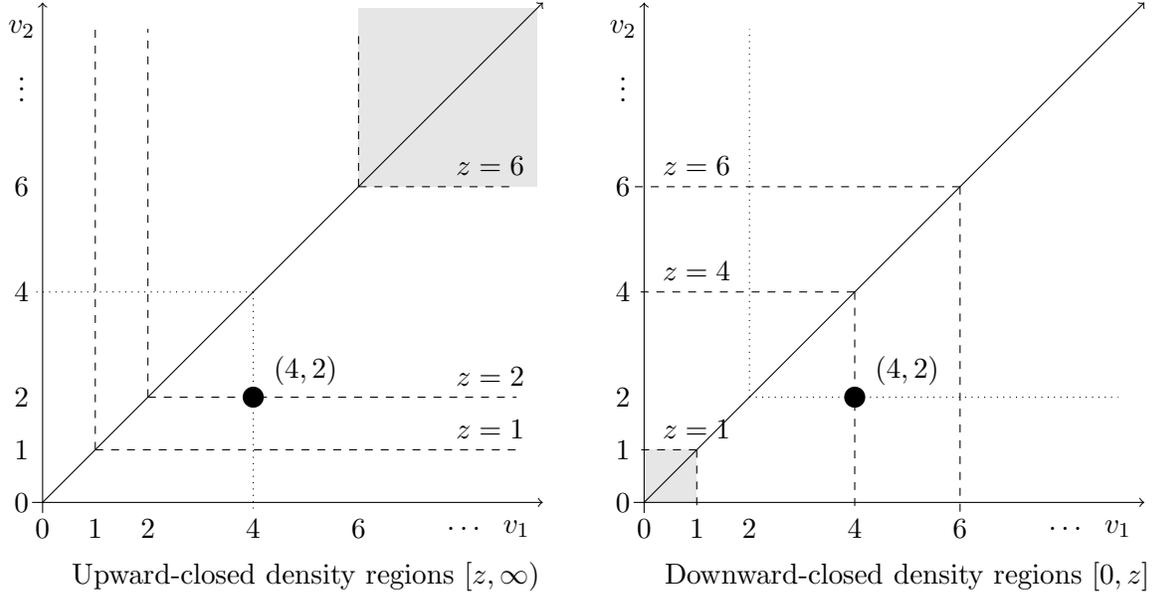

\noindent Here we observe that none of the distributions composing the blends incorporate point mass density.  Therefore the only type of measurement of density that appears in the correlated distribution of this example is of the doubly-continuous form $d\vali[1]d\vali[2]$.  I.e., we only need to check the 2-D density of $g(\vals)$ at each $\vals\in[0,\infty)^2$.  By symmetry we analyze density in the cone $\vali[1]\geq\vali[2]\geq0$.  

Determination of end points of integrals to compute the blends' densities are both (a) described in \Cref{note:blendsanalysis} as part of procedure, and (b) illustrated in \Cref{fig:blendsends}.  Intuitively, we start with an integral over all $z\in[0\infty)$.  However, we truncate the integral end points because not all elements of a blend will put positive density on an input $(vali[1],\vali[2]\leq\vali[1])$.  The calculations of density -- ignoring the continuous density terms $d\vali[1]d\vali[2]$ which are implied by the $2D$ subscript of $g$ -- are given by:

\begin{align}
\label{eqn:expversusunifexline1}
\text{result of $\Sed_{z,1}$ blend} &= \int_0^{\vali[2]}o_{Ez}\cdot \sed_{z,1}(\vali[1])\cdot\sed_{z,1}(\vali[2]) + o_{\text{pm}}\cdot \sed_{0,1}(\vali[1])\cdot\sed_{0,1}(\vali[2])\\
\nonumber
&=\int_0^{\vali[2]} \left(\frac{1}{2}e^{-z}\right) \cdot e^{-(\vali[1]-z)}\cdot e^{-(\vali[2]-z)}~dz +\frac{1}{2}\cdot e^{-\vali[1]}\cdot e^{-\vali[2]} = \frac{1}{2}e^{-\vali[1]} = g_{2D}(\vals)\\
\label{eqn:expversusunifexline2}
\text{result of $\Ud_{0,z}$ blend} &= \int_{\vali[1]}^{\infty}\omega_{Uz}\cdot \ud_{0,z}(\vali[1])\cdot\ud_{0,z}(\vali[2])\\
\nonumber
&= \int_{\vali[1]}^{\infty} \left(\frac{1}{2}z^2e^{-z}\right)\cdot \frac{1}{z}\cdot\frac{1}{z}~dz= \frac{1}{2}e^{-\vali[1]} = g_{2D}(\vals)
\end{align}

\noindent This concludes the blending of Shifted-Exponentials on one side and Uniforms on the other side into the same correlated distribution.  The total weight of the system is 1 from the $\Sed_{z,1}$ side:
\begin{equation*}
o_{\text{pm}}+\int_0^{\infty}o_{Ez} = \frac{1}{2} + \int_0^{\infty} \frac{1}{2}\cdot e^{-z} dz = 1
\end{equation*}

\noindent Total weight is fairly obvious from the Shifted-Exponentials side.  To confirm the total weight from the Uniforms side, we compute the total weight using two iterations of integration-by-parts:
\begin{align*}
\int_0^{\infty}\omega_{Uz} = \int_0^{\infty} \frac{1}{2}z^2e^{-z}dz &= \frac{1}{2}\left(\left[-z^2e^{-z}\right]_0^{\infty}+\int_0^{\infty}2ze^{-z}dz\right)\\
&= \frac{1}{2}\left(0+2\left[\left[-ze^{-z}\right]_0^{\infty}+\int_0^{\infty}e^{-z}dz\right]\right)\\
&= \frac{1}{2}\left(2\left[0+1\right]\right)= 1
\end{align*}

\paragraph{Lower Bound from Revenue Gap.}  Having Shifted-Exponentials-versus-Uniforms as an example of dual blends, we now illustrate how to plug them into \Cref{thm:blendsbound} for a mechanism design setting to obtain a lower bound on prior independent approximation via calculation of $\text{opt}_{2,1}$ and $\text{opt}_{2,2}$.  (Mechanism design is introduced in \Cref{s:mechintro}.)  The setting is a 2-agent truthful auction with a revenue objective and the specific bound that we will obtain is $\piratio^{\scF}\geq 1.1326$.

The first step is to compute optimal revenue for each distribution appearing in either side of the dual blend.  Consider the $\blendi[2]$ side -- i.e., the $\Ud_{0,z}$ side.  The optimal revenue from $n=2$ agents drawn i.i.d.\ from a uniform distribution $\text{U}\left[0,z\right]$ is known to be $\frac{5}{12}z$.  This solution completely covers the distributions used on the $\blendi[2]$ side.  Therefore we have
\begin{equation}
\text{opt}_{2,2} = \int_0^{\infty}\omega_{Uz}\cdot \OPT_{\Ud_{0,z}}(\Ud_{0,z}) =\int_0^{\infty}\left(\frac{1}{2}z^2e^{-z}\right)\cdot \frac{5}{12}z~dz = \frac{1}{2}\cdot\frac{5}{12}\cdot 3! = \frac{5}{4}
\end{equation}

\noindent where the factor of $3!$ results from noting that for positive integers $k$, after repeated integration by parts, $\int_0^{\infty} z^ke^{-z}dz = k!$, and in this case we have $k=3$.

To analyze the $\blendi[1]$ side -- i.e. the $\Sed_{z,1}$ side -- it is sufficient to only look at the virtual value function for arbitrary $\Sed_{z,1}$ because all distributions $\Sed_{z,1}$ can be represented by draws from [$\Sed_{0,1}$ conditioned on $\vali[2]\geq z> 0$].  Later, \Cref{fact:vvconstantunderconditional} will state that given this setup, an observed value $\hat{\val}$ will necessarily {\em have the same virtual value} regardless of which distribution-with-conditioning was used to generate it.  In fact this property holds independently of the hazard rate parameter $\hazr>0$:
\begin{equation}
\vv^{\Sed_{z,\hazr}}(\val) = \val- \frac{1-\left(1-e^{-\hazr(\val-z)}\right)}{\hazr\cdot e^{-\hazr(\val-z)}}=\left(\val-\frac{1}{\hazr}\right)~\forall~z>0,~\val\geq z
\end{equation}

\begin{fact}
\label{fact:expoptrevmech}
For all $z\geq0$ in the $\Sed_{z,\hazr}$ class of distributions including our specific case for $\hazr =1$, there exists a single mechanism which is optimal for them all: a Second Price Auction with a reserve price at $\sfrac{1}{\hazr}$.
\end{fact}

\noindent This simplifies our task to calculate $\text{opt}_{2,1}$ because we can directly add up the revenue of the globally optimal mechanism across the correlated distribution $g(\vals)=\frac{1}{2}e^{-\vali[1]}$.  Recalling $\vali[1]\geq\vali[2]$, we use symmetry across the line $\vali[1]=\vali[2]$ and calculate this as:
\begin{align}
&\text{opt}_{2,1}\\
\nonumber
=&2\left[ \int_0^1 0\cdot d\vali[1] + \int_1^{\infty}\left(\int_0^1\frac{1}{2}e^{-\vali[1]}\text{[reserve price]}~d\vali[2]+\int_1^{\vali[1]}\frac{1}{2}e^{-\vali[1]}\text{[second price]}~d \vali[2]\right)d\vali[1]\right]\\
\nonumber
=&~2\left[ \int_1^{\infty}\left(\int_0^1\frac{1}{2}e^{-\vali[1]}\cdot 1\cdot d\vali[2]+\int_1^{\vali[1]}\frac{1}{2}e^{-\vali[1]}\cdot \vali[2]\cdot d \vali[2]\right)d\vali[1]\right]\\
\nonumber
=&~\int_1^{\infty}e^{-\vali[1]}\left(1+\frac{1}{2}\left(\vali[1]^2-1\right)\right)d\vali[1]\\
\nonumber
=&~\left[\frac{-1}{2}e^{-\vali[1]}\left(\vali[1]^2+2\vali[1]+3\right)\right]_1^{\infty}=\left[\left(\frac{1}{2}\cdot\frac{1}{e}\cdot6\right)-0\right]=\frac{3}{e}\approx 1.1036
\end{align}

\paragraph{Discussion} Having solved for $\text{opt}_{2,1}$ and $\text{opt}_{2,2}$, \Cref{thm:blendsbound} gives us a necessary revenue gap and therefore a lower bound on approximation.  The gap is\footnote{\label{foot:blendexamplesanity} Dual blends analyses generally can depend on difficult and technically tricky computation.  Here we exhibit a quick sanity check on the revenue of the Shifted-Exponentials side.  As stated above, the optimal algorithm is the SPA with a reserve price of 1.  Note then that the SPA with no reserve is sub-optimal, therefore its revenue lower bounds the optimal revenue:

The easiest way to calculate the revenue of the SPA is to note that for each uniform distribution $\Ud_{0,z}$, the SPA gets $\frac{1}{3}z$.  Compare this to the optimal auction per uniform distribution, which got $\frac{5}{12}z$.  From this we see that the SPA simply gets $\sfrac{4}{5}$ths of the revenue of the optimal mechanisms for each distribution on the Uniforms side.  Because $\text{opt}_{2,2}=\sfrac{5}{4}$, it follows that the SPA gets 1.  We confirm that this non-optimal revenue is less than the optimal revenue for the Shifted-Exponentials side which was $\text{opt}_{2,1} \approx 1.1036$.}
\begin{equation}
\label{eqn:s2exampleblendsbound}
\frac{\text{opt}_{2,2}}{\text{opt}_{2,1}}=\frac{\sfrac{5}{4}}{\sfrac{3}{e}} = \frac{5e}{12} \approx 1.1326...
\end{equation}
\noindent and reflects the following intuition.  The adversary commits to an action up front, specifically $\blendi[2]$ the Uniforms distribution.  Interpreting this adversary-moves-first choice through Yao's Minimax Principle and our Benchmark \Cref{lem:mixedbenchmark}, this choice sets the benchmark (in the numerator) to $\text{opt}_{2,2}$.  The designer knows the adversary's strategy and best responds.  However $\blendi[2]^2 = g = \blendi[1]^2$ which shows that even when the designer knows that $g$ was generated by $\blendi[2]$, the designer can not do better than best responding directly to $g$, and further, can not do better than by making a particular assumption that $g$ was generated by the Shifted-Exponentials blend $\blendi[1]$. In fact, critically by \Cref{fact:localoptisworseglobal}, the performance of any mechanism is upper bounded by $\text{opt}_{2,1}$.  Our intermediate conclusion is that our descriptions of Shifted-Expontials and Uniforms as a dual blend result in a revenue gap of any mechanism of at least $1.1326$.

As previously mentioned, we were in fact able to proceed with analysis to this point without even specifying a class of allowable distributions $\scF$.  %
%
In fact, the lower bound holds for {\em any} class of distributions that includes all uniform distributions $\Ud_{0,b}$.  Let $\scF^{\supseteq\text{UNIFORM}} = \left\{ \scF~|~\text{all}~\Ud_{0,b}\in\scF \right\}$ be the (meta)-set of classes of distributions that contain all uniform distributions (as a subset).

\begin{prop}
\label{prop:blendsEversusUexample}
Given a single-item, 2-agent, truthful auction setting with a revenue objective and with agent values in space $[0,\infty)$.  For every class of distributions $\scF\in\scF^{\supseteq\text{{\em UNIFORM}}}$, the optimal prior independent approximation factor of any (truthful) mechanism is lower bounded as:
\begin{equation*}
    \piratio^{\scF}\geq 1.1326
\end{equation*}
\end{prop}

\noindent As a corollary, the bound of equation~\eqref{eqn:s2exampleblendsbound} holds for two classes commonly used within mechanism design -- it holds for both the regular $\scF^{\text{reg}}$ (\Cref{def:reg}) and monotone hazard rate $\scF^{\text{mhr}}$ (\Cref{def:haz}) distributions settings,\footnote{\label{foot:s2exampleversusAB} For MHR and regular settings, our bound here is generally only illustrative -- for $\valspace = [0,\infty)$, and with a restriction to {\em scale-invariant mechanisms} which is conjectured to be without loss, \citet{AB-18} show a tight bound for $\scF^{\text{mhr}}$ of $\piratio^{\scF^{\text{mhr}}}\approx 1.398$ and \citet{HJL-20} show a tight bound for $\scF^{\text{reg}}$ of $\piratio^{\scF^{\text{reg}}}\approx 1.907$.} and this {\em depends only on} the adversary's choice of The Uniforms side of the dual blend because all distributions $\Ud_{a,b}$ are MHR and also regular (for the revenue objective).  
I.e., we have both that $\piratio^{\scF^{\text{reg}}}\geq \sfrac{\text{opt}_{2,2}}{\text{opt}_{2,1}} =1.1326$ and -- as a stronger result because the MHR class is smaller (per \Cref{fact:disthierarchypropfromblends})-- that $\piratio^{\scF^{\text{mhr}}}\geq \sfrac{\text{opt}_{2,2}}{\text{opt}_{2,1}} =1.1326$.

There is an important observation here (as alluded within the proof of \Cref{thm:blendsbound} regarding the implicit relaxation to $\blendi[1]\in\Delta(\scF^{
\text{all}})$).\label{page:blendsboundfromnature}  The Shifted-Exponentials comprising $\blendi[1]$ are also both MHR and regular, {\em but this doesn't matter} -- the upper bound on revenue of any mechanism as results from the Shifted-Exponentials blend-description (of the common correlated distribution) is a consequence of nature itself as follows solely from the adversary's choice of the Uniforms blend, and the structure of that ``consequence" (namely, the revenue-upper-bound structure of the Shifted-Exponentials) faces no constraints at all.  The adversary chooses its blend subject to a particular $\scF$ but the other side of the dual blend can be composed of any subset of distributions in $\scF^{\text{all}}$ (and examples exist for which this is the case).


\subsection{Quadratics-versus-Uniforms: Finite Weight Calculations}
\label{a:finiteblendsconfirm}

For convenience, we copy the descriptions of the distributions used and the weights of the blends from the main body of the paper (page~\pageref{page:finiteweightquadsvsunifs}) and then give the density calculations below.

Quadratics have PDF $\topt{\qud}_z^{\maxval'}(x) = \sfrac{z}{x^2}$ on $[1,\maxval)$ and point mass $\topt{\qud}_z^{\maxval'}(\maxval) = \sfrac{1}{\maxval}$, correspondingly CDF $\topt{\Qud}_z^{\maxval'}(x) = 1-\sfrac{z}{x}$ on $[1,\maxval)$ and $\topt{\Qud}_z^{\maxval'}(\maxval) = 1$.

Uniforms without truncation have PDF $\ud_{1,z}(x) = \sfrac{1}{z-1}$ and CDF $\Ud_{1,z}(x) = \sfrac{x-1}{z-1}$ on $[1,z]$.  Uniforms with truncation have PDF $\topt{\ud}_{1,b}^{\maxval'}(x) = \sfrac{1}{b-1}$ on $[1,\maxval)$ and point mass $\topt{\ud}_{1,b}^{\maxval'}(\maxval) = \sfrac{b-\maxval}{b-1}$, correspondingly $\topt{\Ud}_{1,b}^{\maxval'}(x) = \sfrac{x-1}{b-1}$ on $[1,\maxval)$ and $\topt{\Ud}^{\maxval}(\maxval) = 1$.

The weights on the upward-closed Quadratics blend ($\blendi[1]$) are as follows:
\begin{itemize}
\item point mass of weight $o_{\text{pm}} = 1$ on (truncated) distribution $\topt{\Qud}^{\maxval'}_1$;
\item weights $o_{Qz} = \frac{2}{z}dz$ on all upward-closed (truncated) distributions $\topt{\Qud}^{\maxval'}_{z}$ for $z\in\left[1,\maxval\right]$.
\end{itemize}

\noindent The weights on the downward-closed Uniforms blend ($\blendi[2]$) are as follows:
\begin{itemize}
\item point mass of weight $\omega_{\text{pm}} = \frac{(2\maxval-1)^2}{
\maxval^2}$ on (truncated) distribution $\topt{\Ud}_{1,2\maxval}^{\maxval'}$;
\item weights $\omega_{Uz} = \frac{2(z-1)^2}{z^3}dz$ on all downward-closed distributions $\Ud_{1,z}$ for $z\in\left[1,\maxval\right]$.
\end{itemize}
\noindent With the introduction of pure point masses into underlying distributions, recall that dual blends must match up for every dimension count.  For convenience we re-state \Cref{def:dimensionsofdensity2}.  Then we calculate and confirm all (un-normalized) densities from both sides.

    \begin{numbereddefinition}{\ref{def:dimensionsofdensity}}
    \label{def:dimensionsofdensity2}
    Define the {\em count of dimensional density} by the number of (axis-aligned) dimensions $i\in\left\{1,\dots,n\right\}$ in which density is continuous: $d\vali$ or $dz$.
    \end{numbereddefinition}
\begingroup
\allowdisplaybreaks
\begin{align}
g_{2D}(\vals) &=\int_1^{\vali[2]}o_{Qz}\cdot \topt{\qud}_{z}^{\maxval'}(\vali[1])\cdot\topt{\qud}_{z}^{\maxval'}(\vali[2])+o_{\text{pm}}\cdot\topt{\qud}_1^{\maxval'}(\vali[1])\cdot\topt{\qud}_1^{\maxval'}(\vali[2])\\
\nonumber
&= \int_1^{\vali[2]} \frac{2}{z} \cdot \frac{z}{\vali[1]^2}\cdot\frac{z}{\vali[2]^2}~dz +1\cdot\frac{1}{\vali[1]^2}\cdot\frac{1}{\vali[2]^2} = \frac{1}{\vali[1]^2}\\
g_{2D}(\vals) &= \int_{\vali[1]}^{\maxval}\omega_{Uz}\cdot \ud_{1,z}(\vali[1])\cdot\ud_{1,z}(\vali[2])+\omega_{\text{pm}}\cdot \left(\topt{\Ud}_{1,2\maxval}^{\maxval'}(\vali[1])\right)\cdot\left(\topt{\Ud}_{1,2\maxval}^{\maxval'}(\vali[2])\right)\\
\nonumber
&= \int_{\vali[1]}^{\maxval} \frac{2(z-1)^2}{z^3}\cdot \frac{1}{(z-1)^2}~dz + \frac{(2\maxval-1)^2}{\maxval^2}\cdot\left(\frac{1}{(2\maxval-1)}\right)^2= \frac{1}{\vali[1]^2}\\
g_{0D}(\maxval,\maxval) &= \int_1^{\vali[2]=\maxval}o_{Qz}\cdot \topt{\qud}_z^{\maxval'}(\maxval)\cdot \topt{\qud}_z^{\maxval'}(\maxval)+o_{\text{pm}}\cdot \topt{\qud}_1^{\maxval'}(\maxval)\cdot \topt{\qud}_1^{\maxval'}(\maxval)\\
\nonumber
&= \int_1^{\maxval} \frac{2}{z}\cdot \frac{z}{\maxval}\cdot\frac{z}{\maxval}~dz + 1\cdot\frac{1}{\maxval}\cdot\frac{1}{\maxval} = 1\\
g_{0D}(\maxval,\maxval) &= \int_{\vali[1]=\maxval}^{\maxval} \omega_{Uz}\cdot \ud_{1,z}(\maxval)\cdot\ud_{1,z}(\maxval)+\omega_{\text{pm}}\cdot\left(\topt{\Ud}_{1,2\maxval}^{\maxval'}(\maxval) \right)\cdot\left(\topt{\Ud}_{1,2\maxval}^{\maxval'}(\maxval)\right)\\
\nonumber
&= \int_{\maxval}^{\maxval} \frac{2(z-1)^2}{z^3} \cdot \frac{1}{z-1}\cdot \frac{1}{z-1}~dz + \frac{(2\maxval-1)^2}{\maxval^2}\cdot\left(\frac{\maxval}{(2\maxval-1)}\right)^2 = 1\\
g_{1D}(\maxval,\vali[2]) &=\int_1^{\vali[2]} o_{Qz}\cdot \topt{\qud}_z^{\maxval'}(\maxval)\cdot \topt{\qud}_z^{\maxval'}(\vali[2])+o_{\text{pm}}\cdot \topt{\qud}_1^{\maxval'}(\maxval)\cdot \topt{\qud}_1^{\maxval'}(\vali[2])\\
\nonumber
&= \int_1^{\vali[2]} \frac{2}{z} \cdot \frac{z}{\maxval}\cdot\frac{z}{\vali[2]^2}~dz + \frac{1}{\maxval}\cdot\frac{1}{\vali[2]^2} = \frac{1}{\maxval}\\
g_{1D}(\maxval,\vali[2]) &=\int_{\vali[1]=\maxval}^{\maxval} \omega_{Uz}\cdot \ud_{1,z}(\maxval)\cdot \ud_{1,z}(\vali[2])+\omega_{\text{pm}}\cdot\left(\topt{\Ud}_{1,2\maxval}^{\maxval'}(\maxval)\right)\cdot\left(\topt{\Ud}_{1,2\maxval}^{\maxval'}(\vali[2]) \right)\\
\nonumber
&= \int_{\maxval}^{\maxval} \frac{2(z-1)^2}{z^3}\cdot \frac{1}{z-1} \cdot \frac{1}{z-1}~dz + \frac{(2\maxval-1)^2}{\maxval^2}\cdot\left(\frac{\maxval}{2\maxval-1}\right)\cdot\left(\frac{1}{2\maxval-1}\right) = \frac{1}{\maxval}
\end{align}
\endgroup
\noindent As desired, each side of the dual blends yields the same function $g = (g_{2D},~g_{0D},~g_{1D})$.

\paragraph{Observable Structure of Dual Blends} 

Having completed two finite-weight blends solutions (Shifted-Exponentials-versus-Uniforms in \Cref{a:example} and now Quadratics-versus-Uniforms here in \Cref{a:finiteblendsconfirm}), we identify the following structure which is frequenly observed in example dual blends and general methods (\Cref{s:blendsfrominvdist}) of this paper.

Illustrated here on input support $\valspace=(0,\infty)$, the observed structure is:  the input size is $n=2$; distributions composing $\blendi[1]$ are {\em upward-closed} and are parameterized by $z\in(0,\infty)$ with domain $[z,\infty)$; and distribution in $\blendi[2]$ are {\em downward-closed} and are parameterized by $z\in(0,\infty)$ with domain $(0,z]$.  In further detail:
\begin{itemize}
    \item $\blendi[1]$ is a distribution with weights $o_z$ over realized values of a single distributional parameter for a given upward-closed distribution; e.g., $\blendi[1]$ was a distribution over domain-lower-bounds $z$ of the Shifted-Exponentials in \Cref{a:example};
    \item $\blendi[2]$ is a distribution with weights $\omega_z$ over realized values of a single distributional parameter for a given downward-closed distribution; e.g., $\blendi[2]$ was a distribution over domain-upper-bounds $z$ of the Uniforms in \Cref{a:example}.
\end{itemize}

\noindent We conjecture that no dual blends exist for $n>2$ from our Blends Technique.  For the intuition of this conjecture, see our discussion of ``algebraic consequences of the integral end points" in \Cref{a:discussblendsfromseparate}.  We believe that this upward-closed/downward-closed dual structure is an important property that deserves further study.

\section{Application: Mechanism Design Preliminaries and Proofs}
\label{a:mechanismdesignsetting}
\label{s:setting}
\label{a:setting}
\label{a:mechdesign}
\label{page:mechsetting}

This section gives a formal introduction to mechanism design as the highlighted application of our general results in the paper, which otherwise do not depend on algorithm setting.  The ultimate goal of this section is to present the supporting work and proofs for the mechanism design first given results in \Cref{s:examplemdresults}.

\subsection{Mechanism Design Basics}
\label{s:mechintro}

This sufficient -- albeit lengthy -- section is included for completeness.  Readers who are familiar with the basics of mechanism design may skip it.  However we strive to provide full support for references as we prove our technical mechanism design results (\Cref{thm:finitequadsversusunifsrevpiauction} and \Cref{thm:finitequadsversusunifsressurppiauction}).

We consider mechanism design as it relates to {\em auctions}, i.e., an algorithmic setting of requesting bids from strategic agents, and subsequently allocating items to the agents and charging them monetary payments.  The canonical auction consists of maximizing revenue (i.e., agent payments) by selling one item to one of $n$ agents (possibly randomly) who each have a private value for the item drawn i.i.d.\ from a common Bayesian probability distribution, with the distribution known by the auction designer (i.e, \Cref{def:bayesianmechdp} applied to this setting).  The optimal auction to maximize revenue (or other simple objectives, we define common objectives later) in this setting was solved by \citet{mye-81}.
\footnote{\label{foot:myenobel} For this and related work, Roger Myerson was awarded the Nobel Prize in Economics in 2006 for ``Mechanism Design," jointly with Leonid Hurwicz and Eric Maskin.}

Each agent $i\in\{1,\ldots,n\}$ has value
$\vali$ in a range of known support, e.g., $\vali\in\valspace=\left[0,\infty\right)$ or $\vali\in\valspace=\left[1,\maxval\right]$ for which $\valspace$ is one agent's value space.  Values are private to the agent and are not known by the mechanism.  A profile of $n$ agent values is denoted $\vals = (\vali[1],\ldots,\vali[n])$; the profile with agent $i$'s value replaced with $z$ is $(z,\valsmi) = (\vali[1],\ldots,\vali[i-1],z,\vali[i+1],\ldots,\vali[n])$.  The list of agent values {\em in decreasing order} is $\vali[(1)],\ldots,\vali[(n)]$.\footnote{\label{foot:orderedvals} Re-arranging agents to be labeled in order is typically without loss of generality.  For where it is helpful, we further abstractly define $\vali[\cbs{n+1}]=0$ to be a default, ``sentinel" value.}

A {\em mechanism} collects reports from each agent as {\em bids} and maps them to (possibly randomized) allocations and payments.  A {\em truthful mechanism} is a special case which takes values $\vals$ as input rather than arbitrary bids (and must be designed to incentivize agents to report their values truthfully, see Myerson's characterization below in \Cref{thm:myedsicchar}).

Specifically, a {\em stochastic social choice function} $\allocs$ and a {\em payment function} $\prices$ map a profile of
values $\vals$ respectively to a profile of allocation probabilities, and a profile of expected payments.  Thus, a truthful mechanism is denoted $\mecha=(\allocs^{\mecha},\prices^{\mecha})$.  Where the mechanism is clear from context, we will use the simpler notation $\mecha=(\allocs,\prices)$.  We may also overload notation and write a mechanism's expected performance as a function $\mecha:\vals^n\rightarrow\reals$.

For allocation probability $\alloci$ and expected payment $\pricei$, the agent's expected utility is linear as $\vali\,\alloci - \pricei$ and agents maximize utility in expectation.  We give the most common objectives for mechanism design as a formal definition (for convenience of external reference):

\begin{definition}
\label{def:mechdesignobjectives}
The most common objectives for mechanism design are:
\begin{itemize}
    \item {\em Revenue} is the sum-total over agent payments: $\sum\nolimits_i \pricei$.
    \item {\em Residual surplus} is the sum-total over agent utilities: $\sum\nolimits_i \alloci\vali-\pricei$.
    \item {\em Total welfare} is the sum-total over agent expected-value-of-allocation: $\sum\nolimits_i \alloci\vali$; note that this total respects: $[\text{{\em revenue}} + \text{{\em residual surplus}} = \text{{\em total welfare}}]$.
\end{itemize}
\end{definition}

\noindent With definitions to follow, we restrict attention to mechanisms that are {\em feasible}, {\em dominant strategy incentive compatible} (DSIC/ truthful), and {\em individually rational} (IR), properties which become formal {\em constraints} for mechanism design.  The feasibility constraint for single-item mechanisms requires that for all inputs $\vals$, the profile of expected allocations across all agents sums to at most 1.  The following DSIC and IR constraints must hold for all agents $i$, values $\vali$, and other agent values $\valsmi$.  The DSIC constraint requires that it is always optimal for an agent $i$ to ``bid" value true $\vali$.  In this sense, DSIC mechanisms are {\em truthful}.  The IR constraint requires that an agent $i$ always gets non-negative utility by truthfully bidding $\vali$.  The rest of this \Cref{s:mechintro} presents pertinent structures from the mechanism design literature in order to support main results of this paper which appear in later subsections.

\subsubsection{Characterization of Truthful Equilibrium}
\label{s:myersonchar}

The following
theorem of \citet{mye-81} characterizes social choice functions $\allocs$ that can be implemented by truthful (DSIC) mechanisms, in the context of Nash equilibrium.

\begin{theorem}[\citealp{mye-81}]
\label{thm:myedsicchar}
Allocation and payment rules $(\allocs,\prices)$ are induced by a
dominant strategy incentive compatible mechanism if and only if for
each agent $i$,
\begin{enumerate}
\item (monotonicity) 
\label{thmpart:monotone}
allocation rule $\alloci(\vali,\valsmi)$ is monotone non-decreasing in
$\vali$, and
\item 
\label{thmpart:payment}
(payment identity) payment rule $\pricei(\vals)$ satisfies
\begin{align}
\label{eq:payment-identity}
\pricei(\vals) &= \vali\, \alloci(\vals) - \int_0^{\vali}
\alloci(z,\valsmi)\, \dd z + \pricei(0,\valsmi),
\end{align}
\end{enumerate}
where the payment of an agent with value zero is often
zero, i.e., $\pricei(0,\valsmi) = 0$.
\end{theorem} 

\noindent Unless stated specifically otherwise in this paper, we do fix $\pricei(0,\valsmi) = 0$.

\subsubsection{Standard Mechanisms}
\label{a:standardmechs}
\label{s:spa}

This section describes a number of common auction structures.  The {\em Second Price Auction} (SPA) is a special case of the VCG Mechanism 
which has a number of nice properties: it is a truthful auction, it naturally optimizes total welfare, and it also optimizes revenue when used in conjunction with a correct ``reserve price" (which is a minimum price that any agent must pay to be allocated).

\begin{definition}
\label{def:reservep}
A {\em reserve price} is a minimum price for allocation regardless of any other considerations, e.g. auction parameters or the realized values of other agents.
\end{definition}

\begin{definition}
\label{def:spaauction}
The single-item {\em Second Price Auction (SPA)} with $n$ agents allocates the item to an agent with largest value $\vali[(1)]$ at a price equal to the second-largest value $\vali[(2)]$.
\end{definition}

\noindent The SPA is in fact an example of a {\em $k$-lookahead auction} (\cite{ron-01}) which defines an important class of auctions restricted to those that only ever allocate to the $k$ largest bidders (after ordering and breaking value-ties uniformly at random).  The SPA is a 1-lookahead auction.

\begin{definition}[\cite{ron-01}]
\label{def:klook}
The class of single-item {\em $k$-lookahead mechanisms ($k$-\klookers)} with $n\geq k$ agents is defined by restriction to mechanisms that only ever give positive allocation to the $k$ agents with largest values $\vali[(1)],\ldots,\vali[(k)]$.

Note, the allocations to large-valued agents may condition on the realized values of the un-allocated, small-valued agents with values $\vali[(k+1)],\ldots\vali[(n)]$.
\end{definition}

\noindent A {\em markup mechanism} is a special case of $1$-lookahead that commits to a markup scalar $\ratio\geq 1$ in advance and offers the price $\ratio\cdot\vali[(2)]$ to the largest-valued agent.  The SPA is the edge-case markup mechanism with $\ratio=1$.

\begin{definition}
  \label{def:markup}
  The {\em $\ratio$-markup mechanism} $\mecha_{\ratio}$ offers the price $\ratio\cdot\vali[(2)]$ to the agent with the largest value $\vali[(1)]$.  A {\em randomized markup mechanism} $\mecha_{\hat{\ratio},\boldsymbol{\xi}}$ draws random $\hat{\ratio}$ from a given distribution $\boldsymbol{\xi}$.  The class of randomized markup mechanisms is $\mechaspace^{\text{\em mark}}$.
\end{definition}

\begin{definition}
\label{def:ppost}
An {\em anonymous price posting} auction -- denoted $\text{{\em AP}}_{\pi}$ -- 
posts a take-it-or-leave-it common price $\pay$ and randomly allocates to the agents who are willing to pay $\pay$ (i.e., any agent $i$ with $\vali\geq \pay$).
\end{definition}

\noindent Lastly, a {\em $k$-lottery} is another special case of $k$-lookahead mechanism.

\begin{definition}
\label{def:klotterymech}
\label{page:klottery}
A {\em $k$-lottery} auction -- denoted $\text{{\em LOT}}_k$ -- is a $k$-lookahead in which a price posting mechanism is used internally: set $\pay = \vali[\cbs{k+1}]$ and allocate randomly to the top $k$ agents.  Most generally, {\em the Lottery} mechanism randomly gives away the item for free: $\text{{\em LOT}}_n = \text{{\em AP}}_0$.
\end{definition}

\subsubsection{Myerson Virtual Values}
\label{a:vv}
The single most important component of Myerson's analysis is the concept of {\em virtual value}.  Myerson illustrates how mechanism design and optimization are greatly simplified by using an amortized analysis to calculate performance, specifically by adding up the ``marginal" gain (or loss) from serving an agent over all possible agent types as the price is monotonically decreased (weighted by the agent's distribution over values), according to the mechanism's allocation rule.  For derivation of virtual value and further discussion of its intuition, see \cite{mye-81} and Chapter 3.3.1 of \cite{hart-20}.

\begin{fact}
\label{fact:vv}
    Given an agent with value $\val$ drawn independently from distribution $\dist$, the agent's {\em virtual value function} $\vv^{\dist}$ (mapping value to virtual value) in an auction fixing each of the following objectives is given by:
    \begin{description}
        \item[Revenue Auction] $\vv^{\dist}(\val) = \val - \frac{1-\dist(v)}{\distp(\val)}$
        \item[Residual Surplus Auction] $\vv^{\dist}(\val) = \frac{1-\dist(\val)}{\distp(\val)}$
        \item[Total Welfare Auction] $\vv^{\dist}(\val) = \val$
    \end{description}
\end{fact}

\noindent As observed, we let the definition for virtual value be overloaded across objectives.  Some results given from the perspective of virtual value are constant across settings, exhibiting the power of virtual values as an analytical tool (e.g. \Cref{thm:myerson}, \Cref{thm:commonbayesopt} below).  We end this section with the following useful observation about virtual value functions, which states that virtual value $\vv^{\dist}(\hat{\val})$ at $\hat{\val}\geq z$ is unchanged when a draw from $\dist$ is conditionally known to be at least $z$.

\begin{fact}
\label{fact:vvconstantunderconditional}
Given a revenue, residual surplus, or total welfare objective, and a distribution $\dist$ with domain $[a,b]$ (or $[a,b=\infty)$).  Let $\bott{\dist}^z$ be the distribution resulting from conditioning one random draw $\val\sim \dist$ by $\val\geq z$ for $a\leq z \leq b$.  Then for $\hat{\val}\geq z$,
\begin{equation}
    \vv^{\dist}(\hat{\val}) = \vv^{\bott{\dist}^z}(\hat{\val})
\end{equation}
because the operation of conditioning $\val\geq z$ applies the same multiplicative factor $1/(1-\dist(z))$ to both the $(1-\dist(\val))$ and $\distp(\val)$ terms appearing in the revenue and residual surplus virtual value functions -- which cancels.  For total welfare it is trivially true.
\end{fact}

\subsubsection{Monotone Hazard Rate, Regular, and Irregular Distributions}
\label{a:distributionproperties}

This section describes important properties of distributions -- namely \Cref{def:haz} for monotone hazard rate (MHR) and anti-monotone hazard rate (a-MHR); and \Cref{def:reg} for regularity which is related to the definitions of virtual value for various auction objectives.  The properties define settings for canonical analytical settings within mechanism design.  They affect both the strength and complexity of result statements that can be obtained (by the mechanism design literature generally) by acting as natural restrictions on classes of distributions for robust mechanism design.

\begin{definition}
\label{def:haz}
Given a distribution $\dist$, its {\em hazard rate} function $\lambda^{\dist}(\val) = \sfrac{\distp(\val)}{(1-\dist(\val))}$ describes an ``instantaneous rate of failure" of draws from $\dist$.  {\em Monotone hazard rate (MHR)} distributions have $\sfrac{d\lambda^{\dist}(\val)}{d\val}\geq 0$ for all inputs $\val$.  By comparison, {\em anti-monotone hazard rate (a-MHR)} distributions have $\sfrac{d\lambda^{\dist}(\val)}{d\val}\leq 0$.

Let $\scF^{\text{{\em mhr}}}$ be the class of all MHR distributions and $\scF^{\text{{\em a-mhr}}}$ be the class of all a-MHR distributions (each within a context of known input support).
\end{definition}

\noindent We make two observations relating to hazard rate functions.  First, note that the classes $\scF^{\text{mhr}}$ and $\scF^{\text{{\em a-mhr}}}$ 
are disjoint excepting that they share a ``boundary" when $\sfrac{d\hazf^{\dist}}{d\val~(\val)} =0$ for all inputs $\val$.  Second, note that the multiplicative-inverse of hazard rate $(\lambda^{\dist}(\val))^{-1}$ appeared above in the virtual value function for both revenue and residual surplus objectives.

\begin{definition}
\label{def:reg}
For a virtual value function $\vv(\cdot)$ paramterized by a given auction objective, a distribution $\dist$ is {\em  regular} if $\sfrac{d\vv(\val)}{d\val}\geq0$ for all inputs $\val$.  Otherwise it is {\em irregular}.

Let $\scF^{\text{{\em reg}}}$ be the class of all regular distributions (within the context of known input support and a given auction objective).
\end{definition}

\noindent The following explains relationships between the property-based classes of this section for auctions with specific objectives.

\begin{fact}
\label{fact:revregmhr}
Given a revenue objective, the class of MHR distributions is a subset of the class of regular distributions, which is a subset of all distributions: $\scF^{\text{{\em mhr}}}\subset \scF^{\text{{\em reg}}}\subset \scF^{\text{{\em all}}}$.

Given a residual surplus objective, the class of a-MHR distributions and the class of regular distributions are equal: $\scF^{\text{{\em a-mhr}}} = \scF^{\text{{\em reg}}}  \subset \scF^{\text{{\em all}}}$.
\end{fact}

\noindent With respect to modification of distributions with truncation or conditioning (as defined in \Cref{s:exorestricteddists}), we have the following lemma to describe when distribution properties are necessarily preserved.

\begin{lemma}
\label{lem:truncateddistskeepproperties}
Given a distribution $\dist$ with the MHR property and/or the regularity property in a \underline{revenue} auction setting, its properties are preserved under modification to $\bott{\dist}^a$, $\topt{\dist}^{b'}$, and $\botht{\dist}^{a,b'}$.
\end{lemma}
\begin{proof}
The statement for $\bott{\dist}^a$ follows directly from \Cref{fact:vvconstantunderconditional}.  The statement for $\topt{\dist}^{b'}$ holds because calculations of hazard rate and virtual value (for revenue) for inputs less than $b$ are unaffected by top-truncation to a point mass at input $b$, and at input $b$ the hazard rate becomes $\infty$ and virtual value becomes $b$ which are both automatically sufficient to preserve the respective original properties.  The statement for $\botht{\dist}^{a,b'}$ holds by sequential application of the first two cases.
\end{proof}

\noindent Whether or not properties are preserved under top-conditioning with re-normalization of the density (rather than moving to point mass as in \Cref{lem:truncateddistskeepproperties}) is dependent on the distribution in question.

\subsubsection{Quantiles and ``Revenue" Curves}
\label{s:quantsandrev}

For distribution $\dist$, the \emph{quantile} $\quant$ of an agent with value $\val$ denotes how weak that agent is relative to the distribution $\dist$, i.e., the probability that a random draw from $\dist$ will be at least $\val$.\footnote{\label{foot:percentilefromquantile} For the places we use it, a {\em percentile} of a value $\val$ is $1-\quant$ to reflect how {\em strong} an agent is relative to distribution $\dist$, i.e., the output of the CDF function which is the probability that a random draw from $\dist$ will be at most $\val$.}  Technically, quantiles are defined by the mapping $\quantf_{\dist}(\val) = 1 - \dist(\val) = \text{Pr}\left[\val'\geq \val~|~\val'\sim \dist \right]$.  Denote the function mapping back to value space by $\valf_{\dist}$, i.e., $\valf_{\dist}(\quant)=\dist^{-1}(1-\quant)$ is the value of the agent with quantile $\quant$.  Note that all functions defined for all inputs in {\em quantile space} have domain $[0,1]$, and that a default random quantile $\hat{\quant}$ is a uniform draw from the range $[0,1]$.  
The rest of this section describing ``Revenue Curves" $\rev_{\dist}$ adopts the standard nomenclature of the revenue perspective.  However unless otherwise stated, everything presented in this section {\em for revenue curves in quantile space} holds for alternative objectives if ``revenue" was replaced with the correct ``performance" measurement.

A single agent {\em revenue curve} $\rev_{\dist}:[0,1]\rightarrow \reals$ gives the revenue from posting a price as a function of quantile $\quant$, i.e., of the probability that the agent accepts the price.  For an agent with value distribution $\dist$, price $\valf_{\dist}(\quant)$ is accepted with probability $\quant$, so revenue is $\rev_{\dist}(\quant)=\quant\,\valf_{\dist}(\quant)$.  We overload the function $\rev_{\dist}$ to also take inputs from value space, defined by $\rev_{\dist}(\val)=\rev_{\dist}(\quantf_{\dist}(\val))=\quantf_{\dist}(\val)\cdot\val=(1-\dist(\val))\val$.\footnote{\label{foot:setvaluedrev} In fact, revenue curves for both quantile space and value space domains are potentially set-valued functions.  For quantile space, set-valued outputs result from regions of the domain of $\dist$ where the CDF is a constant smaller than 1, because revenue (or residual surplus) changes while quantile does not.  For value space, set-valued outputs result from regions of quantile space corresponding to point masses, because quantile drops while value is constant.}~\footnote{\label{foot:altrevs} For residual surplus we will have the overloaded definition $\rev_{\dist}(\quant) = q\cdot \left(\expect_{\hat{\quant}}\left[\valf_{\dist}(\quant) ~|~\hat{\quant}\leq \quant \right] - \valf_{\dist}(\quant)\right) =\int_0^{\quant} \left(\valf_{\dist}(\hat{\quant})- \valf_{\dist}(\quant)\right)d\hat{\quant}$.}  The slope of the revenue curve $\rev'_{\dist}$ is {\em marginal revenue}.

\begin{fact}[\cite{mye-81}]
\label{fact:vv_rprime}
The slope of the revenue curve $\rev'_{\dist}$ -- i.e., the marginal revenue function -- is equal to (density-weighted) virtual value $\vv^{\dist}$:
\begin{equation*}
    \vv^{\dist}(\quant) = \rev'_{\dist}(\quant),\quad\quad \vv^{\dist}(\val) = \rev'_{\dist}(\val)\cdot (-\distp(\val))
\end{equation*}
and regular distributions (\Cref{def:reg}) are equivalent to (weakly) concave revenue curves in quantile space.
\end{fact}

\noindent\label{page:ironintro} Towards analyzing irregular distributions, Myerson implements a second amortization technique called {\em ironing}.\label{page:ironingdefinition}  For continuous regions of quantile space where an agent's allocation function has constant output, expected marginal surplus $\rev'_{\dist}$ -- equivalently expected $\vv^{\dist}$ -- can be treated as its average value of the region, at all points in the region.  Technically, given an ironed region $[a,b]\in[0,1]$, the {\em ironed marginal revenue} at all quantiles $\quant\in[a,b]$ is $\rev_{\dist}(a)+\frac{\rev_{\dist}(b)-\rev_{\dist}(a)}{b-a}\cdot(q-a)$.

\label{page:ironedcurvedef} We define a single agent {\em ironed revenue curve} within the context of optimal analysis (rather than allowing arbitrary choice of ironed regions).  A single agent ironed revenue curve $\bar{\rev}_{\dist}:[0,1]\rightarrow\reals$ is defined only for the quantile space domain (and not for value space), and is defined as the {\em concave hull} of the original revenue curve (which is always possible to achieve for a single agent by ironing exactly all of the regions of the domain where the revenue curve and its concave hull are not already equal).  The definition takes advantage of the following.

\begin{fact}
\label{fact:maxironed}
Given a distribution $\dist$ for a single agent, for fixed $\hat{\quant}$ as an {\em a priori} fixed {\em probability of sale}, the maximum revenue achievable given $\hat{\quant}$ is $\bar{\rev}_{\dist}(\hat{\quant})$.
\end{fact}

\noindent The regions where the revenue curve and ironed revenue curve are not equal are described as {\em strictly ironed}.  Paralleling \Cref{fact:vv_rprime}, we have a corresponding definition for {\em ironed virtual value} $\bar{\vv}^{\dist}$, which is again equal to the slope of the ironed revenue curve as marginal ironed revenue.

\begin{fact}[\cite{mye-81}]
\label{fact:ironedvv}
The slope of the ironed revenue curve $\bar{\rev}'_{\dist}$ -- i.e., the marginal ironed revenue function -- is equal to ironed virtual value $\bar{\vv}^{\dist}$:
\begin{equation*}
    \bar{\vv}^{\dist}(\quant) = \bar{\rev}'_{\dist}(\quant), \quad\quad \bar{\vv}^{\dist}(\val) = \bar{\rev}'_{\dist}(\val)\cdot (-\distp(\val))
\end{equation*}
\end{fact}

\noindent 

\subsubsection{Optimal Bayesian Mechanisms and Foundational Results}
\label{s:firstauctionresults}
This section summarizes pertinent results in mechanism design.  Optimizing revenue from a single agent whose value $\val$ is drawn from a known distribution $\dist$ is straightforward.

\begin{fact}
\label{fact:oneagentopt}
For any distributions $\dist$, the optimal mechanism for a single agent posts the {\em monopoly price} $\valf_{\dist}(\monoq)$ (\Cref{def:ppost}) corresponding to the {\em monopoly quantile} $\monoq = \argmax_\quant \rev_{\dist}(\quant)$.
\end{fact}


\noindent Next we work towards the solution for Bayesian settings with $n$ agents (\Cref{thm:myebayesopt} below).  First we show technically how mechanism performance can be measured using virtual values.  Simply, the expected revenue of a mechanism $\mecha$ with $n$ agents is equal to its expected surplus of marginal revenue, equivalently, its expected surplus of virtual value.  \Cref{thm:myerson} gives two related statements and the differences are bolded.\footnote{\label{foot:altrev} Per previous discussion regarding setting, the exact theorem statement of \Cref{thm:myerson} holds for alternative objectives with their respective definitions of $\vv^{\dist}$ and $\bar{\vv}^{\dist}$.}

\begin{theorem}[\citealp{mye-81}]
  \label{thm:myerson}
  Given any incentive-compatible mechanism $\mecha$ 
with \textbf{any} allocation rule $\allocs(\vals)$, 
the expected revenue of mechanism $\mecha$
for agents with values drawn i.i.d\ from $\dist$ is 
equal to its expected surplus of virtual value, 
i.e., 
\begin{equation}
\label{eqn:vvrev}
  \mecha(\dist) = \sum\nolimits_i \expect_{\vals \sim \dist}\left[\pricei(\vals)\right] = \sum\nolimits_i \Big( \expect_{\vals \sim 
  \dist}\left[\vv^{\dist}_i(\vali)\,\alloci(\vals) \right]+ \rev_{i,\dist}(\quanti = 0)\Big)
\end{equation}
Alternatively, given any incentive-compatible mechanism $\mecha$ 
with allocation rule $\allocs(\vals)$, 
the expected revenue of mechanism $\mecha$
for agents with values drawn i.i.d.\ from $\dist$ is 
equal to its expected surplus of \textbf{ironed} virtual value \textbf{if additionally $\allocs(\vals)$ is constant for each agent $i$ on regions that are strictly ironed by $\bar{\rev}_{i,\dist}$}, 
i.e., then
\begin{equation}
\label{eqn:ironvvrev}
  \mecha(\dist) = \sum\nolimits_i \expect_{\vals \sim \dist}\left[\pricei(\vals)\right] = \sum\nolimits_i\Big(\expect_{\vals \sim \dist}
  \left[\bar{\vv}^{\dist}_i(\vali)\,\alloci(\vals)  \right]+ \rev_{i,\dist}(\quanti = 0)\Big)
\end{equation}
\noindent The optimal single-item Bayesian auction $\OPT_{\dist}$ given $\dist$ is the one that maximizes expected surplus of virtual value, or equivalently, the one that maximizes ironed virtual value.
\end{theorem}

\noindent \Cref{thm:myerson} gives a clean description of mechanism performance using a reduction to virtual value, with the abstract description of the optimal mechanism following directly at the end of the theorem statement.  \noindent The following corollary makes explicit the optimal structure for auctions within the setting of regular distributions $\scF^{\text{reg}}$:

\begin{theorem}[\citealp{mye-81}]
\label{thm:commonbayesopt}
\label{thm:myebayesopt}
  For i.i.d., regular, single-item auctions with any objective, the optimal mechanism $\OPT_{\dist}$ is the second-price auction with uniform reserve price equal to the monopoly price.
\end{theorem}

\noindent We will use the following \Cref{lem:DRY-15} for the calculation of the performance of specific optimal mechanisms for our dual blends analyses (towards proving the revenue gap of \Cref{thm:finitequadsversusunifsrevpiauction} in \Cref{s:revgapquadsvsunifs} and the residual surplus gap of \Cref{thm:finitequadsversusunifsressurppiauction2} in \Cref{s:ressurpgapquadsvsunifs}).

For the revenue objective and specifically $n=2$, evaluating a mechanism's performance via virtual values has a natural, geometric interpretation.\label{page:dry15original}  An extension of this lemma is given in \Cref{s:ressurpquadsvsunifssupporting} for use there.

\begin{lemma}[\citealp{DRY-15}]
  \label{lem:DRY-15}
  In i.i.d.\ two-agent single-item settings, the expected revenue
  of the second price auction is twice the area under the revenue
  curve and the expected revenue of the optimal mechanism is twice the
  area under the smallest monotone concave upperbound of the revenue
  curve.
\end{lemma}


\subsubsection{Distribution-Class Boundaries and Equal ``Revenue" Distributions}
\label{s:eqr}
\label{a:eqr}

This section gives technical description relating to structure and usage for some of the most pertinent distributions in mechanism design.  Qualitatively, the distributions discussed in this section are inferred to be important by having one or both of the following properties (in the context of one of the relevant auction objectives for this paper): (a) the distribution defines a {\em boundary} of the MHR/a-MHR or regular classes of distributions; and/or (b) the distribution has {\em constant virtual value} at all values of its domain given the objective.

Not surprisingly, a major theme from identifying these distributions is that particular {\em boundary} distributions which meet definitions of class-restrictions with equality are the same ones used to prove tightness in a variety of theorem statements.  
\ifarxiv (Recall the example in the introduction using equal revenue distribution / point mass distributions.)
\fi
To these ends, notice that MHR and regularity are both properties relating to monotonicity of functions parameterized by distributions, respectively hazard rate and virtual value function.

\begin{definition}
\label{def:bounddist}
A distribution is a {\em boundary distribution} for a given class when its characterizing derivative evaluates to a constant $0$ for all (relevant) inputs.
\end{definition}

\noindent The following fact describes some of these characteristics for distributions that act as boundaries for classes requiring the MHR, a-MHR, or regularity properties.

\begin{fact}
\label{fact:distborder}
The following are true about distribution class boundaries.
\begin{enumerate}
    \item For revenue, the 
    boundary of the MHR 
    class of distributions ($\scF^{\text{{\em mhr}}}$) 
     -- requiring the derivative of hazard rate be equal to $0$ on upwards closed domain-- are as follows:
    \begin{itemize}
    \item the general case where hazard rate is a constant $(\hazr)$:\\
        is the shifted exponential distribution `$\Sed$' parameterized by its (shifted) domain lower bound $a\geq 0$ and its hazard rate $\hazr>0$:
    \begin{align*}
        &\Sed_{a,\hazr}(\val) = 1-e^{-\hazr (\val-a)}~\text{{\em for}}~\val\in[a,\infty)& \hazf^{\Sed_{a,\hazr}}(\val) = \hazr\\
        &\sed_{a,\hazr}(\val) = \hazr\cdot e^{-\hazr (\val-a)}& d\hazf^{\Sed_{a,\hazr}}/d\val~(\val) = 0
    \end{align*}
    \item a special case where hazard rate is infinite $(\infty)$:\\
    is the point mass function `$\Pmd$' (c.f., a Dirac function\footnote{\label{foot:dirac} If formal definitions are necessary, we use the following Dirac function technique, which most naturally aligns with the formal definitions needed to evaluate integrals in the limit $dx\rightarrow 0$,
    \begin{equation*}
        \Pmd_a \vcentcolon= \lim\nolimits_{dx \rightarrow 0} [\Pmd_a^{dx}(x) =(x-a)/dx ~\text{for}~x \in[a,a+dx]],\quad \pmd_a \vcentcolon= \lim\nolimits_{dx \rightarrow 0} [\pmd_a^{dx}(x) = 1/dx]
    \end{equation*}
    and if necessary, a point mass at the upper end point $h$ of an integral is modified to $h+dx$ (which is inconsequential because it is inside the limit).  However we trust that this paper's computation of expected values of functions over inputs drawn from distributions incorporating point masses is clearly correct; which is: by separating out the contribution of the function value at the point mass as an additive term with probability weight equal to the point mass.
    }) parameterized by constant output $a$; this is effectively derived from the shifted exponentials in the general case using in-the-limit analysis as $(\hazr\rightarrow \infty)$:
        \begin{align*}
        &\Pmd_a(\val) = 1 ~\text{{\em at }}~\val = a& \hazf^{\Pmd_a}(\val) = \infty\\
        & \pmd_a(\val) \vcentcolon= \infty & d\hazf^{\Pmd_a}/d\val~(\val) = 0
    \end{align*}
    \end{itemize}
    \item For residual surplus,
    \begin{itemize}
        \item the (common) strong-boundary of the MHR and a-MHR classes of distributions ($\scF^{\text{{\em mhr}}}$ and $\scF^{\text{{\em a-mhr}}}$) is the specific exponential distribution `$\Exd$' (which requires fixing lower bound $a=0$) parameterized only by its hazard rate $\hazr>0$:
    \begin{align*}
        &\Exd_{\hazr}(\val) = 1-e^{-\hazr \val}~\text{{\em for}}~\val\in[0,\infty)& \hazf^{\Exd_{\hazr}}(\val) = \hazr\\
        &\exd_{\hazr}(\val) = \hazr\cdot e^{-\hazr \val}& d\hazf^{\Exd_{\hazr}}/d\val~(\val) = 0
    \end{align*}
    \item in fact $\scF^{\text{{\em a-mhr}}} = \scF^{\text{{\em reg}}}$ (from \Cref{fact:revregmhr}), so the boundary of the class of regular distributions $\scF^{\text{{\em reg}}}$ is again the exponential distribution $\Exd$.  They are the same because the virtual value function (given distribution $\dist$) for residual surplus is equal to $1/(\hazf^{\dist})$ and must be non-decreasing to be regular, and similarly the hazard rate $(\hazf^{\dist})$ must be non-increasing to be a-MHR.  Clearly these are equivalent conditions. Specifically, we have:
    \begin{equation*}
        \vv^{\Exd_{\hazr}}(\val) = 1/\hazr,\quad d\vv^{\Exd_{\hazr}}/d\val~(\val) = 0
    \end{equation*}
    \end{itemize}
    \item For revenue, the boundaries of the class of regular distributions $\scF^{\text{{\em reg}}}$ -- requiring the derivative of virtual value be equal to $0$ -- are as follows:
    \begin{itemize}
        \item an important special case where virtual value is the constant $0$:\\
        is the quadratic distribution `$\Qud$' parameterized by its domain lower bound $a>0$:
    \begin{align*}
        &\Qud_a(\val) = 1-a/\val~\text{{\em for}}~\val\in[a,\infty)& \vv^{\Qud_a}(\val) = 0\\
        & \qud_a(\val) = a/\val^2 & d\vv^{\Qud_a}/d\val~(\val) = 0
    \end{align*}
    \item the general case where virtual value is a constant ($\vv$):\\
    is the shifted quadratic distribution `$\Sqd$' parameterized by its domain lower bound $a\geq 0$ and its shift $\vv < a$ (for which there is no lower bound on $\vv$ and for which setting $\vv=0$ gives the previous special case):
    \begin{align*}
        &\Sqd_{a,\vv}(\val) = 1-(a-\vv)/(\val-\vv)~\text{{\em for}}~\val\in[a,\infty)& \vv^{\Sqd_{a,\vv}}(\val) = \vv\\
        & \sqd_{a,\vv}(\val) = (a-\vv)/(\val-\vv)^2 & d\vv^{\Sqd_{a,\vv}}/d\val~(\val) = 0
    \end{align*}
    \item a special case where virtual value is a positive constant ($a=vv>0$):\\
    is the point mass function `$\Pmd$' (c.f., a Dirac function$^{\ref{foot:dirac}}$) parameterized by constant output $a$; this is effectively derived from the general case using in-the-limit analysis as $(vv\rightarrow a^-)$:
        \begin{align*}
        &\Pmd_a(\val) = 1 ~\text{{\em at }}~\val = a& \vv^{\Pmd_a}(\val) = a\\
        & \pmd_a(\val) \vcentcolon= \infty & d\vv^{\Pmd_a}/dx~(\val) = 0
    \end{align*}
    \end{itemize}
\end{enumerate}
\end{fact}

\noindent As alluded in \Cref{fact:distborder}, quadratic distributions $\Qud_a$ are the $\vv=0$ special case of the shifted quadratic distributions $\Sqd_{a,0}$.  Quadratic distributions play an important role in auction design for the revenue objective, where they are examples of {\em equal revenue distributions} (EQRs).  Equal revenue distributions have the following definition and properties (\Cref{fact:eqrrev}).

\begin{definition}
\label{def:eqr}
A distribution is an {\em equal revenue distribution (EQR)} if {\em all} 1-agent price posting auctions have the same expected revenue.
\end{definition}

\noindent Generally, the Quadratics $\Qud_a$ describe exactly the class of {\em regular} equal revenue distributions (and they maintain both the regularity and equal-revenue properties under top-truncation).

\begin{fact}
\label{fact:eqrrev}
The following are true about equal revenue distributions.
\begin{itemize}
    \item  A sufficient condition for a distribution with domain $[a,\infty)$ to be an {\em equal revenue distribution (EQR)} is that its virtual value function evaluates to $0$ at all quantiles $\quant\in(0,1]$ corresponding to values at least $a$.  All quadratic distributions $\Qud_a$ meet this condition.
    \item Consider a Bayesian revenue auction with 1 agent whose value is drawn from a quadratic distribution $\Qud_a$.  The expected revenue of any price-posting auction 
    with price $\pay\in[a, \infty)$ is $a$, i.e., posting any price $\pay\geq a$ gets {\em equal revenue}.
\end{itemize}
\end{fact}

\noindent There exists an indirect extension of ``EQR" to the residual surplus setting which we name EQRS.  An indirect extension of the EQR-concept is necessary because virtual values for residual surplus are strictly positive everywhere except at $\quant = 0$ where they are 0 (see \Cref{fact:vv}) and on point mass regions of quantile space, so there is no analogous distribution (in an auction with 1 agent) that achieves equal residual surplus for all posted prices.  

For residual surplus, the key adaptation towards establishing an EQRS is to require equal performance of price posting {\em critically within the context of knowing that the price will be accepted by at least 1 agent} (i.e., an agent with unknown larger value).  For residual surplus, assuming a price will trade is fairly natural because heuristically this condition holds at price 0 where all virtual values are positive. Intuitively, only an irrational auction would increase a posted price to a level at which it might not trade.  Further, exponential distributions are the natural class of EQRSs under this condition, which makes sense because exponentials are the unique boundary of regularity for residual surplus.  We reinforce these intuitions by \textbf{bolding} in \Cref{fact:eqpintro} the key differences of EQRS, in comparison to revenue and EQR.

\begin{fact}
\label{fact:eqpintro}
Consider a Bayesian residual surplus auction \textbf{with $\mathbf{n\geq1}$ agents} whose values are drawn i.i.d.\ from an exponential distribution $\Exd_{\hazr}$.  The expected residual surplus of any rational-price\footnote{\label{foot:hackrationalpricedef} For the purposes of \Cref{fact:eqpintro}, a {\em rational-price} for residual surplus is a price that can not be lowered by $d\pay$ without changing the set of agents who would accept it.  The set of rational prices is necessarily $\{\vali[(2)],\ldots,\vali[(n)],0\}$ (note, residual surplus auctions default to rejecting an agent with value equal to price).  All other prices are considered ``irrational" because a priori the residual surplus objective is strictly increasing with a $d\pay$ drop in price.} posting auction with price $\pay$ \textbf{that is guaranteed to trade (without conditioning on the value of any winning agent)} is $1/\hazr$.  I.e., for values $\vali[(1)]\geq\vali[(2)]\ldots\geq\vali[(n)]$, posting any price $\pay \in \{\vali[(2)],\ldots,\vali[(n)],0\}$ gets equal residual surplus.

Sufficient conditions for a distribution to be an {\em equal residual surplus distribution (EQRS)} \textbf{are that its virtual value function evaluates to a constant $\hazr$ at all quantiles $\quant\in (0,1]$, and the distribution has domain lower bound at 0.}  The class of exponential distributions $\Exd_{\hazr}$ meet this condition.
\end{fact}

\noindent We conclude with ``canonical" definitions for equal revenue / residual surplus distributions, because these unique, simple forms are frequently sufficient for result statements.

\begin{definition}
\label{def:eqreqp}
The {\em canonical equal revenue distribution -- i.e., \underline{the} equal revenue distribution --} is $\Eqrd = \Qud_1$.

\noindent The {\em canonical equal residual surplus distribution -- i.e., \underline{the} equal residual surplus distribution} -- is $\Eqrsd= \Exd_1$.
\end{definition}

\subsubsection{Motivating Results in Prior Independent Mechanism Design}
\label{a:ab18andhjl20results}

This section states two recent results in prior independent mechanism design which led up to this work.  Both results are for single item, 2-agent, truthful auctions with a revenue objective.  Each result identifies the mechanism that is optimal for its distribution class, respectively MHR ($\scF^{\text{mhr}}$) and regular ($\scF^{\text{reg}}$).  They include further specifications for setting which differ from the main thrust of this paper: first, the results given here require unbounded values which each have support $\valspace = (0,\infty)$ (which technically will not fit into our bounded value support $\valspace = [1,\maxval]$ for mechanism design results in \Cref{s:examplemdresults}).  Second, the results here require an additional restriction to {\em scale-invariant mechanisms}, i.e., mechanisms whose performance necessarily scales linearly with input vectors.

\begin{theorem}[\citealp{AB-18}]
\label{thm:pioptn2revmhr}
Given a single item, 2-agent revenue auction, the optimal truthful mechanism against MHR distributions $\scF^{\text{\em mhr}}$ for the prior independent design program $(\piratio^{\scF^{\text{\em mhr}}})$ is the SPA.  The worst-case MHR distribution for this mechanism is the truncated-exponential $\topt{\Exd}_1^{.852'}$ with its monopoly quantile $\monoq^* \approx \qsimpmhr$ and its approximation ratio is $\piratio^{\scF^{\text{\em mhr}}} \approx \apxsimpmhr$.
\end{theorem}

\noindent \citet{AB-18} additionally proved a bounded range for the regular setting of the optimal prior independent approximation factor: $\piratio^{\scF^{\text{reg}}}\in[1.80,1.95]$.  Their lower bound was the first ever non-trivial lower bound for 2-agent prior independent mechanism design.  Previously, \citet{DRY-15} had shown an upper bound of 2 and \citet{FILS-15} had exhibited that the upper bound of 2 was not tight.  The next theorem solves this case and gives the optimal mechanism.

\begin{theorem}[\citealp{HJL-20}]
\label{thm:pioptn2revenue}
  Given a single item, 2-agent revenue auction, the
  optimal truthful, scale-invariant mechanism (from the class $\mechaspace^{\text{\em si}}$) against regular distributions $\scF^{\text{\em reg}}$
  for the prior independent design program $(\piratio^{\scF^{\text{\em reg}}})$ is $\mecha_{\hat{r},\xi}$
  which randomizes according to $\boldsymbol{\xi}$ over the second-price auction $\mecha_1$ with
  probability $\xi_1$ and $\optratio$-markup mechanism
  $\mecha_{\optratio}$ with probability $\xi_{\optratio}=1-\xi_1$, where
  $\xi_1 \approx \asimp$ and $\optratio\approx \rsimp$. The
  worst-case regular distribution for this mechanism is the truncated-shifted-quadratic
  $\topt{\Sqd}_1^{9.7405'}$ with its monopoly quantile $\monoq^* \approx \qsimp$ and its approximation
  ratio is $\piratio^{\scF^{\text{\em reg}}} \approx \apxsimp$.
\end{theorem}

\subsection{Revenue Gap from Quadratics-versus-Uniforms}
\label{s:revgapquadsvsunifs}

The goal of this section is to use the Blends Technique (\Cref{def:thetheblendstechnique}) to prove a revenue gap for the Quadratics-versus-Uniforms dual blend, resulting in a prior independent lower bound (summarized in \Cref{eqn:finitequadsversusunifsrevpiauction} in \Cref{thm:finitequadsversusunifsrevpiauction} and copied at the end of this section).

Recall value space is $\valspace^2=[1,\maxval]^2$ with an assumption that $\maxval>2$.\footnote{\label{foot:revgapmaxvalgeq2expl} The assumption of $\maxval>2$ is necessary to make the result interesting.  Because of the assumption that value space has domain lower bound at 1, uniform distributions with domain upper bound at most 2 are trivially optimized by the SPA -- which is the same as for the Quadratics in these dual blends -- and thus do not induce an approximation gap (the ratio is 1).}  We still use symmetry to assume $\vali[1]\geq\vali[2]\geq0$ in calculations and will then count permutations where necessary.  For use in this section and the next, we re-state the finite-weight blends solution of Quadratics-versus-Uniforms (copied from page~\pageref{page:finiteweightquadsvsunifs}).

The weights of the upward-closed Quadratics blend ($\blendi[1]$) are as follows:\label{page:appendixfiniteblendsdef}
\begin{itemize}
\item point mass of weight $o_{\text{pm}} = 1$ on (truncated) distribution $\topt{\Qud}^{\maxval'}_1$;
\item weights $o_{Qz} = \frac{2}{z}dz$ on all upward-closed (truncated) distributions $\topt{\Qud}^{\maxval'}_{z}$ for $z\in\left[1,\maxval\right]$.
\end{itemize}

\noindent The weights of the downward-closed Uniforms blend ($\blendi[2]$) are as follows:
\begin{itemize}
\item point mass of weight $\omega_{\text{pm}} = \frac{(2\maxval-1)^2}{
\maxval^2}$ on (truncated) distribution $\topt{\Ud}_{1,2\maxval}^{\maxval'}$;
\item weights $\omega_{Uz} = \frac{2(z-1)^2}{z^3}dz$ on all downward-closed distributions $\Ud_{1,z}$ for $z\in\left[1,\maxval\right]$.
\end{itemize}
\noindent The total weight is finite ($1 + 2\ln \maxval$), which implies that it is sufficient to use these weights directly in calculating a revenue gap (per \Cref{fact:finiteweightok}).\footnote{\label{foot:recallfiniteweightok} Recall explicitly, we are interested in the ratio $\sfrac{\text{opt}_{2,2}}{\text{opt}_{2,1}}$.  Using the incorrect total weight simply scales both numerator and denominator equally, which cancel as long as total weight is finite.  Following from this, we may slightly abuse language and say that $g$ is a ``distribution" even when referring to its un-normalized weight assignments.}  We also summarize the function $g$:
\begin{align}
    \label{eqn:gapcalcssummaryofg}
    g_{2D}(\vals) &= \text{result of $\topt{\Qud}_z^{\maxval'}$ blend}\\
    \nonumber
    &= \text{result of $\Ud_{0,z}$ blend and $\topt{\Ud}_{1,2\maxval}^{\maxval'}$} =  \frac{1}{\vali[1]^2}\\
    \nonumber
    g_{0D}(\vals) &= \text{result of truncation in $\topt{\Qud}_z^{\maxval'}$ blend}\\
    \nonumber
    &= \text{no point mass density from $\Ud_{0,z}$ blend, result of $\topt{\Ud}_{1,2\maxval}^{\maxval'}$}  = 1\\
    \nonumber
    g_{1D}(\vals) &= \text{result of \textbf{exactly one} truncation in $\topt{\Qud}_z^{\maxval'}$ blend}\\
    \nonumber
    &= \text{no point mass density from $\Ud_{0,z}$ blend, result of $\topt{\Ud}_{1,2\maxval}^{\maxval'}$}= \frac{1}{\maxval}
\end{align}

\noindent Also for use in this section, we state the following fact regarding uniform order statistics.
\begin{fact}
\label{fact:orderstatexpectations}
Given (unordered) $\vals = (\vali[1],\ldots,\vali[n])$ which are $n$ i.i.d. draws from the uniform distribution $\Ud_{0,1}$.  Let $k=1$ be the first, largest order statistic, and count order statistics down to $k=n$ the last, smallest order statistic.  The expected value of an order statistic $\vali[(k)]$ is given by $\mathbf{E}_{\vals\sim\Ud_{0,1}}\left[\vali[(k)]\right] = \frac{n+1-k}{n+1}$.
\end{fact}

\subsubsection{Expected Optimal Revenue from Quadratics}
\label{s:expectedrevenuequads}

We calculate the expected optimal revenue from the Quadratics side $\text{opt}_{2,1}$ using the $\boldsymbol{o}$ weights above.  The revenue of the Quadratics blend is easy to calculate {\em because every distribution that is a component of the blend is an equal revenue distribution} for which offering every price in $[\vali[(2)],\maxval]$ to the largest-valued agent gets the same revenue and is optimal (see \Cref{def:eqr} and its surrounding discussion).

The immediate consequence is that there exists a single mechanism that is optimal against every distribution in the $\blendi[1]$ Quadratics blend: the anonymous price-posting mechanism $\text{AP}_{\maxval}$ with constant price $\maxval$ is globally optimal. Therefore, $\text{opt}_{2,1}=\text{AP}_{\maxval}(g)$ which is the revenue of posting price $\maxval$ against the entire correlated distribution $g$.

The revenue conditioned on selling is obviously $\maxval$.  The probability of selling can be obtained from the equations of line~\eqref{eqn:gapcalcssummaryofg} to determine total density where at least one agent has value $\maxval$, which is exactly the total of 0D and 1D density over all of value space: $g_{0D}\cdot1+ 2 \cdot g_{1D}\cdot(\maxval-1)=\sfrac{(3\maxval-2)}{\maxval}$. Revenue from the Quadratics blend $\blendi[1]$ is given by
\begin{equation}
\text{opt}_{2,1} = \text{AP}_{\maxval}(g) = \maxval \cdot \frac{(3\maxval-2)}{\maxval} =3\maxval-2
 \end{equation}

\subsubsection{Expected Optimal Revenue from Uniforms}
\label{s:expectedrevenueunifs}

Because of the lower bound of value space at 1, optimal revenue analyses for the Uniforms break down by both distribution and type of optimal reserve price (which is either the $\Ud_{0,z}$ monopoly reserve price, or the lower bound 1).  The way we implement all Uniforms distributions here is equivalent to conditioning a random draw from $\Ud_{0,z}$ to be at least 1.  This structure makes \Cref{fact:vvconstantunderconditional} applicable to our distributions, thus we can use virtual values as if the values were drawn from $\Ud_{0,z}$ rather than its respective $\Ud_{1,z}$.  By observation, all Uniforms with positive weight in $\blendi[2]$ are regular.  It is a well-known corollary to \Cref{thm:commonbayesopt} that the optimal mechanism given $n$ agents drawn i.i.d.\ from a uniform distribution $\Ud_{0,z}$ is the SPA with reserve price at the monopoly price $\sfrac{z}{2}$.  We have the following summary of monopoly prices over the distributions in $\blendi[1]$:
\begin{itemize}
\item  monopoly price $\maxval$ for $\topt{\Ud}_{1,2\maxval}^{\maxval'}$ (the truncation at $\maxval$ observably does not affect the monopoly price);
\item monopoly price 1 for $\Ud_{1,z}$ for $z\in[1,2]$ (effectively the SPA);
\item monopoly price $\sfrac{z}{2}$ for $\Ud_{1,z}$ for $z\in[2,\maxval]$ (which is the same as the monopoly price for $\Ud_{0,z}$).
\end{itemize}

\noindent We treat these cases in sequence to compute $\text{opt}_{2,2}$, incorporating the weights $\boldsymbol{\omega}$.  The contribution of the $\topt{\Ud}_{1,2\maxval}^{\maxval'}$ distribution is actually the same as the entire revenue $\text{opt}_{2,2}= 3\maxval-2$ of the Quadratics.  This follows from: its optimal mechanism -- post price $\maxval$ -- is the same as the globally optimal mechanism for Quadratics; and, the distribution $\topt{\Ud}_{1,2\maxval}^{\maxval'}$ accounts for every possible realization in which an agent shows up with value $\maxval$ (when considering any distribution in $\blendi[2]$).

The second and third cases represent revenue from the Uniform blend's random draw according to weights $\omega_{Uz} = \frac{2(z-1)^2}{z^3}dz$.  For the case $z\leq2$, the optimal mechanism is the SPA, therefore the optimal mechanism always sells and its expected revenue is exactly the expected value of $\vali[2]$.  Using \Cref{fact:orderstatexpectations}, the expected value of $\vali[2]$ is [1 plus $\sfrac{1}{3}$ of the width of the range $[1,z]$].  The total contribution from this case $z\leq2$ is:

\begin{align}
\label{eqn:uniformrevenuecasein1comma2}
&\OPT_{\Ud_{1,z}}(\Ud_{1,z})\\
\nonumber
=&\int_1^2 \frac{2(z-1)^2}{z^3}\cdot\left(1+\frac{z-1}{3}\right)~dz = \int_1^2 \frac{2z^2-4z+2+\frac{2}{3}z^3-2z^2+2z-\frac{2}{3}}{z^3}~dz\\
\nonumber
=&\int_1^2 \frac{2}{3} -\frac{2}{z^2} + \frac{4}{3z^3}~dz = \frac{2}{3} + \left[\frac{2}{z}-\frac{2}{3z^2}\right]_1^2 = \frac{2}{3}+ \left(1 - \frac{1}{6}\right) - \left(2- \frac{2}{3}\right)= \frac{1}{6}
\end{align}

\noindent For the last case $z \geq2$, the expected revenue conditioned on selling depends on distribution-specific reserve prices (there is no reduction to the SPA).  We separate uniform draws of the values $\vals$ into three natural sub-cases and calculate the optimal revenue afterwards (given a $\Ud_{1,z}$):~\footnote{\label{foot:explainunifrangeundermonopoly} Note, the quantity $\sfrac{z}{2} - 1 = \sfrac{(z-2)}{2}$ is the length of the (uniform) region below the monopoly price at $\sfrac{z}{2}$, versus $\sfrac{z}{2}$ is the length above it.}
\begin{itemize}
\item both values are smaller than the reserve; we fail to sell, with probability $\frac{(\frac{z-2}{2})^2}{(z-1)^2}$;
\item (2 permutations of) one value is larger and one smaller than the reserve; we sell at the reserve price of $\sfrac{z}{2}$, each permutation with probability $\frac{(\frac{z-2}{2})(\frac{z}{2})}{(z-1)^2}$;
\item both values are larger than the reserve; we sell at the conditional expected value of $\vali[2]$ which is $\mathbf{E}_{\vals\sim\Ud_{\sfrac{z}{2},z}}\left[\vali[2] \right]= \sfrac{2z}{3}$, with probability $\frac{(\frac{z}{2})^2}{(z-1)^2}$.
\end{itemize}
\noindent The optimal revenue from a distribution $\Ud_{1,z}$ for $z\geq2$ is
\begin{equation}
\OPT_{\Ud_{1,z}}(\Ud_{1,z})=
2\cdot\frac{(\frac{z-2}{2})(\frac{z}{2})}{(z-1)^2}\cdot \frac{z}{2} + \frac{(\frac{z}{2})^2}{(z-1)^2}\cdot\frac{2z}{3} = \frac{1}{(z-1)^2}\left(\frac{5}{12}z^3-\frac{1}{2}z^2\right)
\end{equation}
\noindent Analogous to equation~\eqref{eqn:uniformrevenuecasein1comma2}, the total contribution from this last case $z\geq2$ results from a weighted integral and is:
\begin{equation}
\int_2^{\maxval} \frac{2(z-1)^2}{z^3}\cdot \OPT_{\Ud_{1,z}}(\Ud_{1,z})~dz= \int_2^{\maxval} \frac{5}{6} -\frac{1}{z}~dz =\frac{5}{6}(\maxval-2) -\ln \sfrac{\maxval}{2}
\end{equation}
\noindent Total revenue from the Uniforms blend adds up as
\begin{equation}
\text{opt}_{2,2} = (3\maxval-2)+\frac{1}{6}+\frac{5}{6}(\maxval-2) -\ln \sfrac{\maxval}{2} = \frac{23}{6}\cdot \maxval - \frac{7}{2}  -\ln \sfrac{\maxval}{2}
\end{equation}

\subsubsection{The Revenue Gap as Lower Bound}
\label{s:revenuegapfromratio}
For 2-agent, 1-item prior independent auctions with a revenue objective, we have now established a necessary revenue gap via blends:
\begin{equation}
\label{eqn:twentythreeeighteenths}
    \frac{\text{opt}_{2,2}}{\text{opt}_{2,1}} = \frac{\frac{23\maxval}{6}-\frac{7}{2}-\ln(\sfrac{\maxval}{2})}{3\maxval-2},\quad\quad \lim_{\maxval\rightarrow\infty} \frac{\text{opt}_{2,2}}{\text{opt}_{2,1}} = \frac{23}{18} \approx 1.27777
\end{equation}
\noindent where the limit calculation is trivial from application of l'Hopital's rule.  The applicability of the revenue gap as a prior independent lower bound is subject to the design problem's parameter $\scF$ to describe the adversary's allowable choice set of distributions.  By inspection, all of the Uniforms composing $\blendi[2]$ are included in small-class $\scF^{\text{unif}}[1,\maxval]=\{\topt{\Ud}_{1,b}^{\maxval'}~:~ 1 \leq b \}\equiv \text{uniforms on}~[1,b]~\text{truncated at}~\maxval$.  Putting this together with equation~\eqref{eqn:twentythreeeighteenths} and \Cref{thm:blendsbound} gives \Cref{thm:finitequadsversusunifsrevpiauction}:

\begin{numberedtheorem}{\ref{thm:finitequadsversusunifsrevpiauction}}
\label{thm:finitequadsversusunifsrevpiauction2}
Given a single-item, 2-agent, truthful auction setting with a revenue objective and with agent values restricted to the space $[1,\maxval]$ for $\maxval>2$.  For the class of uniform distributions $\scF^{\text{{\em unif}}}$, the optimal prior independent approximation factor of any (truthful) mechanism is lower bounded as:
\begin{equation*}
    \piratio^{\scF^{\text{\em unif}}}_{\maxval} \geq 
    \frac{\text{{\em opt}}_{2,2}}{\text{{\em opt}}_{2,1}}=
    \frac{\frac{23\maxval}{6}-\frac{7}{2}-\ln(\sfrac{\maxval}{2})}{3\maxval-2}=\lowb_{\maxval}^{\scF^{\text{\em unif}}}
\end{equation*}
\noindent The lower bound $\lowb^{\scF^{\text{{\em unif}}}}_{\maxval}\rightarrow 
\sfrac{23}{18}\approx\EXBREV$ as $\maxval\rightarrow\infty$ and this is the supremum of $\lowb^{\scF^{\text{{\em unif}}}}_{\maxval}$ over $\maxval\geq 1$.
\end{numberedtheorem}

\noindent Further,  all Uniforms are regular (for revenue, for which $\scF^{\text{reg}}$ is the standard comparison class of distributions for prior independent design), thus as a corollary, our bound here holds for $\scF^{\text{reg}}$.

Previously in \Cref{thm:pioptn2revenue} in \Cref{a:ab18andhjl20results}, the {\em optimal} prior independent approximation factor was given for the version of this setting which allowed unbounded values in $[0,\infty)$.  The optimal factor was $\sim 1.907$.  Note, we should expect the approximation factor of the restricted value space to be smaller than the unbounded value space, because the mechanism can specifically take advantage of information relating to scale in the latter case.  The optimal mechanism for unbounded value space was an a priori mixture over $\vali[2]$-markup prices of 1 (the SPA) and $\sim 2.44$.  Certainly, we would not expect the optimal mechanism to have the same form for the value space $[1,\maxval]$, because it will not commit a priori to posting a marked-up price of $2\cdot\vali[2]$ when it could be that $\vali[2]\in[\sfrac{\maxval}{2},\maxval]$.  This makes prior independent design in a finite value space setting a distinct problem in terms of analysis, in comparison to the unbounded value space setting.

\subsection{Residual Surplus Gap from Quadratics-versus-Uniforms}
\label{s:ressurpgapquadsvsunifs}

Paralleling the previous section, the goal of this section is to use the Blends Technique to prove a residual surplus gap for the Quadratics-versus-Uniforms dual blend, resulting in a prior independent lower bound (summarized in equation~\eqref{eqn:finitequadsversusunifsressurppiauction} in \Cref{thm:finitequadsversusunifsressurppiauction} and copied at the end of this section).

The description of the dual blend was given in the introduction to the previous section (page~\pageref{page:appendixfiniteblendsdef}).  To re-summarize other points regarding the analysis here: we again use value space $\valspace^2 = [1,\maxval]^2$, 
un-normalized densities (which will cancel at the end), and \Cref{fact:orderstatexpectations} (which states that the expected value of $n$ order statistics from a distribution $\Ud_{0,z}$ divide the range into $n+1$ equal parts).  We include here a similar assumption to the one we had for the revenue gap -- we assume $\maxval\geq \RSMINH$ (otherwise our relaxed analysis does not show a gap).

The residual surplus gap presented here uses the same description of finite-weight Quadratics-versus-Uniforms dual blends, but the adversary will now in fact {\em choose the distribution over the Quadratics and use them to set the benchmark}, whereas the revenue-adversary set the benchmark via the Uniforms.  (Thus, we reassign $\blendi[1]$ to describe the Uniforms side and $\blendi[2]$ to describe the Quadratics side.)  As a consequence, we now have calculation of residual surplus that is ``easy" for the Uniforms rather than for the Quadratics (which is reversed in comparison to revenue calculations). This follows because, for residual surplus, it is now the Uniforms side for which a single dominant mechanism exists (the Lottery of \Cref{def:klotterymech}).  We will calculate $\text{opt}_{2,1}$ for the easy Uniforms side first.

Afterwards, the calculation of $\text{opt}_{2,2}$ for expected optimal residual surplus of the Quadratics faces some technical complexities.  For simplicity, we will calculate instead a lower bound on $\text{opt}_{2,2}$.  This is sufficient because we are designing a residual surplus gap (via the Blends Technique) between the adversary's benchmark set by the Quadratics and an upper bound on expected performance of any algorithm as set by the Uniforms.  By using a lower bound on the ratio's numerator, we will exhibit a weaker -- but legitimate -- non-trivial lower bound on prior independent approximation.

\subsubsection{Expected Optimal Residual Surplus from Uniforms}
\label{s:expectedressurpunifs}

We calculate the expected optimal residual surplus from the Uniforms side ($\text{opt}_{2,1}$) using the $\boldsymbol{\omega}$ weights above.  The residual surplus of the Uniforms blend is easy to calculate because {\em every distribution has everywhere \underline{decreasing} virtual value and therefore it is optimal to iron the entire region of value space}.  The immediate consequence is that there exists a single mechanism that is optimal against every distribution in the $\blendi[1]$ Uniforms blend: the $2$-lottery mechanism $\text{LOT}_2 = \text{AP}_0$ is optimal (see definitions from page~\pageref{page:klottery}).  I.e., $\text{opt}_{2,1}$ can be calculated directly from the expectation of one draw $\val\sim g$ (and with price 0).

The probability of selling is obviously 1.  In this case, the expected residual surplus given  any distribution can be obtained as the mean of the distribution: $\OPT_{\Ud_{1,z}}(\Ud_{1,z}) = 1+\sfrac{(z-1)}{2}=\sfrac{(z+1)}{2}$.  Given the distribution $\topt{\Ud}_{1,2\maxval}^{\maxval'}$, expected residual surplus is calculated to be:

\begin{equation}
\nonumber
\OPT_{\topt{\Ud}_{1,2\maxval}^{\maxval'}}(\topt{\Ud}_{1,2\maxval}^{\maxval'}) = \maxval\cdot\left(\frac{\maxval}{(2\maxval-1)}\right) + \left(1+\frac{(\maxval-1)}{2}\right)\cdot\left(\frac{(\maxval-1)}{(2\maxval-1)}\right)=\frac{(3\maxval^2-1)}{(2(2\maxval-1))}
\end{equation}

\noindent Residual surplus from the Uniforms blend $\blendi[1]$ (using un-normalized weights) gives:
\begin{align}
\text{opt}_{2,1} 
    \nonumber
    &= \omega_{\text{pm}}\cdot \frac{(3\maxval^2-1)}{(2(2\maxval-1))}+\int_1^{\maxval} \omega_{Uz}\cdot\left(\frac{(z+1)}{2}\right)\\
    \nonumber
    &= \frac{(2\maxval-1)^2}{\maxval^2}\cdot \frac{(3\maxval^2-1)}{(2(2\maxval-1))}+\int_1^{\maxval} \frac{2(z-1)^2}{z^3}\cdot \left(\frac{(z+1)}{2}\right)~dz\\
    \nonumber
    &= \frac{(2\maxval-1)\cdot(3\maxval^2-1)}{2\maxval^2}+\left[z-\ln z +\frac{1}{z}-\frac{1}{2z^2}\right]_1^{\maxval}\\
    \nonumber
    &= \frac{6\maxval^3-3\maxval^2-2\maxval+1}{2\maxval^2}+\left[(\maxval-1) -\ln\maxval +\left(\frac{1}{\maxval}-1\right)-\left(\frac{1}{2\maxval^2}-\frac{1}{2}\right) \right]\\
    \label{eqn:foot:ressurpuniformunnorm}
    &= \frac{4\maxval^3-3\maxval^2-\maxval^2\cdot\ln\maxval}{\maxval^2} = 4\maxval-3-\ln\maxval
 \end{align}

\subsubsection{Lower Bound on Expected Optimal Residual Surplus from Quadratics}
\label{s:expectedressurpquads}

The goal of the calculations in this section is to quantify a lower bound on the expected optimal residual surplus from the Quadratics blend $\blendi[2]$.  Thus, we want: $\text{lb}_{2,2}< \text{opt}_{2,2}$.

We do this in place of calculating $\text{opt}_{2,2}$ which is more complicated technically.  Further, the lower bound $\text{lb}_{2,2}$ must be strictly larger than $\text{opt}_{2,1}$ (equation~\eqref{eqn:foot:ressurpuniformunnorm} just above), in which case we can exhibit a prior independent approximation lower bound from the ratio $\sfrac{\text{lb}_{2,2}}{\text{opt}_{2,1}}$ (see the proof of \Cref{thm:finitequadsversusunifsressurppiauction} below).

This section only includes high-level introduction of the structures that are necessary to calculate $\text{lb}_{2,2}$ and state that it is a lower bound.  Therefore we only give here: the residual surplus curve for the Quadratics (recall -- as a function of quantile); the definition for the quantity $\text{lb}_{2,2}$; and \Cref{lem:ressurpofmechopm} which shows that $\text{lb}_{2,2}$ is an appropriate lower bound for our purposes.  Supporting material for this section -- including explanations, sub-calculations, and proofs -- is provided in \Cref{s:ressurpquadsvsunifssupporting}.\label{page:ressurphardboundsupporting}

 First we give the un-ironed residual surplus curve for the Quadratics disributions with CDF of the specific distribution $\topt{\Qud}_{z=1}^{\maxval'}$ given by $\topt{\Qud}_1^{\maxval'}(x) = 1 - \sfrac{1}{x}$ on $x\in[1,\maxval]$, and $\topt{\Qud}_1^{\maxval'}(x) = 1$ for $x\geq \maxval$.  Explanation are given in \Cref{s:ressurpcurvequad1} (page~\pageref{page:quadraticsrscurve}).  The residual surplus curve is
\begin{equation}
\label{eqn:ressurpcurveofquad1h}
\rev_{\topt{\Qud}_1^{\maxval'}}(\quant) =
\begin{cases}
0 & \text{for}~\quant\in[0,\sfrac{1}{\maxval}]\\
\ln(\quant\cdot\maxval) & \text{for}~\quant\in[\sfrac{1}{\maxval},1)\\
[\ln \maxval,1+\ln\maxval] & \text{for}~\quant = 1
\end{cases}
\end{equation}

\noindent Next, recall, we have the definition $\text{opt}_{2,2}= \mathbf{E}_{\dist\sim\blendi[2]}\left[\OPT_{\dist}(\dist) \right]$, which embeds the weights $\boldsymbol{o}$ (from $\blendi[2]$).  At a high level, the quantity $\text{lb}_{2,2}$ is similarly a calculation of weighted residual surplus, according to weights $\boldsymbol{o}$.  With explanation to follow, we formally define:

\begin{align}
    \label{eqn:ressurpquadscalcdefoflb22}
    \text{lb}_{2,2} &\vcentcolon=& \hspace{-2.8cm}
    o_{\text{pm}}\cdot \mechap_{o_{\text{pm}}}(\topt{\Qud}_{1}^{\maxval'}) &\quad+ \left[\int_1^{\maxval} o_z\cdot \text{LOT}_2(\topt{\Qud}_{z}^{\maxval'}) \right]\\
    \nonumber
    & =& \hspace{-2.8cm}
    1\cdot \mechap_{o_{\text{pm}}}(\topt{\Qud}_{1}^{\maxval'}) &\quad+ \left[\int_1^{\maxval} \frac{2}{z}\cdot \text{LOT}_2(\topt{\Qud}_{z}^{\maxval'})~dz \right]
\end{align}

\noindent Specifically, $\text{lb}_{2,2}$ is calculated using the residual surplus of the 2-lottery (on the corresponding distributions) for all weights $o_z = \sfrac{2}{z}\cdot dz$ making up the integral part of the $\blendi[2]$ blend.  (We do this for simplicity even though the 2-lottery is sub-optimal for a range of $z$-parameters of distributions within this component of the blend.)

The only element of $\blendi[2]$ for which it does not use the 2-lottery is $\topt{\Qud}_1^{\maxval'}$ (with weight $o_{\text{pm}}=1$) where it uses the residual surplus $\mechap_{o_{\text{pm}}}(\topt{\Qud}_1^{\maxval'})$ for a specially constructed mechanism $\mechap_{o_{\text{pm}}}$ (for which we defer presentation to \Cref{def:ressurpspecialmech} in \Cref{s:twopieceironmechandperformance}).\label{page:ressurpspecialmech}

The point is that while the lottery is not necessarily optimal where we use its performance, this relaxed lower bound simplifies our calculation generally to only require calculating expected residual surplus for a single Quadratic distribution, in particular the performance of $\mechap_{o_{\text{pm}}}$ on distribution $\topt{\Qud}_{1}^{\maxval'}$.  Note, the total quantity $\text{lb}_{2,2}$ is for comparison only -- there is no prior independent mechanism that can commit to this behavior (which varies by distribution) and achieve this precise performance.

The expected residual surplus $\mechap_{o_{\text{pm}}}(\topt{\Qud}_1^{\maxval'})$ is stated in \Cref{lem:ressurpofmechopm}, though its proof is also deferred to \Cref{s:ressurpspecialmech}.

\begin{lemma}
\label{lem:ressurpofmechopm}
The residual surplus of mechanisms $\mecha_{o_{\text{pm}}}$ and $\text{LOT}_2$ given 2 agents with values drawn i.i.d.\ from $\topt{\Qud}_1^{\maxval'}$ are calculated as \begin{align*}
\mechap_{o_{\text{pm}}}(\topt{\Qud}_1^{\maxval'}) &= \frac{((2+\ln\maxval)\maxval-(1+\ln\maxval)e)}{\maxval}\\
\text{{\em LOT}}_2(\topt{\Qud}_1^{\maxval'}) &= 1 + \ln\maxval
\end{align*}
\end{lemma}

\noindent The following lemma states that the quantity $\text{lb}_{2,2}$ is strictly upper bounded by $\text{opt}_{2,2}$ and strictly lower bounded by $\text{opt}_{2,1}$.  Its proof is deferred to \Cref{s:twopieceironstrictbounds}.

\begin{lemma}
\label{lem:lb22islbversusopt22}
Given $\text{{\em opt}}_{2,1}$ and $\text{{\em opt}}_{2,2}$ resulting from the finite-weight Quadratics-versus-Uniforms dual blends (along with the rest of the local assumptions of this section), and $\text{{\em lb}}_{2,2}$ as defined in equation~\eqref{eqn:ressurpquadscalcdefoflb22}.  Then we have
\begin{equation*}
    \text{{\em opt}}_{2,2} > \text{{\em lb}}_{2,2} > \text{{\em opt}}_{2,1}
\end{equation*}
\end{lemma}

\noindent We now have the outline and justification to calculate the quantity $\text{lb}_{2,2}$ as a meaningful lower bound for $\text{opt}_{2,2}$ towards using the Blends Technique to prove a non-trivial residual surplus gap.  The calculation of $\text{lb}_{2,2}$ is a simple adjustment from $\text{opt}_{2,1}$ which runs the 2-lottery everywhere, versus, the quantity $\text{lb}_{2,2}$ is calculated from running the 2-lottery everywhere except with weight $o_{\text{pm}}=1$ it measures performance $\mechap_{o_{\text{pm}}}(\topt{\Qud}_{1}^{\maxval'})$ rather than $\text{LOT}_2(\topt{\Qud}_{1}^{\maxval'})$.  Therefore we have:
\begin{align}
    \nonumber \text{lb}_{2,2} &= \left[\text{opt}_{2,1}\right] + o_{\text{pm}}\left( \mechap_{o_{\text{pm}}}(\topt{\Qud}_{1}^{\maxval'}) -\text{LOT}_2(\topt{\Qud}_{1}^{\maxval'})\right)\\
    \nonumber
    &=  \left[4\maxval-3-\ln\maxval\right]+1\cdot\left(\frac{((2+\ln\maxval)\maxval-(1+\ln\maxval)e)}{\maxval} - (1 + \ln\maxval) \right)\\
    &=  \frac{4\maxval^2-2\maxval-\maxval\ln\maxval-e\ln\maxval-e}{\maxval}
\end{align}


\subsubsection{The Residual Surplus Gap as Lower Bound}
\label{s:ressurpgapfromratio}

For 2-agent, 1-item prior independent revenue auctions, we have now established a necessary residual surplus gap via blends:
\begin{equation}
\label{eqn:fourthousandthsmax}
    \frac{\text{opt}_{2,2}}{\text{opt}_{2,1}} >\frac{\text{lb}_{2,2}}{\text{opt}_{2,1}}= \frac{4\maxval^2-2\maxval-\maxval\ln\maxval-e\ln\maxval-e}{4\maxval^2-3\maxval-\maxval\ln\maxval} ,\quad\quad \lim_{\maxval\rightarrow\infty} \frac{\text{lb}_{2,2}}{\text{opt}_{2,1}} = 1
\end{equation}
\noindent where the limit calculation is obvious from observing highest-order terms (equivalently, from repeated application of l'Hopital's rule).  Evaluation in the limit makes clear that ratio-gaps {\em from our loose calculations} for any finite $\maxval$ are the result only of differences in lower order terms.

The applicability of the residual surplus gap as a prior independent lower bound is subject to the design problem's parameter $\scF$ to describe the adversary's allowable choice set of distributions.  However for residual surplus problems, there is precedent to allow the full set of distributions $\scF^{\text{all}}$.\footnote{\label{foot:ressurpeasierthanrev} In contrast to revenue, there exist prior independent mechanism design results for residual surplus that are constant approximations when allowing all distributions \citep{HR-14}.  Intuitively, residual surplus auctions are ``easier" than revenue because residual surplus virtual values are non-negative for all values.}  
Putting this together with equation~\eqref{eqn:fourthousandthsmax} and \Cref{thm:blendsbound} gives the following theorem to exhibit an approximation lower bound, which parallels \Cref{thm:finitequadsversusunifsrevpiauction} for revenue.


\begin{numberedtheorem}{\ref{thm:finitequadsversusunifsressurppiauction}}
\label{thm:finitequadsversusunifsressurppiauction2}
Given a single-item, 2-agent, truthful auction setting with a residual surplus objective and with agent values restricted to the space $[1,\maxval]$ for $\maxval\geq \RSMINH$.  For the class of quadratic distributions $\scF^{\text{{\em quad}}}$, the optimal prior independent approximation factor of any (truthful) mechanism is lower bounded as:
\begin{equation*}
    \piratio^{\scF^{\text{\em quad}}}_{\maxval} \geq
    \frac{\text{{\em opt}}_{2,2}}{\text{{\em opt}}_{2,1}}> \frac{4\maxval^2-2\maxval-\maxval\ln\maxval-e\ln\maxval-e}{4\maxval^2-3\maxval-\maxval\ln\maxval} = \lowb_{\maxval}^{\scF^{\text{\em quad}}}
\end{equation*}
\noindent The lower bound $\lowb^{\scF^{\text{{\em quad}}}}_{\maxval}\rightarrow1$ as $\maxval\rightarrow\infty$.  As an example bound: for $\maxval \in \mathbb{N}$, the maximum of $\lowb^{\scF^{\text{{\em quad}}}}_{\maxval}$ is achieved at $\maxval=\EXBRSWN$ with $\lowb^{\scF^{\text{\em quad}}}_{\EXBRSWN} \approx \EXBRS$.
\end{numberedtheorem}

\noindent As a corollary, our bound here holds for $\scF^{\text{all}}$.

\subsection{Supporting Work for Quadratics-versus-Uniforms Residual Surplus Gap}
\label{s:ressurpquadsvsunifssupporting}

This section provides material to support \Cref{s:expectedressurpquads}.  The presentation generally assumes its terms, assumptions, and context while only restating the most important definitions here. 

\begin{itemize}
    \item \Cref{s:ressurpcurvequad1} justifies the residual surplus curve of $\topt{\Qud}_1^{\maxval'}$, stated previously in equation~\eqref{eqn:ressurpcurveofquad1h}.
    \item \Cref{s:twopieceironmechandperformance} defines $\mecha_{o_{\text{pm}}}$ and calculates its residual surplus on $\topt{\Qud}_1^{\maxval'}$.  
\item \Cref{s:twopieceironstrictbounds} calculates the residual surplus of the 2-lottery on $\topt{\Qud}_1^{\maxval'}$ to show that $\mecha_{o_{\text{pm}}}$ has better performance, and concludes that we have strictly $\text{opt}_{2,2}>\text{lb}_{2,2}>\text{opt}_{2,1}$.
\item \Cref{s:ressurphardboundwhyuselowerbound} gives further supporting analysis for completeness, for example it explains the choice to use $\text{lb}_{2,2}$ rather than calculating $\text{opt}_{2,2}$ and describes the 
un-ironed residual surplus curve for $\topt{\Qud}_1^{\maxval'}$ and the design of $\mecha_{o_{\text{pm}}}$.
\end{itemize}


\noindent One critical assumption that we do repeat is $\maxval\geq \RSMINH$.  Copying equation~\eqref{eqn:ressurpquadscalcdefoflb22} for local reference, we have 
\begin{equation*}
    \label{eqn:ressurpquadscalcdefoflb22v2}
    \text{lb}_{2,2} =1\cdot \mechap_{o_{\text{pm}}}(\topt{\Qud}_{1}^{\maxval'}) \quad+\quad \left[\int_1^{\maxval} \frac{2}{z}\cdot \text{LOT}_2(\topt{\Qud}_{z}^{\maxval'})~dz \right]
\end{equation*}

\subsubsection{The Residual Surplus Curve for the Quadratic on $[1,\maxval]$}
\label{s:ressurpcurvequad1}

\label{page:quadraticsrscurve} \noindent This section explains the un-ironed residual surplus curve for $\topt{\Qud}_{z=1}^{\maxval'}$.  Recall, the CDF of the specific distribution $\topt{\Qud}_{z=1}^{\maxval'}$ is given by $\topt{\Qud}_1^{\maxval'}(x) = 1 - \sfrac{1}{x}$ on $x\in[1,\maxval]$, and $\topt{\Qud}_1^{\maxval'}(x) = 1$ for $x\geq \maxval$.  I.e., the CDF has a vertical line segment at $x =\maxval$ where it maps to the set-value $[1-\sfrac{1}{\maxval},1]$, because the distribution is top-truncated with a point mass at $\maxval$.  With explanation to follow, we restate the residual surplus curve:
\begin{equation*}
\rev_{\topt{\Qud}_1^{\maxval'}}(\quant) =
\begin{cases}
0 & \text{for}~\quant\in[0,\sfrac{1}{\maxval}]\\
\ln(\quant\cdot\maxval) & \text{for}~\quant\in[\sfrac{1}{\maxval},1)\\
[\ln \maxval,1+\ln\maxval] & \text{for}~\quant = 1
\end{cases}
\end{equation*}

\noindent This residual surplus curve $\rev_{\topt{\Qud}_1^{\maxval'}}$ 
is illustrated in \Cref{fig:twopieceiron}.  We proceed to justify this equation.  

The residual surplus curve $\rev_{\topt{\Qud}_1^{\maxval'}}(\cdot)$ is defined piece-wise including (a) a piece that is identically 0 for quantiles $\quant \in [0, \sfrac{1}{\maxval}]$ (from top-truncation); and (b) a piece that is a vertical line segment at $\quant = 1$ of length $z=1$ (from consideration of price-posting in $[0,1]$).  The lower end point of this vertical line segment is identified by the residual surplus curve at quantile $\quant = 1$ corresponding to the lower bound on the distribution's domain in value space, in this case value $\val = 1$.   We will show next that the height of the residual surplus curve corresponding to this point is in fact $\ln \maxval$.  The expected residual surplus from one agent value drawn from $\topt{\Qud}_1^{\maxval'}$ with a posted price of 1 is:
\begin{equation*}
    \text{AP}_1(\topt{\Qud}_1^{\maxval'})=\expecta_{\val\sim\topt{\Qud}_1^{\maxval'}}\left[ \val-1\right]=\left[\frac{1}{\maxval}\cdot\maxval + \int_1^{\maxval} \frac{1}{z^2}\cdot z~dz\right]-1 = \ln\maxval
\end{equation*}
\noindent Thus, the exact description of the vertical line segment is the set-valued output range of $[\ln\maxval,1+\ln\maxval]$ at input $\quant = 1$.  We now show that generally, the residual surplus curve for quantiles $\quant \in [\sfrac{1}{\maxval},1]$ is described by $\rev_{\topt{\Qud}_1^{\maxval'}}(\quant) = \ln (\quant\cdot \maxval)$.  (Combined with line-segment-pieces (a) and (b) which have already been explained, this completes the description of the residual surplus curve of $\topt{\Qud}_1^{\maxval'}$.)

Given the distribution $\topt{\Qud}_1^{\maxval'}$ and the residual surplus objective, as functions of value inputs $\val\in[1,\maxval]$, and then as functions of quantiles $\quant\in[\sfrac{1}{\maxval},1]$, the virtual value function and quantile/value functions are
\begin{align*}
     \vv^{\topt{\Qud}_1^{\maxval'}}(\val) &= \frac{1-\topt{\Qud}_1^{\maxval'}(\val)}{\topt{\qud}_1^{\maxval'}(\val)} = \frac{ 1 - (1-\sfrac{1}{\val})}{\sfrac{1}{\val^2}} = \val\\
    \quantf_{\topt{\Qud}_1^{\maxval'}}(\val) &= 1 - \topt{\Qud}_1^{\maxval'}(\val) = 1-(1-\sfrac{1}{\val}) = \sfrac{1}{\val}\\
    \vv^{\topt{\Qud}_1^{\maxval'}}(\quant) &= \sfrac{1}{\quant},\quad\quad \valf_{\topt{\Qud}_1^{\maxval'}}(\quant) = \sfrac{1}{\quant}
\end{align*}
\noindent Using the identity $\vv^{\dist}(\quant) = \rev'_{\dist}(\quant)$ from \Cref{fact:vv_rprime} and then integrating the function $\vv^{\topt{\Qud}_1^{\maxval'}}$ just given ($\int_{\sfrac{1}{\maxval}}^q\sfrac{1}{z}~dz$), we confirm the case for $\quant\in[\sfrac{1}{\maxval},1]$ of equation~\eqref{eqn:ressurpcurveofquad1h} that states: the residual surplus curve is described by $\rev_{\topt{\Qud}_1^{\maxval'}}(\quant) = \ln (\quant\cdot \maxval)$ on this sub-domain of quantile space.

\subsubsection{Definition and Residual Surplus of the ``Two-piece-iron" Mechanism}
\label{s:twopieceironmechandperformance}
\label{s:ressurpspecialmech}

The goal of this section is to define the special mechanism $\mecha_{o_{\text{pm}}}$ as needed by \Cref{s:expectedressurpquads} for the calcuation of $\text{lb}_{2,2}$, and calculate its expected residual surplus given 2 agents with values drawn i.i.d.\ from $\topt{\Qud}_1^{\maxval'}$ within the deferred proof of \Cref{lem:ressurpofmechopm}.  Note, the presentation is within the context of proving that the residual surplus $\mechap_{o_{\text{pm}}}(\topt{\Qud}_1^{\maxval'})$ is strictly more than the residual surplus $\text{LOT}_2(\topt{\Qud}_1^{\maxval'})$ of the 2-lottery.

For use in this section, we need to extend the definition of an {\em ironed residual surplus curve} to allow arbitrary ironing.\footnote{\label{foot:extendironeddef} The original definition for an ironed residual surplus curve was derived from the definition of an ironed revenue curve, specifically in the context of {\em optimal ironing}.  See page~\pageref{page:ironedcurvedef}.}\label{page:definearbitraryiron}  Let $\mathcal{Q}$ be any {\em possibly-non-optimal} set of non-overlapping ranges to be ironed (thus each element of $\mathcal{Q}$ is a subset of $[0,1]\in\reals$).  Given a distribution $\dist$, define $\bar{\rev}_{\dist,\mathcal{Q}}$ to be the residual surplus curve given base-distribution $\dist$ and its residual surplus curve $\rev_{\dist}$, and then accounting for $\mathcal{Q}$ as the given list of ironed regions.

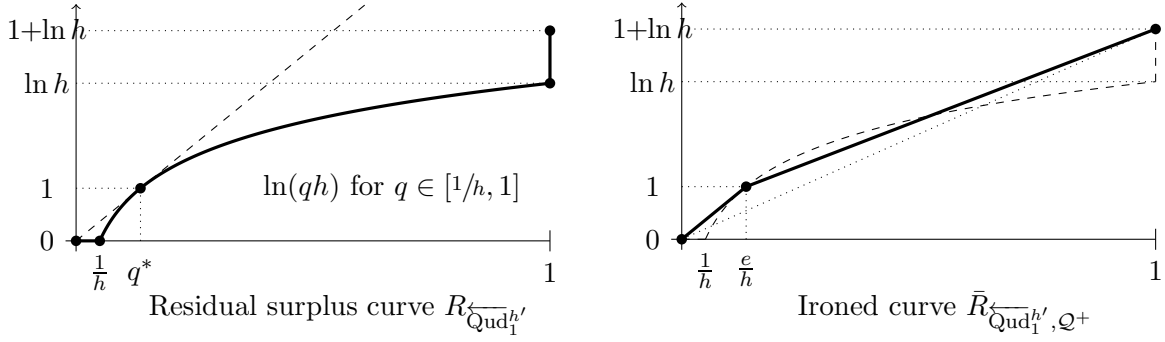
\begin{figure}[t]
\begin{flushleft}
\begin{minipage}[t]{0.48\textwidth}
\centering
\begin{tikzpicture}[scale = 0.7,pile/.style={->}]

\draw (-0.2,0) -- (9, 0);
\draw [pile] (0, -0.2) -- (0, 4.5);

\draw [very thick] (9, 2.995732) -- (9, 3.995732);
\draw [very thick] (0, 0) -- (0.45, 0);
\draw (9,-0.2) -- (9,0.2);
\draw [dotted] (1.223227, 0) -- (1.223227, 1);
\draw [dotted] (0,1) -- (1.223227, 1);
\draw [dotted] (0,2.995732) -- (9, 2.995732);
\draw [dotted] (0,3.995732) -- (9, 3.995732);

\draw [dashed] (0,0) -- (5.504521,4.5);

\fill[black] (0,0) circle (0.1cm);
\fill[black] (0.45,0) circle (0.1cm);
\fill[black] (9,2.995732) circle (0.1cm);
\fill[black] (9,3.995732) circle (0.1cm);
\fill[black] (1.223227,1) circle (0.1cm);

\begin{scope}[very thick]

\draw plot [smooth, tension=0.8] coordinates {
(0.450,0.000000)
(0.495,0.095310)
(0.540,0.182322)
(0.585,0.262364)
(0.630,0.336472)
(0.675,0.405465)
(0.720,0.470004)
(0.765,0.530628)
(0.810,0.587787)
(0.855,0.641854)
(0.900,0.693147)
(0.945,0.741937)
(0.990,0.788457)
(1.035,0.832909)
(1.080,0.875469)
(1.125,0.916291)
(1.170,0.955511)
(1.215,0.993252)
(1.260,1.029619)
(1.305,1.064711)
(1.350,1.098612)
(1.395,1.131402)
(1.440,1.163151)
(1.485,1.193922)
(1.530,1.223775)
(1.575,1.252763)
(1.620,1.280934)
(1.665,1.308333)
(1.710,1.335001)
(1.755,1.360977)
(1.800,1.386294)
(1.845,1.410987)
(1.890,1.435085)
(1.935,1.458615)
(1.980,1.481605)
(2.025,1.504077)
(2.070,1.526056)
(2.115,1.547563)
(2.160,1.568616)
(2.205,1.589235)
(2.250,1.609438)
(2.295,1.629241)
(2.340,1.648659)
(2.385,1.667707)
(2.430,1.686399)
(2.475,1.704748)
(2.520,1.722767)
(2.565,1.740466)
(2.610,1.757858)
(2.655,1.774952)
(2.700,1.791759)
(2.745,1.808289)
(2.790,1.824549)
(2.835,1.840550)
(2.880,1.856298)
(2.925,1.871802)
(2.970,1.887070)
(3.015,1.902108)
(3.060,1.916923)
(3.105,1.931521)
(3.150,1.945910)
(3.195,1.960095)
(3.240,1.974081)
(3.285,1.987874)
(3.330,2.001480)
(3.375,2.014903)
(3.420,2.028148)
(3.465,2.041220)
(3.510,2.054124)
(3.555,2.066863)
(3.600,2.079442)
(3.645,2.091864)
(3.690,2.104134)
(3.735,2.116256)
(3.780,2.128232)
(3.825,2.140066)
(3.870,2.151762)
(3.915,2.163323)
(3.960,2.174752)
(4.005,2.186051)
(4.050,2.197225)
(4.095,2.208274)
(4.140,2.219203)
(4.185,2.230014)
(4.230,2.240710)
(4.275,2.251292)
(4.320,2.261763)
(4.365,2.272126)
(4.410,2.282382)
(4.455,2.292535)
(4.500,2.302585)
(4.545,2.312535)
(4.590,2.322388)
(4.635,2.332144)
(4.680,2.341806)
(4.725,2.351375)
(4.770,2.360854)
(4.815,2.370244)
(4.860,2.379546)
(4.905,2.388763)
(4.950,2.397895)
(4.995,2.406945)
(5.040,2.415914)
(5.085,2.424803)
(5.130,2.433613)
(5.175,2.442347)
(5.220,2.451005)
(5.265,2.459589)
(5.310,2.468100)
(5.355,2.476538)
(5.400,2.484907)
(5.445,2.493205)
(5.490,2.501436)
(5.535,2.509599)
(5.580,2.517696)
(5.625,2.525729)
(5.670,2.533697)
(5.715,2.541602)
(5.760,2.549445)
(5.805,2.557227)
(5.850,2.564949)
(5.895,2.572612)
(5.940,2.580217)
(5.985,2.587764)
(6.030,2.595255)
(6.075,2.602690)
(6.120,2.610070)
(6.165,2.617396)
(6.210,2.624669)
(6.255,2.631889)
(6.300,2.639057)
(6.345,2.646175)
(6.390,2.653242)
(6.435,2.660260)
(6.480,2.667228)
(6.525,2.674149)
(6.570,2.681022)
(6.615,2.687847)
(6.660,2.694627)
(6.705,2.701361)
(6.750,2.708050)
(6.795,2.714695)
(6.840,2.721295)
(6.885,2.727853)
(6.930,2.734368)
(6.975,2.740840)
(7.020,2.747271)
(7.065,2.753661)
(7.110,2.760010)
(7.155,2.766319)
(7.200,2.772589)
(7.245,2.778819)
(7.290,2.785011)
(7.335,2.791165)
(7.380,2.797281)
(7.425,2.803360)
(7.470,2.809403)
(7.515,2.815409)
(7.560,2.821379)
(7.605,2.827314)
(7.650,2.833213)
(7.695,2.839078)
(7.740,2.844909)
(7.785,2.850707)
(7.830,2.856470)
(7.875,2.862201)
(7.920,2.867899)
(7.965,2.873565)
(8.010,2.879198)
(8.055,2.884801)
(8.100,2.890372)
(8.145,2.895912)
(8.190,2.901422)
(8.235,2.906901)
(8.280,2.912351)
(8.325,2.917771)
(8.370,2.923162)
(8.415,2.928524)
(8.460,2.933857)
(8.505,2.939162)
(8.550,2.944439)
(8.595,2.949688)
(8.640,2.954910)
(8.685,2.960105)
(8.730,2.965273)
(8.775,2.970414)
(8.820,2.975530)
(8.865,2.980619)
(8.910,2.985682)
(8.955,2.990720)
(9.000,2.995732)
};

\end{scope}

\draw (9, -0.6) node {$1$};
\draw (.45, -0.6) node {$\frac{1}{\maxval}$};
\draw  (1.223227, -0.6) node {$\quant^*$};

\draw (-0.55, 0) node {$0$};
\draw (-0.55, 2.995732) node {$\ln\maxval$};
\draw (-0.55, 3.995732) node {$1\hspace{-0.1cm}+\hspace{-0.1cm}\ln\maxval$};
\draw (-0.55, 1) node {$1$};

\draw (6, 1) node {$\ln(\quant\maxval)$ for $\quant\in[\sfrac{1}{\maxval},1]$};
\draw (5,-1.4) node { Residual surplus curve $\rev_{\topt{\Qud}_1^{\maxval'}}$};

\end{tikzpicture}
\end{minipage}
\begin{minipage}[t]{0.48\textwidth}
\centering
\begin{tikzpicture}[scale = 0.7,pile/.style={->}]

\draw (-0.2,0) -- (9, 0);
\draw [pile] (0, -0.2) -- (0, 4.5);

\draw [dashed] (9, 2.995732) -- (9, 3.995732);
\draw [dashed] (0, 0) -- (0.45, 0);
\draw (9,-0.2) -- (9,0.2);
\draw [dotted] (1.223227, 0) -- (1.223227, 1);
\draw [dotted] (0,1) -- (1.223227, 1);
\draw [dotted] (0,2.995732) -- (9, 2.995732);
\draw [dotted] (0,3.995732) -- (9, 3.995732);

\fill[black] (0,0) circle (0.1cm);
\fill[black] (9,3.995732) circle (0.1cm);
\fill[black] (1.223227,1) circle (0.1cm);

\draw [dotted] (0,0) -- (9, 3.995732);

\draw [very thick] (0,0) -- (1.223227,1);
\draw [very thick] (1.223227,1) -- (9, 3.995732);

\begin{scope}[dashed]

\draw plot [smooth, tension=0.8] coordinates {
(0.450,0.000000)
(0.495,0.095310)
(0.540,0.182322)
(0.585,0.262364)
(0.630,0.336472)
(0.675,0.405465)
(0.720,0.470004)
(0.765,0.530628)
(0.810,0.587787)
(0.855,0.641854)
(0.900,0.693147)
(0.945,0.741937)
(0.990,0.788457)
(1.035,0.832909)
(1.080,0.875469)
(1.125,0.916291)
(1.170,0.955511)
(1.215,0.993252)
(1.260,1.029619)
(1.305,1.064711)
(1.350,1.098612)
(1.395,1.131402)
(1.440,1.163151)
(1.485,1.193922)
(1.530,1.223775)
(1.575,1.252763)
(1.620,1.280934)
(1.665,1.308333)
(1.710,1.335001)
(1.755,1.360977)
(1.800,1.386294)
(1.845,1.410987)
(1.890,1.435085)
(1.935,1.458615)
(1.980,1.481605)
(2.025,1.504077)
(2.070,1.526056)
(2.115,1.547563)
(2.160,1.568616)
(2.205,1.589235)
(2.250,1.609438)
(2.295,1.629241)
(2.340,1.648659)
(2.385,1.667707)
(2.430,1.686399)
(2.475,1.704748)
(2.520,1.722767)
(2.565,1.740466)
(2.610,1.757858)
(2.655,1.774952)
(2.700,1.791759)
(2.745,1.808289)
(2.790,1.824549)
(2.835,1.840550)
(2.880,1.856298)
(2.925,1.871802)
(2.970,1.887070)
(3.015,1.902108)
(3.060,1.916923)
(3.105,1.931521)
(3.150,1.945910)
(3.195,1.960095)
(3.240,1.974081)
(3.285,1.987874)
(3.330,2.001480)
(3.375,2.014903)
(3.420,2.028148)
(3.465,2.041220)
(3.510,2.054124)
(3.555,2.066863)
(3.600,2.079442)
(3.645,2.091864)
(3.690,2.104134)
(3.735,2.116256)
(3.780,2.128232)
(3.825,2.140066)
(3.870,2.151762)
(3.915,2.163323)
(3.960,2.174752)
(4.005,2.186051)
(4.050,2.197225)
(4.095,2.208274)
(4.140,2.219203)
(4.185,2.230014)
(4.230,2.240710)
(4.275,2.251292)
(4.320,2.261763)
(4.365,2.272126)
(4.410,2.282382)
(4.455,2.292535)
(4.500,2.302585)
(4.545,2.312535)
(4.590,2.322388)
(4.635,2.332144)
(4.680,2.341806)
(4.725,2.351375)
(4.770,2.360854)
(4.815,2.370244)
(4.860,2.379546)
(4.905,2.388763)
(4.950,2.397895)
(4.995,2.406945)
(5.040,2.415914)
(5.085,2.424803)
(5.130,2.433613)
(5.175,2.442347)
(5.220,2.451005)
(5.265,2.459589)
(5.310,2.468100)
(5.355,2.476538)
(5.400,2.484907)
(5.445,2.493205)
(5.490,2.501436)
(5.535,2.509599)
(5.580,2.517696)
(5.625,2.525729)
(5.670,2.533697)
(5.715,2.541602)
(5.760,2.549445)
(5.805,2.557227)
(5.850,2.564949)
(5.895,2.572612)
(5.940,2.580217)
(5.985,2.587764)
(6.030,2.595255)
(6.075,2.602690)
(6.120,2.610070)
(6.165,2.617396)
(6.210,2.624669)
(6.255,2.631889)
(6.300,2.639057)
(6.345,2.646175)
(6.390,2.653242)
(6.435,2.660260)
(6.480,2.667228)
(6.525,2.674149)
(6.570,2.681022)
(6.615,2.687847)
(6.660,2.694627)
(6.705,2.701361)
(6.750,2.708050)
(6.795,2.714695)
(6.840,2.721295)
(6.885,2.727853)
(6.930,2.734368)
(6.975,2.740840)
(7.020,2.747271)
(7.065,2.753661)
(7.110,2.760010)
(7.155,2.766319)
(7.200,2.772589)
(7.245,2.778819)
(7.290,2.785011)
(7.335,2.791165)
(7.380,2.797281)
(7.425,2.803360)
(7.470,2.809403)
(7.515,2.815409)
(7.560,2.821379)
(7.605,2.827314)
(7.650,2.833213)
(7.695,2.839078)
(7.740,2.844909)
(7.785,2.850707)
(7.830,2.856470)
(7.875,2.862201)
(7.920,2.867899)
(7.965,2.873565)
(8.010,2.879198)
(8.055,2.884801)
(8.100,2.890372)
(8.145,2.895912)
(8.190,2.901422)
(8.235,2.906901)
(8.280,2.912351)
(8.325,2.917771)
(8.370,2.923162)
(8.415,2.928524)
(8.460,2.933857)
(8.505,2.939162)
(8.550,2.944439)
(8.595,2.949688)
(8.640,2.954910)
(8.685,2.960105)
(8.730,2.965273)
(8.775,2.970414)
(8.820,2.975530)
(8.865,2.980619)
(8.910,2.985682)
(8.955,2.990720)
(9.000,2.995732)
};

\end{scope}

\draw (9, -0.6) node {$1$};
\draw (.45, -0.6) node {$\frac{1}{\maxval}$};
\draw  (1.223227, -0.6) node {$\frac{e}{\maxval}$};

\draw (-0.55, 0) node {$0$};
\draw (-0.55, 2.995732) node {$\ln\maxval$};
\draw (-0.55, 3.995732) node {$1\hspace{-0.1cm}+\hspace{-0.1cm}\ln\maxval$};
\draw (-0.55, 1) node {$1$};

\draw (5,-1.4) node { Ironed curve $\bar{\rev}_{\topt{\Qud}_1^{\maxval'},\mathcal{Q}^+}$};

\end{tikzpicture}
\end{minipage}
\end{flushleft}
\caption{
\label{fig:twopieceiron}
\iffinalprint
\else
We illustrate the effect of ironing on our truncated quadratic distribution (defined by equation~\eqref{eqn:ressurpcurveofquad1h} below), according to $\mecha_{o_{\text{pm}}}$, with un-ironed on the left-hand side and ironed on the right-hand side.  The example is graphed for scale $\maxval=20$.  Note: $\quant^*=\sfrac{e}{\maxval}$; the left figure indicates that ironing the interval $\quant\in[0,\sfrac{e}{h}]$ is tangent to the original curve and therefore optimal on this region; the right figure makes clear that the point $(\sfrac{e}{\maxval},1)$ is above the line that represents ironing everywhere (for sufficiently large $\maxval$), but is otherwise suboptimal.
\fi
}
\end{figure}

Call the mechanism $\mecha_{o_{\text{pm}}}$ the {\em two-piece-iron} mechanism.  Its definition depends on a critical quantile $\quant^*=\sfrac{e}{\maxval}$.  Specifically motivated by $\topt{\Qud}_1^{\maxval'}$, the mechanism irons two regions: (1) large values below and (2) small values above this quantile.  Ultimately, the mechanism $\mecha_{o_{\text{pm}}}$ runs an {\em ironed} second price auction on the two inferred types (one common type for each ironed region).  For illustration of $\mecha_{o_{\text{pm}}}$ applied to $\rev_{\topt{\Qud}_1^{\maxval'}}$, see \Cref{fig:twopieceiron}.
Formally, define the following sets of ironed ranges which will be used by our subsequent analysis:
\begin{itemize}
    \item  $\mathcal{Q}^* = \{[0,\quant^*=\sfrac{e}{\maxval}]\}$; corresponding to value-range $[\sfrac{\maxval}{e},\maxval]$;
    \item $\mathcal{Q}^+ = \{[0,\quant^*],~[\quant^*,1] \}$; corresponding to value-ranges $[\sfrac{\maxval}{e},\maxval]$ and $[0,\sfrac{\maxval}{e}]$ (where identifying the lower bound of the second value range to be 0 is a necessary distinction because $\rev_{\topt{\Qud}_1^{\maxval'}}$ is set-valued at $\quant = 1$);
    \item $\mathcal{Q}^1 = \{[0,1]\}$ corresponding to value-range $[0,\maxval]$ (which is effectively the lottery).
\end{itemize}

\begin{definition}
\label{def:ressurpspecialmech}
Define the {\em two-piece-iron} mechanism $\mecha_{o_{\text{pm}}}$ for $n=2$ agents to be the {\em ironed second price auction} which respects the ironed ranges $\mathcal{Q}^+ = \{[0,\quant^*],~[\quant^*,1] \}$.

Equivalently, $\mecha_{o_{\text{pm}}}$ irons the regions of value space $[\sfrac{maxval}{e},\maxval]$ and $[0,\sfrac{\maxval}{e}]$ and runs the second price auction on these two inferred types.
\end{definition}

\noindent Thus, when agent values are drawn i.i.d.\ from $\topt{\Qud}_1^{\maxval'}$, the residual surplus $\mecha_{o_{\text{pm}}}(\topt{\Qud}_1^{\maxval'})$ may be calculated using the ironed residual surplus curve $\bar{\rev}_{\topt{\Qud}_1^{\maxval'},\mathcal{Q}^+}(\cdot)$.  
(See \Cref{fig:twopieceiron}.)

In fact for $\maxval\geq \RSMINH$, ironing the region $[0,\quant^*=\sfrac{e}{\maxval}]$ is optimal given the underlying distribution $\topt{\Qud}_1^{\maxval'}$; and ironing the region $[\sfrac{e}{\maxval},1]$ is strictly suboptimal (for intuition for this, see \Cref{fig:twopieceiron-z-vs-h}).

We conclude this section with the deferred proof of \Cref{lem:ressurpofmechopm}, which depends on \Cref{lem:dry-15-extension} below as an extension of \Cref{lem:DRY-15} \citep{DRY-15}.

\begin{numberedlemma}{\ref{lem:ressurpofmechopm}}
\label{lem:ressurpofmechopm2}
The residual surplus of mechanisms $\mecha_{o_{\text{pm}}}$ and $\text{LOT}_2$ given 2 agents with values drawn i.i.d.\ from $\topt{\Qud}_1^{\maxval'}$ are calculated as \begin{align*}
\mechap_{o_{\text{pm}}}(\topt{\Qud}_1^{\maxval'}) &= \frac{((2+\ln\maxval)\maxval-(1+\ln\maxval)e)}{\maxval}\\
\text{{\em LOT}}_2(\topt{\Qud}_1^{\maxval'}) &= 1 + \ln\maxval
\end{align*}
\end{numberedlemma}
\begin{proof}
Using \Cref{lem:dry-15-extension} below (which extends \cite{DRY-15} to allow ironing and any auction objective, in our case residual surplus) and the definition of $\mecha_{o_{\text{pm}}}$, the residual surplus $\mechap_{o_{\text{pm}}}(\topt{\Qud}_1^{\maxval'})$ is calculated as twice the area under the ironed residual surplus curve $\bar{\rev}_{\topt{\Qud}_1^{\maxval'} ,\mathcal{Q}^+}$.  This area is calculated from
\begin{align*}
    \text{area under}~\bar{\rev}_{\topt{\Qud}_1^{\maxval'},\mathcal{Q}^+} &= \text{area in quantile range}~[0,\sfrac{e}{\maxval}] + \text{area in quantile range}~[\sfrac{e}{\maxval},1]\\
    &= \sfrac{1}{2}\cdot \bar{\rev}_{\topt{\Qud}_1^{\maxval'},\mathcal{Q}^+}(\sfrac{e}{\maxval})\cdot\left(\frac{e}{\maxval}-0\right)\\
    &\quad+ \frac{1}{2}\cdot\left( \bar{\rev}_{\topt{\Qud}_1^{\maxval'},\mathcal{Q}^+}(\sfrac{e}{\maxval})+\bar{\rev}_{\topt{\Qud}_1^{\maxval'},\mathcal{Q}^+}(1)\right)\left(1-\frac{e}{\maxval}\right)\\
    &= \frac{1}{2}\cdot 1\cdot\frac{e}{\maxval}+\frac{1}{2}\cdot\left(1+(1+\ln\maxval) \right) \cdot\frac{\maxval-e}{\maxval}\\
    &= \frac{1}{2}\cdot\frac{(2+\ln\maxval)\maxval-(1+\ln\maxval)e }{\maxval}
\end{align*}
\noindent where the residual surplus at the end points of the ironed ranges -- namely, quantiles $\quant = \sfrac{e}{\maxval}$ and $\quant = 1$ -- are from the definition of $\rev_{\topt{\Qud}_1^{\maxval'}}$ (equation~\eqref{eqn:ressurpcurveofquad1h} earlier).  Therefore the residual surplus of the mechanism is $\mechap_{o_{\text{pm}}}(\topt{\Qud}_1^{\maxval'})=\sfrac{((2+\ln\maxval)\maxval-(1+\ln\maxval)e)}{\maxval}$.

Using \Cref{lem:dry-15-extension} and equation~\eqref{eqn:ressurpcurveofquad1h} again -- this time applied to ironed revenue curve $\bar{\rev}_{\topt{\Qud}_1^{\maxval'},\mathcal{Q}^1}$, i.e., with respect to the lottery's ironing $\mathcal{Q}^1$ -- the residual surplus of the 2-lottery is  $\text{LOT}_2(\topt{\Qud}_1^{\maxval'})= 1+\ln\maxval$.
\end{proof}

\noindent For completeness, we prove the extension of \Cref{lem:DRY-15} \citep{DRY-15} to apply both (a) for an arbitrary auction objective, and (b) to allow arbitrary ironing.  In the case of ironing, the SPA must be interpreted as treating each ironed range as a single value space type -- it allocates all agents in an ironed range uniformly.\footnote{\label{foot:useironfordry} Recall, this type of treatment is a necessary condition to apply the technique of ironing -- see the introduction of ironing in discussion on page~\pageref{page:ironintro} and its conditional use in \Cref{thm:myebayesopt}.}   Define this mechanism as the {\em Ironed Second Price Auction}.

Let $\mathcal{Q}$ be a set of ironed ranges (in quantile space; as defined on page~\pageref{page:definearbitraryiron}), and let $\bar{\vv}^{\dist,\mathcal{Q}}$ be the ironed virtual value function given an underlying distribution $\dist$ that has been ironed on ranges according to $\mathcal{Q}$.
\begin{lemma}
\label{lem:dry-15-extension}
  In i.i.d.\ two-agent single-item settings given distribution $\dist$, {\em for any auction objective} let $\rev_{\dist}$ be the performance curve in quantile space 
  and $\bar{\rev}_{\dist,\mathcal{Q}}$ be an ironed performance curve given a set of ironed ranges $\mathcal{Q}$.
  
  The expected performance of the ISPA subject to $\bar{\rev}_{\dist,\mathcal{Q}}$ -- assuming uniform allocation to agents within each ironed range as if the range was one type -- is twice the area under the curve $\bar{\rev}_{\dist,\mathcal{Q}}$.
\end{lemma}
\begin{proof}
We note the following up front.  Without loss of generality, our 2 agents have ordered values $\vali[(1)]\geq\vali[(2)]$, equivalently, ordered quantiles $\quant_{(1)}\leq\quant_{(2)}$.  The ISPA mechanism of the statement is symmetric.

The technique of this proof is to sum up the performance of the ISPA mechanism (for arbitrary ironing) by calculating expected performance over the distribution of the smaller quantile-order-statistic $\quant_{(1)}$.  To outline, we: identify this conditional performance as a function of virtual value; and then insert this quantity into the existing proof of \Cref{lem:DRY-15} \citep{DRY-15}.

For agents labeled according to order statistic $i$, define $\alloci[(i)] =\alloci[(i)]^{\text{ISPA}}(\valf_{\dist}(\quant_{(i)}), \valf_{\dist}(\quant_{(j\neq i)}))$.  Define $\mathcal{R}(\quant_1)$ to be the expectation of the winning agent's virtual value conditioned on the smaller quantile being $\quant_1$.  (Note, the winner is not necessarily the agent $i=1$.)  Thus, we have:
\begin{equation*}
\mathcal{R}(\quant_1) = \mathbf{E}_{\quant_2\sim\Ud_{\quant_1,1}}\left[\alloci[(1)]\cdot \bar{\vv}^{\dist,\mathcal{Q}}(\quant_1)+\alloci[(2)]\cdot\bar{\vv}^{\dist,\mathcal{Q}}(\quant_2) \right] = \bar{\vv}^{\dist,\mathcal{Q}}(\quant_1)=\bar{\rev}'_{\dist,\mathcal{Q}}(\quant_1)
\end{equation*}
\noindent We get the second equality here because the following holds for all inputs $\quant_1$ into function $\mathcal{R}$: {\em pointwise within the expectation}: either $\alloci[(1)]=1$, or otherwise $\alloci[(1)]+\alloci[(2)]=1$ and $\bar{\vv}^{\dist,\mathcal{Q}}(\quant_1)=\bar{\vv}^{\dist,\mathcal{Q}}(\quant_2)$.  In any case, $\mathbf{E}_{\quant_2\sim\Ud_{\quant_1,1}}\left[\alloci[\{1\}]\cdot \bar{\vv}^{\dist,\mathcal{Q}}(\quant_1)+\alloci[\{2\}]\cdot\bar{\vv}^{\dist,\mathcal{Q}}(\quant_2) \right] = \bar{\vv}^{\dist,\mathcal{Q}}(\quant_1)$.  We can now effectively implement the proof of \Cref{lem:DRY-15} which did not accommodate ironing and which was stated for the specific objective of revenue.

Let $\text{osd}(\quant) = 2(1-\quant)$ be the density function of the smallest order-statistic $\quant_{(1)}$ out of 2 agents' quantiles drawn i.i.d.\ from $\Ud_{0,1}$.  (Note, for simplicity, we dropped all parameters from the distribution name `osd.')

We are now prepared to evaluate $\text{ISPA}(\dist)$ using $\bar{\mathcal{R}}_{\dist,\mathcal{Q}}$ and $\text{osd}$:
\begin{align*}
\text{ISPA}(\dist) &= \int_0^1 \text{osd}(\quant)\cdot \bar{\rev}'_{\dist,\mathcal{Q}}(\quant)~d\quant\\
&= \int_0^1 2\cdot 1\cdot (1-\quant)\cdot \bar{\rev}'_{\dist,\mathcal{Q}}(\quant)~d\quant\\
&= 2\cdot\bar{\rev}_{\dist,\mathcal{Q}}(1)-2\cdot\int_0^1\quant\cdot\bar{\rev}'_{\dist,\mathcal{Q}}(\quant)~d\quant\\
&= 2\cdot\bar{\rev}_{\dist,\mathcal{Q}}(1)-2\cdot\left[\quant\cdot\bar{\rev}_{\dist,\mathcal{Q}}(\quant) \right]_0^1+2\cdot \int_0^1 \bar{\rev}_{\dist,\mathcal{Q}}(\quant)~d\quant = 2\cdot \int_0^1 \bar{\rev}_{\dist,\mathcal{Q}}(\quant)~d\quant\quad \qedhere
\end{align*}
\end{proof}

\subsubsection{The ``Two-piece-iron" Mechanism is Sufficient for a Lower Bound}
\label{s:twopieceironstrictbounds}

In this section we show that the residual surplus  $\mechap_{o_{\text{pm}}}(\topt{\Qud}_1^{\maxval'}) = \sfrac{((2+\ln\maxval)\maxval-(1+\ln\maxval)e)}{\maxval}$ of \Cref{lem:ressurpofmechopm2} is strictly worse than optimal and strictly better than the 2-lottery given $\maxval\geq\RSMINH$.  The main goal is to give the deferred proof of \Cref{lem:lb22islbversusopt22v2} which states that $\text{opt}_{2,2}>\text{lb}_{2,2}>\text{opt}_{2,1}$.

\begin{lemma}
\label{lem:Moflb22islbversusLOTofopt22}
Given $\maxval\geq \RSMINH$.  The mechanism $\mecha_{o_{\text{pm}}}$ is strictly sub-optimal: $\OPT_{\topt{\Qud}_1^{\maxval'}}(\topt{\Qud}_1^{\maxval'}) > \mechap_{o_{\text{pm}}}(\topt{\Qud}_1^{\maxval'})$; and the mechanism $\mecha_{o_{\text{pm}}}$ strictly dominates the lottery: $\mechap_{o_{\text{pm}}}(\topt{\Qud}_1^{\maxval'}) > \text{{\em LOT}}_2(\topt{\Qud}_1^{\maxval'})$.
\end{lemma}
\begin{proof}
First, we prove the lower bound on $\mechap_{o_{\text{pm}}}(\topt{\Qud}_1^{\maxval'})$ in the lemma statement.  We use $\mechap_{o_{\text{pm}}}(\topt{\Qud}_1^{\maxval'}) = \sfrac{((2+\ln\maxval)\maxval-(1+\ln\maxval)e)}{\maxval}$ and $\text{{\em LOT}}_2(\topt{\Qud}_1^{\maxval'}) = 1 + \ln\maxval$ from \Cref{lem:ressurpofmechopm2}.  We reduce the condition that the difference is positive:
\begin{align}
\nonumber \mechap_{o_{\text{pm}}}(\topt{\Qud}_1^{\maxval'}) - \text{LOT}_2(\topt{\Qud}_1^{\maxval'}) &= \frac{(2+\ln\maxval)\maxval-(1+\ln\maxval)e}{\maxval} - (1+\ln\maxval)\\
\nonumber
    &= \frac{\maxval-(1+\ln\maxval)e}{\maxval}>0\\
    \label{eqn:ressurpexampleslopecomparison}
    \Leftrightarrow\quad &\quad~\maxval-(1+\ln\maxval)e > 0
\end{align}
\noindent Treating the left-hand side of the inequality in line~\eqref{eqn:ressurpexampleslopecomparison} as a function of $\maxval$, it is negative and decreasing for $\maxval\in[1,e]$, it is increasing for all $\maxval > e$, and it has a 0 within the range $\maxval\in[8.55,8.56]$ (and then is positive for $\maxval>8.56$ because it is increasing).  Therefore the lower bound $\maxval\geq \RSMINH$ is sufficient for the lowerbound on $\mechap_{o_{\text{pm}}}(\topt{\Qud}_1^{\maxval'})$.

For the upper bound in the lemma statement, the analysis and discussion surrounding equation~\eqref{eqn:ressurpoptslopelargequant} below in \Cref{s:ressurphardboundwhyuselowerbound} are sufficient to show that the ironing by $\mecha_{o_{\text{pm}}}$ of the range $[\sfrac{e}{\maxval},1]$ is strictly suboptimal (given $\maxval\geq\RSMINH$, which infers the lower end point of this range is upper bounded as $\sfrac{e}{\maxval}<0.31$).  It is dominated specifically in comparison to ironing the quantile-space upward-closed range $[\quant^{\&},1]$ for the optimal value of $\quant^{\&}\in[0.31,0.32]$ defined and proved in \Cref{s:ressurphardboundwhyuselowerbound}.
\end{proof}

\begin{numberedlemma}{\ref{lem:lb22islbversusopt22}}
\label{lem:lb22islbversusopt22v2}
Given $\text{{\em opt}}_{2,1}$ and $\text{{\em opt}}_{2,2}$ resulting from the finite-weight Quadratics-versus-Uniforms dual blends (along with the rest of the local assumptions of this section), and $\text{{\em lb}}_{2,2}$ as defined in equation~\eqref{eqn:ressurpquadscalcdefoflb22}. Then we have
\begin{equation*}
    \text{{\em opt}}_{2,2} > \text{{\em lb}}_{2,2} > \text{{\em opt}}_{2,1}
\end{equation*}
\end{numberedlemma}
\begin{proof}
\Cref{lem:Moflb22islbversusLOTofopt22} states that  if $\maxval\geq \RSMINH$, then (a) $\OPT_{\topt{\Qud}_1^{\maxval'}}(\topt{\Qud}_1^{\maxval'}) > \mechap_{o_{\text{pm}}}(\topt{\Qud}_1^{\maxval'})$ and (b) $\mechap_{o_{\text{pm}}}(\topt{\Qud}_1^{\maxval'}) > \text{LOT}_2(\topt{\Qud}_1^{\maxval'})$.  Following directly from these and from definitions we have:
\begin{align*}
    &&\text{opt}_{2,2} &=& o_{\text{pm}}\cdot\OPT_{\topt{\Qud}_1^{\maxval'}}(\topt{\Qud}_1^{\maxval'})&\quad&& + \quad \left[\int_1^{\maxval} o_z\cdot \OPT_{\topt{\Qud}_1^{\maxval'}}(\topt{\Qud}_{z}^{\maxval'}) \right] \\
    >&&\text{lb}_{2,2} &=& o_{\text{pm}}\cdot \mechap_{o_{\text{pm}}}(\topt{\Qud}_{1}^{\maxval'}) &\quad&&+\quad \left[\int_1^{\maxval} o_z\cdot \text{LOT}_2(\topt{\Qud}_{z}^{\maxval'}) \right]\\
    \nonumber
    >&& \text{opt}_{2,1}& =&o_{\text{pm}}\cdot \text{LOT}_2(\topt{\Qud}_{1}^{\maxval'}) &\quad&&+\quad \left[\int_1^{\maxval} o_z\cdot \text{LOT}_2(\topt{\Qud}_{z}^{\maxval'}) \right]
\end{align*}
\noindent where the definition of $\text{opt}_{2,1}$ may use the weights $\boldsymbol{o}$ and the distributions in $\blendi[2]$ rather than its original definition which respectively used $\boldsymbol{\omega}$ and $\blendi[1]$. This last point holds because it runs the constant lottery mechanism on all inputs anyway and $\delta_1^2=g=\delta_2^2$.
\end{proof}

\subsubsection{Complexities and Technicals of the Quadratics Residual Surplus Blend}
\label{s:ressurphardboundwhyuselowerbound}

The definition of $\mecha_{o_{\text{pm}}}$ was given in \Cref{def:ressurpspecialmech}.  For completeness, here we build up the motivation for it -- effectively reverse-engineering it to be sufficient for the residual surplus gap which is our goal (of equation~\eqref{eqn:fourthousandthsmax} of \Cref{s:ressurpgapfromratio}).  
Here is an outline of this section:
\begin{enumerate}
    \item identify technical difficulties of residual surplus curves $\rev_{\topt{\Qud}_z^{\maxval'}}(\cdot)$ for arbitrary $z$ and motivate the relaxation to $\text{lb}_{2,2}$ and the assumption of $\maxval\geq \RSMINH$ for simplicity;
    \item given $\maxval\geq\RSMINH$, determine that $[0,\sfrac{e}{\maxval}]$ is an element of the set of optimal ranges to iron, by analyzing slopes of possible quantile-downward-closed ironed ranges;
    \item from slopes of respective ironed ranges of $\mecha_{o_{\text{pm}}}$ and $\text{LOT}_2$, re-confirm the statement of \Cref{lem:Moflb22islbversusLOTofopt22} with dependence on equation~\eqref{eqn:ressurpexampleslopecomparison} in its proof;
    \item identify the optimal quantile-downward-closed range to iron, which is independent of $\maxval$.
\end{enumerate}

\paragraph{(1) Explanation of the choice to simplify from $\text{opt}_{2,2}$ to $\text{lb}_{2,2}$ and $\maxval\geq\RSMINH$.}  This is a discussion of the difficulties of revenue curves for the class of Quadratics $\topt{\Qud}_z^{\maxval'}$ (with positive weight in $\blendi[2]$).  Generalizing equation~\eqref{eqn:ressurpcurveofquad1h}, residual surplus curves for Quadratics and arbitrary $z$ are described by:

\begin{equation}
\label{eqn:ressurpcurveofquadzh}
\rev_{\topt{\Qud}_z^{\maxval'}}(\quant) =
\begin{cases}
0 & \text{for}~\quant\in[0,\sfrac{z}{\maxval}]\\
z\ln(\quant\cdot\frac{\maxval}{z}) & \text{for}~\quant\in[\sfrac{z}{\maxval},1)\\
[z\ln\frac{ \maxval}{z},z+z\ln\frac{\maxval}{z}] & \text{for}~\quant = 1
\end{cases}
\end{equation}

\noindent The first challenge is that for ``large" $z\rightarrow\maxval$, the optimal mechanism for $\topt{\Qud}_z^{\maxval'}$ is the lottery.  There is a threshold for $z$ above which this becomes true (see point (4) below and also \Cref{fig:twopieceiron-z-vs-h} which illustrates the threshold-change in the ironing structure of the residual surplus curve).

We greatly simplify this complication as follows: with an assumption of $\maxval\geq \RSMINH$, then for $z=1$, the lottery mechanism for $\topt{\Qud}_1^{\maxval'}$ is strictly not optimal.  The calculation of the lower bound quantity $\text{lb}_{2,2}$ uses the performance of the lottery on all other distributions, even though the lottery is sub-optimal for many of these distributions.  Critically however, for ``small" $z$ -- and specifically for $z=1$ where there is a point mass $o_{\text{pm}}$, and relying on our assumption of ``large" $\maxval\geq \RSMINH$ -- the optimal mechanism in response to distribution $\topt{\Qud}_z^{\maxval'}$ will \underline{not} use the lottery to iron the entire region of values, and rather, a distinct mechanism is strictly preferred. 

The calculation of expected performance for the optimal mechanism for $\topt{\Qud}_z^{\maxval'}$ is itself complicated.  To simplify, we relax the optimal mechanism to $\mecha_{o_{\text{opm}}}$ which irons on just two regions, an optimal region over small quantiles and ``all other large" quantiles (see \Cref{fig:twopieceiron} for illustration).

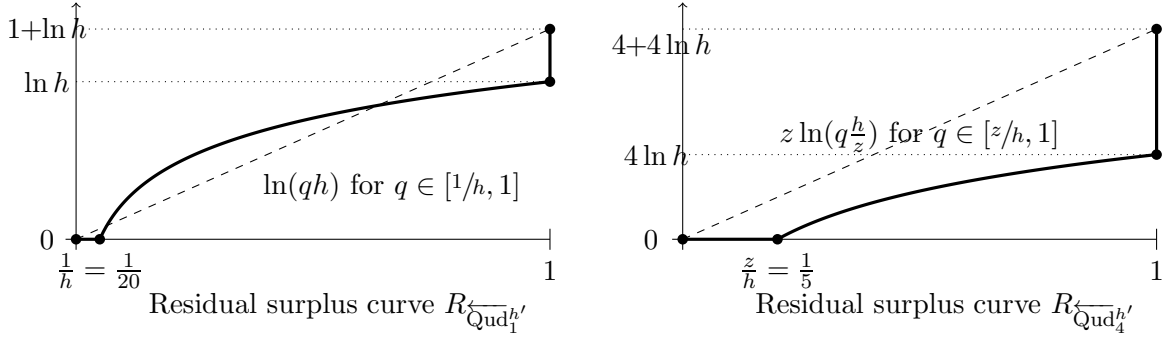
\begin{figure}[t]
\begin{flushleft}
\begin{minipage}[t]{0.48\textwidth}
\centering
\begin{tikzpicture}[scale = 0.7,pile/.style={->}]

\draw (-0.2,0) -- (9, 0);
\draw [pile] (0, -0.2) -- (0, 4.5);

\draw [very thick] (9, 2.995732) -- (9, 3.995732);
\draw [very thick] (0, 0) -- (0.45, 0);
\draw (9,-0.2) -- (9,0.2);
\draw [dotted] (0,2.995732) -- (9, 2.995732);
\draw [dotted] (0,3.995732) -- (9, 3.995732);

\draw [dashed] (0,0) -- (9, 3.995732);

\fill[black] (0,0) circle (0.1cm);
\fill[black] (0.45,0) circle (0.1cm);
\fill[black] (9,2.995732) circle (0.1cm);
\fill[black] (9,3.995732) circle (0.1cm);

\begin{scope}[very thick]

\draw plot [smooth, tension=0.8] coordinates {
(0.450,0.000000)
(0.495,0.095310)
(0.540,0.182322)
(0.585,0.262364)
(0.630,0.336472)
(0.675,0.405465)
(0.720,0.470004)
(0.765,0.530628)
(0.810,0.587787)
(0.855,0.641854)
(0.900,0.693147)
(0.945,0.741937)
(0.990,0.788457)
(1.035,0.832909)
(1.080,0.875469)
(1.125,0.916291)
(1.170,0.955511)
(1.215,0.993252)
(1.260,1.029619)
(1.305,1.064711)
(1.350,1.098612)
(1.395,1.131402)
(1.440,1.163151)
(1.485,1.193922)
(1.530,1.223775)
(1.575,1.252763)
(1.620,1.280934)
(1.665,1.308333)
(1.710,1.335001)
(1.755,1.360977)
(1.800,1.386294)
(1.845,1.410987)
(1.890,1.435085)
(1.935,1.458615)
(1.980,1.481605)
(2.025,1.504077)
(2.070,1.526056)
(2.115,1.547563)
(2.160,1.568616)
(2.205,1.589235)
(2.250,1.609438)
(2.295,1.629241)
(2.340,1.648659)
(2.385,1.667707)
(2.430,1.686399)
(2.475,1.704748)
(2.520,1.722767)
(2.565,1.740466)
(2.610,1.757858)
(2.655,1.774952)
(2.700,1.791759)
(2.745,1.808289)
(2.790,1.824549)
(2.835,1.840550)
(2.880,1.856298)
(2.925,1.871802)
(2.970,1.887070)
(3.015,1.902108)
(3.060,1.916923)
(3.105,1.931521)
(3.150,1.945910)
(3.195,1.960095)
(3.240,1.974081)
(3.285,1.987874)
(3.330,2.001480)
(3.375,2.014903)
(3.420,2.028148)
(3.465,2.041220)
(3.510,2.054124)
(3.555,2.066863)
(3.600,2.079442)
(3.645,2.091864)
(3.690,2.104134)
(3.735,2.116256)
(3.780,2.128232)
(3.825,2.140066)
(3.870,2.151762)
(3.915,2.163323)
(3.960,2.174752)
(4.005,2.186051)
(4.050,2.197225)
(4.095,2.208274)
(4.140,2.219203)
(4.185,2.230014)
(4.230,2.240710)
(4.275,2.251292)
(4.320,2.261763)
(4.365,2.272126)
(4.410,2.282382)
(4.455,2.292535)
(4.500,2.302585)
(4.545,2.312535)
(4.590,2.322388)
(4.635,2.332144)
(4.680,2.341806)
(4.725,2.351375)
(4.770,2.360854)
(4.815,2.370244)
(4.860,2.379546)
(4.905,2.388763)
(4.950,2.397895)
(4.995,2.406945)
(5.040,2.415914)
(5.085,2.424803)
(5.130,2.433613)
(5.175,2.442347)
(5.220,2.451005)
(5.265,2.459589)
(5.310,2.468100)
(5.355,2.476538)
(5.400,2.484907)
(5.445,2.493205)
(5.490,2.501436)
(5.535,2.509599)
(5.580,2.517696)
(5.625,2.525729)
(5.670,2.533697)
(5.715,2.541602)
(5.760,2.549445)
(5.805,2.557227)
(5.850,2.564949)
(5.895,2.572612)
(5.940,2.580217)
(5.985,2.587764)
(6.030,2.595255)
(6.075,2.602690)
(6.120,2.610070)
(6.165,2.617396)
(6.210,2.624669)
(6.255,2.631889)
(6.300,2.639057)
(6.345,2.646175)
(6.390,2.653242)
(6.435,2.660260)
(6.480,2.667228)
(6.525,2.674149)
(6.570,2.681022)
(6.615,2.687847)
(6.660,2.694627)
(6.705,2.701361)
(6.750,2.708050)
(6.795,2.714695)
(6.840,2.721295)
(6.885,2.727853)
(6.930,2.734368)
(6.975,2.740840)
(7.020,2.747271)
(7.065,2.753661)
(7.110,2.760010)
(7.155,2.766319)
(7.200,2.772589)
(7.245,2.778819)
(7.290,2.785011)
(7.335,2.791165)
(7.380,2.797281)
(7.425,2.803360)
(7.470,2.809403)
(7.515,2.815409)
(7.560,2.821379)
(7.605,2.827314)
(7.650,2.833213)
(7.695,2.839078)
(7.740,2.844909)
(7.785,2.850707)
(7.830,2.856470)
(7.875,2.862201)
(7.920,2.867899)
(7.965,2.873565)
(8.010,2.879198)
(8.055,2.884801)
(8.100,2.890372)
(8.145,2.895912)
(8.190,2.901422)
(8.235,2.906901)
(8.280,2.912351)
(8.325,2.917771)
(8.370,2.923162)
(8.415,2.928524)
(8.460,2.933857)
(8.505,2.939162)
(8.550,2.944439)
(8.595,2.949688)
(8.640,2.954910)
(8.685,2.960105)
(8.730,2.965273)
(8.775,2.970414)
(8.820,2.975530)
(8.865,2.980619)
(8.910,2.985682)
(8.955,2.990720)
(9.000,2.995732)
};

\end{scope}

\draw (9, -0.6) node {$1$};
\draw (.45, -0.6) node {$\frac{1}{\maxval}=\frac{1}{20}$};

\draw (-0.55, 0) node {$0$};
\draw (-0.55, 2.995732) node {$\ln\maxval$};
\draw (-0.55, 3.995732) node {$1\hspace{-0.1cm}+\hspace{-0.1cm}\ln\maxval$};

\draw (6, 1) node {$\ln(\quant\maxval)$ for $\quant\in[\sfrac{1}{\maxval},1]$};
\draw (5,-1.4) node { Residual surplus curve $\rev_{\topt{\Qud}_1^{\maxval'}}$};

\end{tikzpicture}
\end{minipage}
\begin{minipage}[t]{0.48\textwidth}
\centering
\begin{tikzpicture}[scale = 0.7,pile/.style={->}]

\draw (-0.2,0) -- (9, 0);
\draw [pile] (0, -0.2) -- (0, 4.5);

\draw [very thick] (9, 1.609438) -- (9, 3.995732);
\draw [very thick] (0, 0) -- (1.8, 0);
\draw (9,-0.2) -- (9,0.2);
\draw [dotted] (0,1.609438) -- (9, 1.609438);
\draw [dotted] (0,3.995732) -- (9, 3.995732);

\draw [dashed] (0,0) -- (9, 3.995732);

\fill[black] (0,0) circle (0.1cm);
\fill[black] (1.8,0) circle (0.1cm);
\fill[black] (9,1.609438) circle (0.1cm);
\fill[black] (9,3.995732) circle (0.1cm);

\begin{scope}[very thick]

\draw plot [smooth, tension=0.8] coordinates {
(1.8,0)
(1.89,0.04879)
(1.98,0.09531)
(2.07,0.139762)
(2.16,0.182322)
(2.25,0.223144)
(2.34,0.262364)
(2.43,0.300105)
(2.52,0.336472)
(2.61,0.371564)
(2.7,0.405465)
(2.79,0.438255)
(2.88,0.470004)
(2.97,0.500775)
(3.06,0.530628)
(3.15,0.559616)
(3.24,0.587787)
(3.33,0.615186)
(3.42,0.641854)
(3.51,0.667829)
(3.6,0.693147)
(3.69,0.71784)
(3.78,0.741937)
(3.87,0.765468)
(3.96,0.788457)
(4.05,0.81093)
(4.14,0.832909)
(4.23,0.854415)
(4.32,0.875469)
(4.41,0.896088)
(4.5,0.916291)
(4.59,0.936093)
(4.68,0.955511)
(4.77,0.97456)
(4.86,0.993252)
(4.95,1.011601)
(5.04,1.029619)
(5.13,1.047319)
(5.22,1.064711)
(5.31,1.081805)
(5.4,1.098612)
(5.49,1.115142)
(5.58,1.131402)
(5.67,1.147402)
(5.76,1.163151)
(5.85,1.178655)
(5.94,1.193922)
(6.03,1.20896)
(6.12,1.223775)
(6.21,1.238374)
(6.3,1.252763)
(6.39,1.266948)
(6.48,1.280934)
(6.57,1.294727)
(6.66,1.308333)
(6.75,1.321756)
(6.84,1.335001)
(6.93,1.348073)
(7.02,1.360977)
(7.11,1.373716)
(7.2,1.386294)
(7.29,1.398717)
(7.38,1.410987)
(7.47,1.423108)
(7.56,1.435085)
(7.65,1.446919)
(7.74,1.458615)
(7.83,1.470176)
(7.92,1.481605)
(8.01,1.492904)
(8.1,1.504077)
(8.19,1.515127)
(8.28,1.526056)
(8.37,1.536867)
(8.46,1.547563)
(8.55,1.558145)
(8.64,1.568616)
(8.73,1.578979)
(8.82,1.589235)
(8.91,1.599388)
(9,1.609438)
};

\end{scope}

\draw (9, -0.6) node {$1$};
\draw (1.8, -0.6) node {$\frac{z}{\maxval}=\frac{1}{5}$};

\draw (-0.6, 0) node {$0$};
\draw (-0.5, 1.609438) node {$4\ln\maxval$};
\draw (-0.4, 3.695732) node {$4\hspace{-0.1cm}+\hspace{-0.1cm}4\ln\maxval$};

\draw (4.5, 2) node {$z\ln(\quant\frac{\maxval}{z})$ for $\quant\in[\sfrac{z}{\maxval},1]$};
\draw (5,-1.4) node { Residual surplus curve $\rev_{\topt{\Qud}_4^{\maxval'}}$};

\end{tikzpicture}
\end{minipage}
\end{flushleft}
\caption{
\label{fig:twopieceiron-z-vs-h}
The dashed lines show the ironing of the lottery mechanism.  As $z\rightarrow \maxval$, there is a threshold beyond which the lottery mechanism becomes optional.  Both graphics depict $\maxval = 20$, but note their vertical scales are not equal.  The left side uses $z=1$ for which the lottery is not optimal.  The right side uses $z=4$ for which the lottery is optimal.  Setting $\maxval\geq \RSMINH$ is sufficient to guarantee that at least for the relevant corner case which has $z=1$, the lottery is not optimal.
}
\end{figure}

\paragraph{(2) The optimal region for partial (downard-closed) ironing.}  Now we find the {\em optimal} value $\val\in[1,\maxval]$ to iron all values above it, equivalently, the {\em optimal downward-closed region of quantile space}.  This is a step towards motivating the definition of $\mecha_{o_{\text{pm}}}$ as chosen in \Cref{def:ressurpspecialmech}.

Given the graph of the residual surplus curve, we find this optimal quantile range by considering a line segment with one endpoint as the origin $(\quant = 0, \rev_{\topt{\Qud}_1^{\maxval'}}(0)=0)$, and the other endpoint on the revenue curve at $\rev_{\topt{\Qud}_1^{\maxval'}}(\quant) = \ln (\quant\cdot \maxval))$ for $\quant\in[\sfrac{1}{\maxval},1]$.  We search for the line segment of this type with largest slope.  Directly from ``change in $y$ over change in $x$," the slope function and its derivative are given by:~\footnote{\label{foot:ressurpcalcslopetangent} In fact, by continuity of the derivative of the residual surplus curve in this region, the line segment will be tangent to the residual surplus curve if the optimal quantile is interior, i.e., in $(\sfrac{1}{\maxval},1)$.}
\begin{align}
    \label{eqn:ressurpcalcslope}
    \zeta(\quant) &= \frac{\ln (\quant \cdot \maxval) } {\quant}\quad \text{on}~[\sfrac{1}{\maxval},1]\\
    \nonumber
    \quant^2\cdot \zeta'(\quant) &= \quant\cdot\frac{1}{\quant} -(\ln(\quant\cdot\maxval)\cdot 1 = 1-\ln(\quant\cdot\maxval)
\intertext{such that the derivative shows that the slope function achieves its maximum at $\quant^* = \sfrac{e}{\maxval}$:}
    \label{eqn:ressurpoptslopesmallquant}
    \zeta'(\quant^*)&=\zeta'(\sfrac{e}{\maxval}) = 0
\end{align}
\noindent The optimal range for ironing of small quantiles is $[0,\sfrac{e}{\maxval}]$.  Letting $\zeta^*$ be the optimal slope of the ironed region and recalling $\vv^{\topt{\Qud}_1^{\maxval'}}(\quant) = \sfrac{1}{\quant}$, we have
\begin{equation*}\zeta^*= \zeta(\quant^*) = \zeta(\sfrac{e}{\maxval}) = \frac{\ln\left( \sfrac{e}{\maxval}\cdot{\maxval}\right)}{\sfrac{e}{\maxval}} = \frac{\maxval}{e}
\end{equation*}
\paragraph{(3) Ironed slopes confirm equation~\eqref{eqn:ressurpexampleslopecomparison} is sufficient for \Cref{lem:Moflb22islbversusLOTofopt22}.}  \Cref{lem:dry-15-extension} in \Cref{s:quadsversusuniformsfromg1ofv1timesg2ofv2} shows that residual surplus is proportional to area under an {\em ironed residual surplus curve.}  From the geometry of the ironed curves used respectively by $\mecha_{o_{\text{pm}}}$ and $\text{LOT}_2$, it is clear that the question of which has larger area under the curve reduces to the question of which has the larger slope on the range $[0,\sfrac{e}{\maxval}]$.

Consider comparing (a) the ironed slope $\zeta^*= \sfrac{\maxval}{e}$ just calculated in (3) as used by $\mecha_{o_{\text{pm}}}$; to (b) the slope $\zeta^1 = 1+\ln\maxval$ of the 2-lottery which irons everywhere.  Per the reduction just mentioned, we have $\mechap_{o_{\text{pm}}}> \text{LOT}_2$ if $\sfrac{\maxval}{e}>1+\ln\maxval$ which is equivalent to equation~\eqref{eqn:ressurpexampleslopecomparison}.

\paragraph{(4) The optimal region for partial (downard-closed) ironing.}  To end this section, we show that the optimal set of ironed ranges for $\topt{\Qud}_1^{\maxval'}$ is $\mathcal{Q}^{\&} = \{[0,\sfrac{e}{\maxval}],~ [\quant^{\&} ,1] \}$, with $\quant^{\&}$ identified below, and the optimality of the set $\mathcal{Q}^\&$ self-evident from inspection of the geometry of the residual surplus curve $\bar{\rev}_{\topt{\Qud}_1^{\maxval'}}$.  We do this by finding the quantile $\quant^{\&}$ at which the tangent line intersects the point $(1,1+\ln\maxval)$ in the residual surplus curve graph (i.e., the top right corner point).  The correct quantile $\quant^{\&}$ is the one -- {\em observably independent of $\maxval$} --  that satisfies the equality:
\begin{align}
    \nonumber 
    \ln\left(\quant^{\&}\cdot\maxval\right)+(1-\quant^{\&})\cdot\frac{1}{\quant^{\&}} &= 1 + \ln\maxval\\
    \label{eqn:ressurpoptslopelargequant}
    \Leftrightarrow \quad\quad \ln\quant^{\&}+\frac{1}{\quant^{\&}} &= 2
\end{align}
\noindent which is a unique $\quant^{\&} \in(\sfrac{e}{\RSMINH},0.32]$ (where we chose the lower end point as motivated by $\maxval\geq\RSMINH$).

We end with some notes.  Naturally for $\maxval^{\&}$ representing its threshold value in $[8.55,8.56]$ at which $\text{lb}_{2,2}=\text{opt}_{2,1}$, this critical $\maxval^{\&}$ sets $\quant^*=\sfrac{e}{\maxval^{\&}} = \quant^{\&}$.  Increasing $\maxval$ above $\maxval^{\&}$, we still have that $\quant^{\&}$ is constant but $\quant^*=\sfrac{e}{\maxval}$ is decreasing.  Therefore the optimal ironing leaves the range $[\sfrac{e}{\maxval},\quant^{\&}]$ un-ironed for $\maxval>\maxval^{\&}$.

\section{Supporting Material for \Cref{s:blendsfrominvdist}}
\label{a:generalblendresults}
\label{page:genblendsapp}

This section presents supporting material \Cref{thm:blendsgeneratorfromseparate}.  \Cref{a:blendsgeneratorfromseparate} gives the proof.  \Cref{a:discussblendsfromseparate} analyzes the structure of the order-statistic-separable class of dual blends solutions as identified by \Cref{thm:blendsgeneratorfromseparate}, in particular for understanding the functions $g_1,~g_2$ and $G_1,~G_2$ of the theorem statement.

\subsection{Proof of \Cref{thm:blendsgeneratorfromseparate}}
\label{s:blendsgeneratorfromseparate}
\label{a:blendsgeneratorfromseparate}

\begin{numberedtheorem}{\ref{thm:blendsgeneratorfromseparate}}
\label{thm:blendsgeneratorfromseparate2}
Consider non-negative functions $g_1(\cdot)$ and $g_2(\cdot)$ each with domain $(0,\infty)$.  For every $z>0$, let $g_{1,z}$ be $g_1$ restricted to the domain $[z,\infty)$ and $g_{2,z}$ be $g_2$ restricted to the domain $(0,z]$.

Each $\blendi$ blend is a distribution over the set $\{g_{i,z}~:~z>0\}$.  Let $o_{g_1}(z)$ and $\omega_{g_2}(z)$ be functions ({\em as free parameters which we may design)} to describe weights corresponding respectively to each $g_{1,z}$ and to each $g_{2,z}$.

First, assume $g_1(\cdot)$ and $g_2(\cdot)$ satisfy the following conditions:
\begin{enumerate}
\item The function $\chi(z) = \frac{g_1(z)}{g_2(z)}$ evaluated in the limit at $\infty$ is $0$, i.e., $\lim_{z\rightarrow \infty}\chi(z) = 0$;
\item the function $\psi(z) = \frac{g_2(z)}{g_1(z)}$ evaluated in the limit at $0$ is $0$, i.e., $\lim_{z\rightarrow 0} \psi(z) = 0$;
\item $\chi(z)$ must be weakly decreasing, equivalently, $\psi(z)$ must be weakly increasing;
\end{enumerate}
Then the weights functions $o_{g_1}(z)=d\psi(z)$ and $\omega_{g_2}(z)=-d\chi(z)$ 
give a dual blends solution with: $$g(\vals) = g_1(v_1)\cdot g_2(v_2)~\text{{\em for}}~\vals= (\vali[1],\vali[2]\leq\vali[1])$$

\noindent If the following condition additionally holds:
\begin{enumerate}
\setcounter{enumi}{3}
\item the integrals $G_1(z)=\int_z^{\infty} g_1(y)~dy$ and $G_2(z)=\int_0^z g_2(y)~dy$ are positive and finite for all $x\in(0,\infty)$;
\end{enumerate}
then for the same function $g$, there exists a dual blends solution (by modification from the original solution) for which all of the $g_{1,z}$ and $g_{2,z}$ functions are distributions.  
\end{numberedtheorem}
\begin{proof}
At a high level, the proof is constructive: it is possible to back out weights functions $o_{g_1}(z)$ and $\omega_{g_2}(z)$.  Per the statement, let $\chi(z) = \frac{g_1(z)}{g_2(z)}$.  Choose
\begin{equation}
\label{eqn:omegag2equalsdchi}
\omega_{g_2}(z) = (-1)\cdot d\chi(z)
\end{equation}
such that the upwards-closed integral over all $g_{2,z}$ (where $g_{2,z}(\vali[1])$ is positive) gives\footnote{\label{foot:g2zreducestog1afterintegralendpoints} Note within the sequence of equation~\eqref{eqn:g2ofztog1timesg2} that function $g_{2,z}$ is used in the starting evaluation, where its domain informs the integral end points; after the end points are fixed however, we have $g_{2,z}=g_2$ everywhere. Thus the first step can simplify to the common function $g_2$ and pull multiplicative constants out of the integral.}
\begin{align}
\nonumber
\int_{\vali[1]}^{\infty} \omega_{g_2}(z)\cdot g_{2,z}(\vali[1])\cdot g_{2,z}(\vali[2]) &= \left(g_2(\vali[1])\cdot g_2(\vali[2])\right)\int_{\vali[1]}^{\infty} (-1)\cdot~d\chi(z)\\
\label{eqn:g2ofztog1timesg2}
&= \left(g_2(\vali[1])\cdot g_2(\vali[2])\right) \left[ (-1)\cdot\frac{g_1(z)}{g_2(z)}\right]_{\vali[1]}^{\infty} = g_1(\vali[1])\cdot g_2(\vali[2])
\end{align}
where Condition (1) in the theorem statement is sufficient for the final equality.  Similarly, let $\psi(z) = \frac{g_2(z)}{g_1(z)}$ and choose
\begin{equation}
\label{eqn:og1equalsdpsi}
o_{g_1}(z) = d\psi(z) 
\end{equation}
such that the downward-closed integral over all $g_{1,z}$ (where $g_{1,z}(\vali[2])$ is positive) gives
\begin{align}
\nonumber
\int_0^{\vali[2]} o_{g_1}(z)\cdot g_{1,z}(\vali[1])\cdot g_{1,z}(\vali[2]) &= \left(g_1(\vali[1])\cdot g_1(\vali[2])\right) \int_0^{\vali[2]}d\psi(z)\\
&= \left(g_1(\vali[1])\cdot g_1(\vali[2])\right) \left[\frac{g_2(z)}{g_1(z)}\right]_0^{\vali[2]} = g_1(\vali[1])\cdot g_2(\vali[2])
\end{align}
where Condition (2) is sufficient for the final equality.

By implicit assumption throughout this paper, the weights $o_{g_1}$ and $\omega_{g_2}$ (and the function $g$) must be non-negative everywhere.  Observing weights definitions in equations~\eqref{eqn:omegag2equalsdchi} and~\eqref{eqn:og1equalsdpsi}, Condition (3) is sufficient to meet these high-level assumptions.\footnote{\label{foot:positivityofg} Without our global assumptions on weights here, Condition (3) could be relaxed.}  This completes the proof of the main theorem statement.  To prove the distributions-special-case using Condition (4), we show how to use definitions in the theorem statement to modify the $\omega_{g_2}$-side calculations above (and leave the $o_{g_1}$-side to follow from symmetry, similar to the symmetry above between the two sides).  For this setting, we have a modified blends solution.  Critically, we have $\tilde{g}_{1,z}(x) = \sfrac{g_{1,z}(x)}{G_1(z)}$ and $\tilde{g}_{2,z}(x)=\sfrac{g_{2,z}(x)}{G_2(z)}$.  Condition (4) is sufficient to guarantee that all of the functions $\tilde{g}_{1,z}$ and $\tilde{g}_{2,z}$ are in fact probability distributions.  Choose
\begin{equation}
\label{eqn:weightisdchitimesG2}
\tilde{\omega}_{g_2}(z) = (-1)\cdot d\chi(z)\cdot \left(G_2(z)\right)^2
\end{equation}
which ``corrects for the normalization" within each $\tilde{g}_{i,z}$ by re-factoring the weights, such that the same effective calculation as before goes through.  I.e., the following upward-closed integral gives
\begin{align}
\label{eqn:separationofgiandGi}
\int_{\vali[1]}^{\infty} \tilde{\omega}_{g_2}(z)\cdot \tilde{g}_{2,z}(\vali[1])\cdot \tilde{g}_{2,z}(\vali[2])&=\\
\nonumber
\int_{\vali[1]}^{\infty} \tilde{\omega}_{g_2}(z)\cdot\left(\frac{g_{2,z}(\vali[1])}{G_2(z)}\right)\left(\frac{g_{2,z}(\vali[2])}{G_2(z)}\right) &= \left(g_2(\vali[1])\cdot g_2(\vali[2])\right)\int_{\vali[1]}^{\infty} (-1)\cdot d\chi(z)\\
\nonumber
&= \left(g_2(\vali[1])\cdot g_2(\vali[2])\right) \left[ (-1)\cdot\frac{g_1(z)}{g_2(z)}\right]_{\vali[1]}^{\infty} = g_1(\vali[1])\cdot g_2(\vali[2])
\end{align}
again relying on Conditions (1) and (3).  Condition (2) is sufficient for the $o_{g_1}$-side to work out symmetrically, which uses the modification $\tilde{o}_{g_1}(z) = d\psi(z)\cdot\left(G_1(z) \right)^2$.
\end{proof}

\label{page:quadsversusuniformsfromg1ofv1timesg2ofv2} \noindent Next we illustrate the math of \Cref{thm:blendsgeneratorfromseparate} for our main example of Quadratics-versus-Uniforms in \Cref{s:quadsversusuniformsfromg1ofv1timesg2ofv2}.  We give a second example in \Cref{s:quadsversuscubicsorderstatblends}, for which $G_2(z)=\int_0^z g_2(x)~dx$ evaluates to $\infty$ and therefore the functions $g_{2,z}$ can not possibly be converted to probability distributions by trying to normalize their total weights.

\subsubsection{Quadratics-versus-Uniforms Dual Blend from Order-statistic Separability}
\label{s:quadsversusuniformsfromg1ofv1timesg2ofv2}

We show how our main example of Quadratics-versus-Uniforms fits into \Cref{thm:blendsgeneratorfromseparate}.  We use the distribution-version of the theorem which includes its Condition (4).  Motivated by the Quadratics, let $g_1(x) = \sfrac{1}{x^2}$ inducing $G_1(x) = \sfrac{1}{x}$.  Motivated by the Uniforms, let $g_2(x)=1$ inducing $G_2(x)=x$.  Recall we assume $\vali[1]\geq\vali[2]>0$.  Therefore on the Uniforms side we have the following.  Note that in fact, these calculations apply for arbitrary upward-finite $g_2$ because we can wait until the end to substitute.
\begin{equation*}
\chi(z) = \frac{g_1(z)}{g_2(z)} = \frac{1}{z^2\cdot g_2(z)},\quad\quad d\chi(z) = d\left(\frac{1}{z^2\cdot g_2(z)}\right)= (-1)\cdot\frac{2g_2(z)+z\cdot g'_2(z)}{z^3\cdot(g_2(z))^2}\cdot dz
\end{equation*}
(where the evaluation of $d\chi(z)$ doesn't matter but we write it for completeness).  We further have
\begin{equation*}
\omega_{g_2}(z) = (-1)\cdot d\left(\frac{1}{z^2\cdot g_2(z)}\right)\cdot \left(G_2(z)\right)^2 = \left( \frac{2 g_2(z) + z\cdot g'_2(z)}{z^3\cdot (g_2(z))^2}\right)\cdot(G_2(z))^2\cdot dz\geq 0
\end{equation*}
with the final inequality included to illustrate that it is non-negative.\footnote{\label{foot:omegag2strictweightcheck} Here we can also already confirm that\\
$\omega_{g_2}=\left( \frac{2 g_2(z) + z\cdot g'_2(z)}{z^3\cdot (g_2(z))^2}\right)\cdot(G_2(z))^2\cdot dz = \left( \frac{2\cdot 1 + z\cdot 0}{z^3\cdot (1)^2}\right)\cdot(z)^2\cdot dz=\sfrac{2}{z\cdot dz}$ as it should be, given the example.}   From here we have 
\begin{align*}
&\int_{\vali[1]}^{\infty}\omega_{g_2}\cdot\left(\frac{g_2(\vali[1])}{G_2(z)}\right)\left(\frac{g_2(\vali[2])}{G_2(z)}\right)\\
=& \int_{\vali[1]}^{\infty}(-1)\cdot d\left(\frac{1}{z^2\cdot g_2(z)}\right)\cdot \left(G_2(z)\right)^2 \cdot\left(\frac{g_2(\vali[1])}{G_2(z)}\right)\left(\frac{g_2(\vali[2])}{G_2(z)}\right)\\
=&\left(g_2(\vali[1])\cdot g_2(\vali[2])\right)\cdot \int_{\vali[1]}^{\infty}(-1)\cdot d\left(\frac{1}{z^2\cdot g_2(z)}\right) = \left(g_2(\vali[1])\cdot g_2(\vali[2])\right)\left[\frac{1}{z^2\cdot g_2(z)}\right]_{\infty}^{\vali[1]}= \frac{1}{\vali[1]^2}\cdot g_2(\vali[2])
\end{align*}
as desired, because $g_2(x) = 1$ and $g(\vals) = \sfrac{1}{\vali[1]^2}$ is correct for infinite-weight Quadratics-versus-Uniforms dual blends of \Cref{s:tensorexampleoutline}.  On the Quadratics side, symmetric to the analysis above, we have 
\begin{equation*}
\psi(z) = \frac{g_2(z)}{g_1(z)} = z^2,\quad\quad d\psi(z) = d\left(z^2\right)= 2z\cdot dz
\end{equation*}
\begin{equation*}
\omega_{g_1}(z) = d\left(z^2\right)\cdot \left(G_1(z)\right)^2 =  2z\cdot\left(\frac{1}{z}\right)^2\cdot dz = \frac{2}{z}\cdot dz\geq 0
\end{equation*}
Finally we have
\begin{align*}
&\int_{0}^{\vali[2]}o_{g_1}\cdot\left(\frac{g_1(\vali[1])}{G_1(z)}\right)\left(\frac{g_1(\vali[2])}{G_1(z)}\right) = \int_{0}^{\vali[2]}d\left(z^2\right)\cdot \left(G_1(z)\right)^2\cdot\left(\frac{g_1(\vali[1])}{G_1(z)}\right)\left(\frac{g_1(\vali[2])}{G_1(z)}\right)\\
= & (g_1(\vali[1])\cdot g_2(\vali[2]))\cdot\int_0^{\vali[2]}d(z^2) =(g_1(\vali[1])\cdot g_2(\vali[2]))\left[ ~z^2~\right]_0^{\vali[2]} =\frac{1}{\vali[1]^2}\cdot \frac{1}{\vali[2]^2}\cdot \vali[2]^2 = \frac{1}{\vali[1]^2}
\end{align*}

\subsubsection{Blends from Order-statistic Independence that are not Distributions}
\label{s:quadsversuscubicsorderstatblends}

We give a simple second example which illustrates \Cref{thm:blendsgeneratorfromseparate}.  The dual blends have one side as Quadratics and the other side as {\em Cubics}.  In this case, the Quadratics have downward-closed domain and can not be normalized to distributions because the function $G_2(x)=\int_0^z \sfrac{1}{y^2}~dy=\infty$.  Without further comment, we write down the evaluations of all necessary elements using the definitions of \Cref{thm:blendsgeneratorfromseparate}:

\begin{align*}
    g_1(x) &= \frac{1}{x^3}~\text{for}~ x\in(0,\infty)&g_2(x)&= \frac{1}{x^2}~\text{for}~ x\in(0,\infty)\\
    g_{1,z}(x) &= \frac{1}{x^3}~\text{for}~ x\in[z,\infty)&g_{2,z}(x)&= \frac{1}{x^2}~\text{for}~ x\in(0,z]\\
    \psi(z) &= z~\text{for}~ z\in(0,\infty) & \chi(z) &= \frac{1}{z}~\text{for}~ z\in(0,\infty)\\
    o_{g_1}(z) &= 1\cdot dz~\text{for}~ z\in(0,\infty) & \omega_{g_2}(z) &= \frac{1}{z^2}\cdot dz~\text{for}~ z\in(0,\infty)
\end{align*}
\begin{align*}
        \int_0^{\infty} o_{g_1}(z)\cdot g_{1,z}(\vali[1])\cdot g_{1,z}(\vali[2])&= \int_0^{\vali[2]} o_{g_1}(z)\cdot g_1(\vali[1])\cdot g_1(\vali[2])= \int_0^{\vali[2]}1\cdot \frac{1}{\vali[1]^3}\cdot\frac{1}{\vali[2]^3}~dz\\
        &= \frac{1}{\vali[1]^3}\cdot \frac{1}{\vali[2]^2}= g_1(\vali[1])\cdot g_2(\vali[2])= g(\vals) \\
        = \int_0^{\infty} \omega_{g_2}(z)\cdot g_{2,z}(\vali[1])\cdot g_{2,z}(\vali[2])&= \int_{\vali[1]}^{\infty} \omega_{g_2}(z)\cdot g_2(\vali[1])\cdot g_2(\vali[2])= \int_{\vali[1]}^{\infty}\frac{1}{z^2}\cdot \frac{1}{\vali[1]^2}\cdot \frac{1}{\vali[2]^2}~dz
\end{align*}
\noindent Note -- this solution concept would fail if we assigned the Quadratics to be upward-closed and the Cubics to be downward-closed because the monotonicity conditions of \Cref{thm:blendsgeneratorfromseparate} would be violated.

\subsection{Discussion of \Cref{thm:blendsgeneratorfromseparate}}
\label{a:discussblendsfromseparate}

Through the rest of this section, we discuss a number of intuitive observations regarding the structure of \Cref{thm:blendsgeneratorfromseparate}.

\paragraph{The $g_i$ functions as ``un-normalized" density functions.}  The proof of \Cref{thm:blendsgeneratorfromseparate} makes clear how a function like $g_2(\vali[2])=1$ is the common function representing {\em un-normalized density} of every downward-closed uniform distribution $\Ud_{0,z}$.  I.e., a process to generate any downward-closed uniform distribution is to start with $g_2(\vali[2])=1$ on $[0,\infty)$, restrict it to the domain $[0,z]$, and then divide by the total area under the curve $\int_0^z g_2(y)~dy =z$.  This gives the PDF $\ud_{0,z}(y) = \sfrac{1}{z}$.

Similarly, the function $g_1(\vali[1])$ gives the un-normalized density of every upward-closed quadratic distribution $\Qud_{z}$.  To normalize $g_1$ to become distribution $\Qud_z$, we divide by the tail area $\int_z^{\infty}g_1(y)~dy=\sfrac{1}{z}$ and the resulting PDF is exactly $\qud_z(y)=\sfrac{z}{y^2}$.

\paragraph{Application to distributions requires finite tails.}  In the statement of \Cref{thm:blendsgeneratorfromseparate}, the special case for which we construct $\tilde{g}_{1,z}$ and $\tilde{g}_{2,z}$ to necessarily be probability distributions required additionally Condition (4) which states, ``the integrals $G_1(z)=\int_z^{\infty} g_1(x)dx$ and $G_2(z)=\int_0^z g_2(x)dx$ are positive and finite for all $z$."  This is necessary because, e.g., $\tilde{g}_{1,z}(x)=\sfrac{g_{1,z}(x)}{G_1(z)}$ would otherwise be not well-defined or 0.  See the previous example in \Cref{s:quadsversuscubicsorderstatblends}.

The interpretation of Condition (4) is that $g_1$ must be everywhere ``upward-finite" and $g_2$ must be everywhere ``downward-finite."
\begin{definition}
\label{def:upwardfinite}
Given a non-negative function $g_i(x)$ with domain $(0,\infty)$.  The function $g_i(\cdot)$ is {\em upward-finite} if $\int_z^{\infty}g_i(x)~dx$ is finite for every $z$, and it is {\em downward-finite} if $\int_0^z g_i(x)~dx$ is finite for every $z$.
\end{definition}

\noindent We identify a couple consequences of this structure.  First, it makes permanent the setting of integral end points when calculating density at a fixed input $(\vali[1],\vali[2]\leq\vali[1])$ from each side of the dual blends (recall \Cref{fig:blendsends} in \Cref{a:example}).  Second, it allows us to write any number of simple corollaries to state existence of classes of dual blends that have distributions as elements of the blends, for example:

\begin{corollary}
\label{cor:posnegparetoblendsfromseparate}
Consider parameterized functions $g^{\eta}(x) = \sfrac{1}{x^{\eta}}$ for any $\eta\in\mathbb{R}$.  Setting $g_1=g^{\eta_+}$ for $\eta_+>1$ and $g_2=g^{\eta_-}$ for $\eta_-<1$ will meet all conditions (1) through (4) of \Cref{thm:blendsgeneratorfromseparate}.  Thus, there is a dual blends solution for which the elements are distributions from any $g^{\eta_+}$ and $g^{\eta_-}$.
\end{corollary}

\paragraph{Algebraic consequences of the integral endpoints in the construction.}  The assigned integral end points of the dual blends calculations -- as resulting from \Cref{def:upwardfinite} -- are critical to making the algebra work out.  Specifically, each side observably employs an integral endpoint to ``correct" the $g_i(\vali[j\neq i])$ term which originally appears inside the integral, as a ``constant" given the integration per $dz$.

E.g., the equation in line~\eqref{eqn:separationofgiandGi} at the end of the proof of \Cref{thm:blendsgeneratorfromseparate} makes this clear: both $g_i(\vali)$ and $g_i(\vali[j])$ terms get pulled out.  After this step, the evaluation of the integral {\em given its endpoints} is needed to both construct $g_i(\vali)$ and cancel $g_i(\vali[j])$ -- there are no other algebraic tools available to construct the function $g$.  In fact, we can't change $g_i(\vali)$ and it passes intact as a factor of $g$.  Evaluation of the integral must replace the $g_i(\vali[j])$ term with a $g_j(\vali[j])$ term which is the second factor of $g$.  Then the weights terms are designed to get the overall integrand correct so that the anti-derivative function evaluates the ``extreme" end point (at 0 or $\infty$) to 0 and the other end point at $\vali[j]$ to convert an original $g_i(\vali[j])$ term to $g_i(\vali)$ as needed within the order-statistic-separable function $g$.

With this algebraic set up in mind, it should now be clear why we should not expect a direct extension of The Blends Technique (or general dual blends solutions) for $n\geq 3$.  For example, consider trying to construct dual blends for the function 
\begin{equation*}g(\vali[1],\vali[2],\vali[3])=g_1(\vali[1])\cdot g_2(\vali[2])\cdot g_3(\vali[3])
\end{equation*}
by direct analogy to the $n=2$ case.  The problem for generalization is that the design for $n=2$ gives each side exactly two ``degrees of freedom" to set $g_1$ and $g_2$.  To attempt the same design for $n=3$, let $i,j$ be distinct elements of the set $\{1,2,3\}$.  Each side of the (supposed) dual blend must be symmetric from a functional starting point:
\begin{align*}
    g_i(\vali[1])\cdot g_i(\vali[2])\cdot g_i(\vali[3]) \int_0^b\left( \cdot \right)dz &= g(\vals) = g_1(\vali[1])\cdot g_2(\vali[2])\cdot g_3(\vali[3])\\
    &= g_j(\vali[1])\cdot g_j(\vali[2])\cdot g_j(\vali[3]) \int_a^{\infty}\left( \cdot \right)dz 
\end{align*}
but there is no way to evaluate the integrals -- no matter what their integrands are or what their endpoints are -- to combine with each of $\prod\nolimits_k g_i(\vali[k])$ and $\prod\nolimits_k g_j(\vali[k])$ to get $g_1(\vali[1])\cdot g_2(\vali[2])\cdot g_3(\vali[3])$.  The only solution is $g_i=g_j$.

\paragraph{The $G_i(\cdot)$ functions as continuous scalars.}  By inspection of equation~\eqref{eqn:separationofgiandGi}, the (finite) functions $G_i(x)$ can in fact be set to any function that is strictly positive (or even more generally, non-zero) {\em as long as they are still offset by $G_i(\cdot)$ terms in the weights functions}. Therefore, the $g_i(\cdot)$ functions only need to be subject to the restrictions on $\chi(\cdot)$ and $\psi(\cdot)$ for there to exist a blend $g(\vals) = g_1(\vali[1])\cdot g_2(\vali[2])$.

\begin{corollary}
\label{cor:arbitraryG1andG2asscalars}
Consider non-negative functions $g_1(\cdot)$ and $g_2(\cdot)$ each with domain $(0,\infty)$.  Let $g_{1,z}$ be $g_1$ restricted to the domain $[z,\infty)$ and $g_{2,z}$ be $g_2$ restricted to the domain $(0,z]$.  Assume there exists a dual blends solution $$g(\vals) = g_1(v_1)\cdot g_2(v_2)~\text{for}~\vals= (\vali[1],\vali[2]\leq\vali[1])$$ using weights $o_{g_1}(z)$ and $\omega_{g_2}(z)$.

Then for any finite, positive functions $G_1,~G_2$, the weights $\tilde{o}_{g_1}(z) = o_{g_1}(z)\cdot G_1(z)$ and $\tilde{\omega}_{g_2}(z) = \omega_{g_2}(z)\cdot G_2(z)$ applied to functions $\sfrac{g_{i,z}}{G_i(z)}$ describes the same dual blends solution from $g_1,~g_2$.
\end{corollary}

\noindent There exists a comparison here to $n$th-order tensors.  
The specific observation here is that this \Cref{cor:arbitraryG1andG2asscalars} is analogous to dividing a symmetric tensor's element-vector $\boldsymbol{a}_z$ by a factor $\kappa_z$ and multiplying its scalar $l_z$ by $\kappa_z^n$.  In the same way that we can multiply-and-divide by the respective $G_i$ with no effect on $g$, these multiplicative factors cancel and have no effect on $n$th-order tensor  $T=l_z\cdot\left(\boldsymbol{a}_z\otimes\ldots\otimes\boldsymbol{a}_z \right)$.

\section{Deferred Full Presentation of \Cref{s:blendsisidd}}
\label{a:info_design}
\label{page:blendsisiddapp}

To restate the main goal of this section from \Cref{s:blendsisidd}: for the problem of information-design-design (from equation~\eqref{eqn:blendsisidd}), we want to set up and give the proof for \Cref{prop:blendsisidd}:

\begin{numberedprop}{\ref{prop:blendsisidd}}
\label{prop:blendsisidd2}
Consider the prior independent design problem (\Cref{def:pidesign}) given a class of distributions $\scF$, a class of algorithms $\algspace$, and $n$ inputs.  Optimization of the Blends Technique approach to prior independent lower bounds is described by:
\begin{equation*}
    \piratio^{\scF} \geq \sup_{g\in\mathcal{G}}~\left[\frac
    { \sup_{\blendi[2]\in\{\blend~|~\blend\in\Delta(\scF)~\text{and}~\blend^n=g\}}\left( \expecta_{\dist\sim\blendi[2]}\left[\OPT_{\dist}(\dist) \right]\right) }
    { \inf_{\blendi[1]\in\{\blend~|~\blend\in\Delta(\scF^{\text{all}})~\text{and}~\blend^n=g\}}\left(\expecta_{\dist\sim\blendi[1]}\left[\OPT_{\dist}(\dist) \right] \right)}
    \right]
\end{equation*}
\noindent Further, its Numerator Game and its Denominator Game can be independently instantiated as problems of constrained information design.
\end{numberedprop}

\noindent An outline for this section is: \Cref{s:intro_info_design} gives an introduction to information design.  \Cref{s:iddreduction} gives an intuitive explanation of the reduction of the Numerator and Denominator Games within equation \eqref{eqn:blendsisidd} to information design; it includes \Cref{lem:iolemma} which shows that the crux of the reduction is a straightforward application of Bayes Law.  \Cref{s:blackwell} introduces Blackwell ordering and observes that our dual blends in the Quadratics-versus-Uniforms example of \Cref{s:tensorexampleoutline} do not have a Blackwell ordering.

\paragraph{Related Work for Information Design}
The canonical model of information design with a single sender and
single receiver was introduced by \citet{RS-09} and \citet{KG-11}.  A
few points of context with this literature are as follows.  In our
setting the allowable posterior distributions are constrained.  The
early work of \citet{GR-04} -- in which the sender can only present certain
kinds of evidence -- can be viewed as a posterior-constrained setting of
information design.  In our mechanism design applications, the receiver
is a seller and faces a number of potential buyers.  \citet{BBM-15}
previously studied information design in such a scenario with only one buyer, with the goal
of characterizing the feasible outcomes that a regulator (the sender)
can obtain in terms of the tradeoff between revenue and residual surplus.  While it is
not directly related to the methods of this paper, there is a
literature starting with \citet{DNPW-19} that shows that some problems
of information design are computationally tractable.  See
\citet{BM-19} for a more complete survey of the breadth of literature
on information design.


\subsection{Introduction to Information Design}
\label{s:intro_info_design}

From economics, information design is a game between two players -- a Sender and a Receiver -- who have unaligned objective functions.  There is an unknown state of the world $\theta$ from a set of states $\Theta$.  Realized state $\hat{\theta}$ is Bayesian and is drawn from a prior $\bar{\pi}$ that is common knowledge.

The Sender observes $\hat{\theta}$ and sends a signal $\mes$ from the signal space $\messpace$.  This is implemented by: up front, the Sender commits to a signalling strategy $\sigma:\Theta\rightarrow \Delta(\messpace)$ that maps states to distributions over signals, with $\sigma\in\Sigma$ the space of (possibly restricted) strategies.\footnote{\label{foot:infostructmapVScorrelated} In the context of a fixed prior, there is a bijection between signalling strategies and information structures as we define them.  The economics literature may use the term ``information structure" for our signalling strategies.}  Strategies $\sigma$ implicitly lead to {\em information structures} because of the existence of the prior $\bar{\pi}$ -- information structures describe the ex ante correlated distribution over paired state-and-signal.  The takeaway is that information structures represent strategic design by the Sender to convert the prior $\bar{\pi}$ into a structured system of posteriors (conditional for each $\mes$) for specific use by the Receiver.
\begin{definition}
\label{def:infostructure}
An {\em information structure} $\info:\Theta\times\messpace\rightarrow[0,1]$ is a correlated probability distribution over state and signal.
\end{definition}
\noindent We make two critical observations: an information structure is induced from a given prior $\bar{\pi}$ over state and a signalling strategy $\sigma$; and in turn, an information structure induces a posterior distribution over states (conditioned on a realized output signal $\hat{\mes}$).

After the Sender commits to $\sigma$, the Sender observes $\hat{\theta}$ and sends a signal $\hat{\mes}$ to the Receiver as randomly drawn from $\sigma(\hat{\theta})$.  The Receiver sees $\hat{\mes}$ and chooses an action $\omega$ from its action space $\Omega$.  Finally, each player has utility functions respectively as $S:\Theta \times \Omega\rightarrow \reals$ and $R: \Theta \times \Omega\rightarrow \reals$.

It is standard to assume that the Receiver plays a best-response action: given the context of knowing $\bar{\pi}$ and $\sigma$, it uses $\hat{\mes}$ to get a posterior distribution over state and then simply optimizes against the posterior.  This leaves the Sender's construction of $\sigma$ as the unique strategic consideration, called {\em information design}.  The utility functions $S$ and $R$ typically embed a degree of objectives being orthogonal -- or adversarial.  If the utility functions are aligned (which will be true for one of our cases), information design is trivial unless the Sender's signal space is restricted to not be able to fully reveal the realized state.

\subsection{Reduction of Blends Technique Sub-problems to Information Design}
\label{s:iddreduction}

This section explains the reduction from the Numerator and Denominator Game sub-problems within the reorganized Blends Technique in equation~\eqref{eqn:blendsisidd}, to constrained information design.  I.e., this section proves \Cref{prop:blendsisidd}.

The key element of the reduction is to carefully constrain the Sender's space of signalling strategies to {\em blends-revelation signalling strategies}, defined as follows.  Effectively, we implement a {\em Revelation Principle for information design} which states that the Sender's signal may as well be a \underline{correct} posterior over state space -- {\em which we further require to be a symmetric product distribution} -- so that the Receiver only needs to best respond to the posterior-signal.

\begin{definition}
\label{def:truthblendssigs}
Within information design, we define a {\em blends-revelation signalling strategy} (BRSS) to be a signalling strategy $\sigma^{\text{{\em br}}}$ in which:
\begin{itemize}
    \item signals are distributions;
    \item the marginal distribution over signal-distributions resulting from $\sigma^{\text{{\em br}}}$ as a blend induces $g$;
    \item the Receiver's posterier given any signal-distribution $\dist$ is in fact $\dist^n$.
\end{itemize}
\end{definition}

\noindent Fixing a prior independent design problem (PIP), the instantiation of its Numerator and Denominator Games (equation~\eqref{eqn:blendsisidd}, \Cref{def:blendsisidd}) as information design problems is now from the following reduction.  The reductions are the same with the exception of the Sender's objective function (described in the last point).

\begin{itemize}
    \item $\Theta = \valspace^n$; state space is the input space of the prior independent algorithm;
    \item $\bar{\pi} = g$; the prior over states is equal to the correlated distribution $g$ (for any $g$ as fixed by the outer program in equation~\eqref{eqn:blendsisidd});
    \item $\messpace = \scF$; signal space is restricted to be the PIP's allowable class of distributions $\scF$;
    \item $\Sigma = \{\sigma^{\text{br}}~|~\sigma^{\text{br}}~\text{is a BRSS}\}$; \textbf{the key element of the reduction:} signalling strategy space $\Sigma$ is the set of blends-revelation signalling strategies of \Cref{def:truthblendssigs}; note that at least one such signalling strategy must exist because $\bar{\pi}=g$ was constructed up front from a blend and can in fact be implemented (see \Cref{lem:iolemma} below);
    \item $\Omega = \algspace$; the Receiver's action space is naturally the algorithm space $\algspace$ from the PIP, and
    \item $R(\vals, \algo) = \algo(\vals)$; the Receiver's utility is equal to the objective of the algorithm designer in the PIP;
    \item the Sender's utility $S$ is either perfectly aligned with the Receiver's objective (Numerator Game) or perfectly adversarial to it (Denominator Game):
    \begin{itemize}
        \item Numerator Game: $S(\vals, \algo) = \algo(\vals)= R(\vals,\algo)$;
        \item Denominator Game: $S(\vals, \algo) = -\algo(\vals)= -R(\vals,\algo)$;
    \end{itemize}
    but note how in both cases, because the Sender's signal is always a {\em correct posterior} (per \Cref{def:truthblendssigs}), it must be that the marginal distribution over posteriors has exactly the structure of a blend.
\end{itemize}

\noindent The challenge for the Sender is how to produce and optimize strategies that meet \Cref{def:truthblendssigs}.  From application of Bayes Law, it turns out that the Sender is able to design a signalling strategy in advance that, ex post observing state, simulates a random latent variable distribution to provide as signal to the Receiver, in a way that the Receiver will use the distribution-signal as if it is correct.  

The Sender chooses a signalling strategy $\sigma$ using the following outline.  Similar to the PIP's adversary, the Sender optimizes over $\blend\in\Delta(\scF)$ such that $\blend^n=g$.  It uses observed state $\vals$ to do Bayesian updating on the distribution $\blend$ (over distributions $\dist\in\scF$) and then randomly draws $\hat{\dist}$ from the posterior to send as the signal.  

As the final key piece which we state next and prove, this Sender's choice of randomized $\sigma$ yields a ``correct" posterior (for every realized $\hat{\dist}$ as signal), thereby satisfying the last requirement of \Cref{def:truthblendssigs}.  The proof makes clear that {lem:iolemma} is an application of Bayes Law.

\begin{lemma}
\label{lem:iolemma}
Given state space $\Theta$ equal to algorithm input space $\valspace^n$ and prior $\bar{\pi}$ over stats as a blend $g = \blend\in\Delta(\scF)$.  Given realized $\vals\sim g$, let the Sender's signalling strategy draw distribution-signal $\hat{\dist}$ from the posterior of $\blend$ (given $\vals$).  Then the Receiver's induced posterior over state is $\hat{\dist}^n$ and the distribution over induced posteriors is $\blend$.
\end{lemma}
\begin{proof}
Given correlated $g$, it is equivalent to assume that inputs were drawn from a two-step procedure: first draw $\dist\sim\blend\in\Delta(\scF)$ and then draw $n$ inputs i.i.d.\ from $\dist$.  Consider from this perspective that density in the original function $g$ is further broken down for each input to reflect density of its latent variable $\dist$, i.e., consider correlated density $g^+$ over $\valspace^n\times\scF$.  (Note, we can recover the function $g$ by fixing each $\vals$ and integrating over $\scF$.)

Given the definition of the Designer's signalling strategy, the correlated distribution over (state, signal) is exactly equal to $g^+$.  The reason is that given $g^+$ the final correlated description over $\valspace^n \times \scF$, for every $(\vals,\dist)$, Bayes Law states that the following quantities are equal: $\text{Pr}[\vals~|~\dist]\cdot \text{Pr}[\dist] = \text{Pr}[\dist~|~\vals]\cdot \text{Pr}[\vals]$.  So we note the following equivalence when applied to our problem.
\begin{itemize}
    \item The left-hand side of our Bayes-Law-equation gives an unfalsifiable description of how inputs were generated 
    (see the first sentence of proof).
    \item The right-hand side describes how (state, signal) pairs are generated within the Information Design game: first there is a randomly revealed state; and second, per the pre-committed signalling strategy, there is a random mapping from state to signals using a posterior distribution (from updating $\blend$ given $\vals$).
\end{itemize}

\noindent The final point is: receiving a specific signal $\hat{\dist}$, the Receiver's posterior over $\valspace^n$ is obtained from conditioning $g^+$ given $\hat{\dist}$ and then the Receiver's posterior is exactly $\hat{\dist}^n$ as desired.  Note for this last point, the Receiver has access to $g^+$ because the Designer's choice of $\blend$ is known.
\end{proof}

\subsection{Assessment of Blends' Blackwell Ordering}
\label{s:blackwellgarbles}
\label{s:horizons}
\label{s:blackwell}

This section considers if dual blends have the property that one side of the dual blend is ``strictly more informative" than the other side (per \citet{bla-53}), as part of understanding what is driving prior independent lower bounds that follow from dual blends.  It introduces Blackwell (partial) ordering and shows that the two sides of a dual blend do not generally have a Blackwell ordering: for our main example of (\Cref{s:tensorexampleoutline}), there is no informational relationship.

\subsubsection{Blackwell Ordering: Global Usefulness and Blackwell ``Garbles"}
\label{s:blackwelldefsintro}

\citet{bla-53} proposed a framework of {\em partial ordering} between the distributions over signals -- and their respective systems of posteriors -- of two information structures $\info_1$ and $\info_2$ to reflect a notion of information-dominance called Blackwell ordering.  It is based on two strong properties being equivalent.  According to one description, there is an ordered relationship between the information structures based on {\em global usefulness}, i.e., if one information structure $\info_1$ is preferred to $\info_2$ {\em for every possible} utility function (by an optimizer using a random signal).  The equivalent descriptive property is called a {\em garbling} and it applies when the signals of $\info_2$ can themselves be interpreted as obfuscating mixtures over the signal probabilities of $\info_1$ while maintaining exactly the same induced prior over state.

The next two definitions are presented within the context of the problem statement of information design: we are given an information design problem with fixed prior $\bar{\pi}$ over state space $\Theta$, and two feasible information structures $\info_1:\Theta\times\messpace_1\rightarrow[0,1]$ and $\info_2:\Theta\times\messpace_2\rightarrow[0,1]$ for respective signal spaces $\messpace_1$ and $\messpace_2$, and a class of allowable Receiver-algorithms $\algspace$.
%
\begin{fact}
\label{fact:signalstrategyfromstateandinfostruct}
Given $\bar{\pi}$ and an information structure $\info_2$, the signalling strategy $\sigma_{\info_2}$ (which induces $\info_2$ starting from $\bar{\pi}$) can necessarily be reverse-engineered.
\end{fact}
\begin{definition}[\citealp{bla-53}]
\label{def:globaluseful}
Let $\sigma_{\info_1},~\sigma_{\info_2}$ be the signalling strategies induced by $\bar{\pi}$ and the respective information structures (per \Cref{fact:signalstrategyfromstateandinfostruct}).  Let $\algo_1^*,~\algo_2^*$ be optimal algorithms given respective information structures.  Information structure $\info_1$ has greater {\em global usefulness} than $\info_2$ if \underline{for every} (Borel-measurable) Receiver's utility function $R$, expected \underline{optimal} utility is weakly greater when signals are drawn given $\sigma_{\info_1}$ compared to signals drawn given $\sigma_{\info_2}$ (i.e., if $\info_1$ is {\em preferred} to $\info_2$):
\begin{equation}
    \label{eqn:globalusefuldef}
    \mathbf{E}_{\hat{\theta}\sim\bar{\pi},~\hat{\mes}\sim\sigma_{\info_1}(\hat{\theta})}\left[R(\hat{\theta},\algo_1^*(\hat{\mes}))\right]
    \geq \mathbf{E}_{\hat{\theta}\sim\bar{\pi},~\hat{\mes}\sim\sigma_{\info_2}(\hat{\theta})}\left[R(\hat{\theta},\algo_2^*(\hat{\mes}))\right]
\end{equation}
\end{definition}
%
%
\begin{definition}[\citealp{bla-53}]
\label{def:garble}
The information structure $\info_2$ is a {\em garble} of $\info_1$ if there exists a mapping $\eta: \messpace_1\times\messpace_2  \rightarrow [0,1]$ such that $\int_{\messpace_2}\eta(\mes_i,\mes_j)~d\mes_j=1$ for all $\mes_i\in\messpace_1$; 
and for every $\mes_j\in\messpace_2$ and every state $\theta\in\Theta$ we have
\begin{equation}
    \label{eqn:garbledef}
    \info_2(\theta,\mes_j) = \int_{\messpace_1}\eta(\mes_i,\mes_j)\cdot \info_1(\theta,\mes_i)~d\mes_i
\end{equation}
\noindent i.e., we identify $\info_1$ as being {\em more informative (per garbling order)} in comparison to $\info_2$.\footnote{\label{foot:altgarbledef} An intuitive explanation of garbles is: each signal $\mes_j\in\messpace_2$ can be interpreted as a distribution over the signals of $s_i\in\messpace_1$.  In response to each signal $\mes_j\in\messpace_2$, we respond with a single optimal algorithm for the posterior given $\mes_j$, which may not be optimal given each signal $\mes_i\in\messpace_1$ in the implicit distribution; hence, $\info_2$ has ``garbled" $\info_1$.}
\end{definition}
%
%
%
\noindent We now give Blackwell's classic theorem which states equivalence of \Cref{def:globaluseful} and \Cref{def:garble}.
\begin{theorem}[\citealp{bla-53}]
\label{thm:blackwelleqv}
An information structure $\info_1$ has greater global usefulness than $\info_2$ {\em if and only if} $\info_2$ is a garble of $\info_1$.

As contrapositive (in one direction), if there exist two distinct utility functions $R_1$ and $R_2$ such that $\info_1$ is {\em strictly} preferred to $\info_2$ for $R_1$ but $\info_2$ is {\em strictly} preferred to $\info_1$ for $R_2$, then there can not be a garbling order relationship between the information structures.
\end{theorem}
\noindent The ``strictly different" preferences of distinct information structures given distinct utility functions is necessary to apply the contrapositive.  Motivated by \Cref{thm:blackwelleqv}, the common ordering from global usefulness and garbles is called {\em Blackwell ordering}.
%

\subsubsection{Dual Blends Do Not Generally Have Blackwell Ordering}
\label{s:blendsnotgarbles}

Dual blends are represented by simple information structures when their common correlated distribution $g$ (over inputs in $\valspace^n$) is interpreted as the prior $\bar{\pi}$ over state and when signalling strategies are designed as in \Cref{s:iddreduction} using \Cref{lem:iolemma}.

When message space $\messpace$ is set equal to distribution class $\scF$ (as support for elements of the blends), it is clear that blends properly define an information structure (of \Cref{def:infostructure}) as a distribution over paired message-and-state.  For example, describe a blend by $\info_2(\dist_z,\vals)=o_z\cdot \prod_k \distp_z(\vali[k])$.  These are ``simple" because the blend already describes posteriors which are independent given a signal $\dist_z$ as: $n$ i.i.d.\ draws from $\dist_z$.

Describing blends as information structures aligns exactly with the calculations of lower bounds in the Blends Technique.  Recall for blends $\delta_1^n,~\delta_2^n$, expected ``optimal performance" within an algorithm setting is given by $\text{opt}_{n,i}= \mathbf{E}_{\dist\sim\delta_i^n,~\vals\sim \dist}\left[\text{OPT}_{\dist}(\vals) \right]$, and then a lower bound is given by $\sfrac{\text{opt}_{n,2}}{\text{opt}_{n,1}}$.  Each blend is a possible information structure to represent the same underlying correlated distribution over states, and for each blend the quantity $\text{opt}_{n,i}$ is the optimal performance in expectation over state, as the algorithm knows the realized distribution-signal.

We are ready to state by counterexample that dual blends do not generally have Blackwell ordering, using the Quadratics-versus-Uniforms example of \Cref{s:tensorexampleoutline} and \Cref{def:infostructure}).  Recall \Cref{thm:finitequadsversusunifsrevpiauction} for revenue auctions used an adversary choosing the benchmark from the Uniforms side of the dual blends but \Cref{thm:finitequadsversusunifsressurppiauction2} for residual surplus used the Quadratics side, and that the settings have distinct objective functions.

 \Cref{thm:finitequadsversusunifsrevpiauction} and \Cref{thm:finitequadsversusunifsressurppiauction2} each show strict performance gaps for their respective settings.  Thus, these results give an immediate example meeting the condition of the contrapositive statement in \Cref{thm:blackwelleqv}, because a ``Receiver" strictly prefers distinct information structures depending on the auction objective.

\begin{numberedcorollary}{\ref{thm:blendsnotblackwell}}
\label{thm:blendsnotblackwell2}
Finite-weight Quadratics-versus-Uniforms dual blends 
are an example for which there is no relationship according to Blackwell ordering.
\end{numberedcorollary}

\iftensapp
\input{zconf_tensor.tex}
\fi

\ifbmapp
\input{zconf_benchmark.tex}
\fi

\ifextsapp
\input{zconf_extensions}
\fi

\end{appendix}

\end{document}